\documentclass{article}

\usepackage{graphicx}       
\graphicspath{ {./} }       
\usepackage{amsmath}        
\usepackage{amssymb}        
\usepackage[utf8]{inputenc} 
\usepackage[T1]{fontenc}    
\usepackage{url}            
\usepackage{booktabs}       
\usepackage{amsfonts}       
\usepackage{nicefrac}       
\usepackage{microtype}      
\usepackage{todonotes}      
\usepackage{algorithm}      
\usepackage{algorithmic}    
\usepackage{hyperref}       
\usepackage{enumitem}       
\usepackage{mathtools}      
\usepackage{natbib}         
\usepackage{multirow}
\setcitestyle{authoryear}
\usepackage{dsfont}
\usepackage{caption}

\usepackage{geometry}
\usepackage{amsthm}
\newtheorem{theorem}{Theorem}
\newtheorem{lemma}{Lemma}

\newtheorem{assumption}{Assumption}
\newtheorem{proposition}{Proposition}
\newtheorem{remark}{Remark}[section]

\renewcommand{\cite}[1]{\citep{#1}}

\renewcommand{\cite}[1]{\citep{#1}}
\newcommand{\Indicator}[1]{\mathds{1}\LRl{#1}}
\def\SCOP{\operatorname{SCOP}}

\def\FCR{\operatorname{FCR}}
\def\FDR{\operatorname{FDR}}
\def\FCP{\operatorname{FCP}}
\def\PI{\operatorname{PI}}
\def\hmu{\hat{\mu}}
\def\htau{\hat{\tau}}
\def\hkappa{\hat{\kappa}}
\def\hgS{\hat{\gS}}

\def\unif{\operatorname{Unif}}
\def\rank{\operatorname{Rank}}

\newcommand{\LRs}[1]{\left(#1\right)}
\newcommand{\LRm}[1]{\left[#1\right]}
\newcommand{\LRl}[1]{\left\{#1\right\}}

\def\gD{{\mathcal{D}}}

\def\gM{{\mathcal{M}}}
\def\gS{{\mathcal{S}}}
\def\gU{{\mathcal{U}}}
\def\gC{{\mathcal{C}}}
\def\gE{{\mathcal{E}}}
\def\gN{{\mathcal{N}}}

\def\sP{\mathbb{P}}

\newcommand{\E}{\mathbb{E}}

\def\ermR{{\textnormal{R}}}
\def\ermT{{\textnormal{T}}}

\newcommand{\LRabs}[1]{\left|#1\right|}
\def\sR{{\mathbb{R}}}

\def\rvx{{\mathbf{x}}}
\newcommand{\Eqmark}[2]{\stackrel{(#1)}{#2}}

\def\sP{{\mathbb{P}}}

\title{Selective Conformal Inference with False Coverage-statement Rate Control}

\author{
Yajie Bao$^a$, Yuyang Huo$^b$, Haojie Ren$^a$ and Changliang Zou$^b$\footnote{Corresponding Author: nk.chlzou@gmail.com}\\
$^a$School of Mathematical Sciences,  Shanghai Jiao Tong University \\ Shanghai, P.R. China \\
$^b$School of Statistics and Data Science, Nankai University\\
 Tianjin, P.R. China
}

\begin{document}
\maketitle
\begin{abstract}
Conformal inference is a popular tool for constructing prediction intervals. We consider here the scenario of post-selection/selective conformal inference, that is prediction intervals are reported only for individuals selected from unlabeled test data. To account for multiplicity, we develop a general split conformal framework to construct selective prediction intervals with the false coverage-statement rate control. We first investigate the \citet{benjamini2005false}'s false coverage rate-adjusted method in the present setting, and show that it is able to achieve false coverage-statement rate control but yields uniformly inflated prediction intervals. We then propose a novel solution to the problem called {\it selective conditional conformal prediction}. Our method performs selection procedures on both the calibration set and test set, and then constructs conformal prediction intervals for the selected test candidates with the aid of conditional empirical distribution obtained by the post-selection calibration set. When the selection rule is exchangeable, we show that our proposed method can exactly control the false coverage-statement rate in a model-free and distribution-free guarantee. For non-exchangeable selection procedures involving the calibration set, we provide non-asymptotic bounds for the false coverage-statement rate under mild distributional assumptions. Numerical results confirm the effectiveness and robustness of our method in false coverage-statement rate control and show that it achieves more narrowed prediction intervals over existing methods across various settings.
\medskip

\noindent {\it Keywords}:  Conditional empirical distribution; Distribution-free;  Non-exchangeable conditions; Post-selection inference; Prediction intervals; Split conformal.

\end{abstract}


\section{Introduction}\label{sec:introdution}

To improve the prediction performance in modern data, many sophisticated machine learning algorithms including various ``black-box'' models are proposed. While often witnessing empirical success, quantifying prediction uncertainty is one of the major issues for interpretable machine learning. 
Conformal inference \citep{vovk1999machine,vovk2005algorithmic} provides a powerful and flexible tool to quantify the uncertainty of predictions. With the extensive availability of big data, making predictive inference on all available unlabeled data is either unnecessary or inefficient in many applications. For example, in the recruitment decisions, only some selected viable candidates can get into interview processes \citep{shehu2016adaptive}. In the drug discovery trials, researchers select promising drug-target pairs based on predicting candidates’ activity for further clinical trials \citep{dara2021machine}. In such problems, the most common way is to select a subset of individuals with some rules through some statistical/machine learning algorithms at first, and then perform statistical inference only on the selected samples.

Consider a typical setting in which we observe a labeled data set $\gD_{l} = \{(X_i,Y_i) \in \sR^d \times \sR\}_{i=1}^{2n}$ and a test set $\gD_u=\{X_i\}_{i=2n+1}^{2n+m}$ whose outcomes $\{Y_i\}_{i=2n+1}^{2n+m}$ are unobserved. Suppose a selection rule is performed on the test set, and outputs a selected test subset $\hat{\gS}_{u}\subseteq\{2n+1,\ldots,2n+m\}$. We aim to construct the prediction interval of the unknown label $Y_j$ for each candidate $j\in\hat{\gS}_{u}$. In this work, we focus on the split conformal inference framework \citep{papadopoulos2002inductive,lei2018distribution}, which is a distribution-free inferential tool to provide valid uncertainty quantification for complex machine learning models. The labeled data $\gD_l$ is split into two disjoint parts of size $n$, the training set $\gD_t$ for fitting a prediction model $\hat{\mu}(X)$, and the calibration set $\gD_c$ for constructing a prediction interval for the label $Y_j$ around $\hat{\mu}(X_j)$. Let $\gC\subset\{1,\cdots,2n\}$ be the index set of calibration data and denote the residuals as $R_i=|Y_i-\hat{\mu}(X_i)|$ for ${i\in \gC}$. For each test point $j$, the corresponding $(1 - \alpha)$ marginal conformal prediction interval is
\begin{align}\label{eq:marginal_PI}
    \PI^{\text{M}}_j = \hat{\mu}(X_{j}) \pm Q_{\alpha}\LRs{\{R_i\}_{i\in \gC}},
\end{align}
where $Q_{\alpha}(\{R_i\}_{i\in \gC})$ denotes  $\lceil (1 - \alpha)(|\gC|+1)\rceil \text{-th smallest value in } \{R_i\}_{i\in \gC}.$ If the data points in calibration set and test set are independent and identically distributed, then $\PI^{\text{M}}_j$ enjoys the marginal coverage guarantee: $\sP(Y_{j} \in  \PI^{\text{M}}_j) \geq 1-\alpha$.

However, as pointed out by \citet{benjamini2005false}, ignoring the multiplicity in the construction of selected parameters' confidence intervals will result in a distorted average coverage. Similar issues also appear in post-selection predictive inference. Let $\{\PI_j\}_{j\in \hgS_u}$ be the prediction intervals reported by some predictive inference method. The proportion of selected labels not covered by respective $\PI_j$'s can be much larger than the marginal nominal confidence level. {In our real data application of house price prediction, at the nominal level $90\%$, the average coverage ratio of selected prediction intervals is only $61.51\%$,
while our proposed method achieves $91.43\%$ average coverage ratio.}
Without proper adjustment, we cannot obtain reliable intervals for selected test candidates. 
\citet{benjamini2005false} pioneered the criterion, false coverage-statement rate, $\FCR$ to take into account multiplicity in the confidence intervals of multiple selected parameters. The concept of $\FCR$ can readily be adapted to the present conformal prediction setting. We define it as the expected ratio of the number of reported prediction intervals failing to cover their respective true outcomes to the total number of reported prediction intervals, say
\begin{equation}\label{eq:FCR_def} 
    \FCR = \E\LRm{\frac{|\{j\in \hat{\gS}_{u}: \ Y_j \not\in \PI_j\}|}{\max\{|\hat{\gS}_u|, 1\}}}.
\end{equation}
\citet{benjamini2005false} also provided a selection-agnostic method that multiplies the confidence level $\alpha$ by a quantity related to the proportion of selected candidates over all candidates. We extend this adjusted approach to the conformal prediction setting and prove it can indeed control $\FCR$ at the target level. However, the $\FCR$-adjusted method is generally known to yield uniformly inflated confidence intervals \citep{weinstein2013selection}, which is also verified in our numerical experiments. This is because the adjusted confidence intervals discard the selection event when calculating the miscoverage probabilities of selected confidence intervals.

This paper develops a novel conformal framework to construct post-selection prediction intervals while controlling the $\FCR$ around the target level. The key ingredient of our proposal entails performing a pre-specified selective procedure on both the calibration set and test set. Let $\hgS_c \subseteq \gC$ be the indices of the selected calibration set. With the help of conditional empirical distribution obtained by the labeled samples in $\hgS_c$, we can construct the conditional conformal prediction intervals for $j\in \hgS_u$ as
\begin{align}\label{eq:SCOP_PI}
    \PI^{\SCOP}_j = \hmu(X_j) \pm Q_{\alpha}\LRs{\{R_i\}_{i\in \hgS_c}},
\end{align}
where $Q_{\alpha}(\{R_i\}_{i\in \hgS_c})$ denotes  $\lceil (1 - \alpha)(|\hgS_c|+1)\rceil \text{-th smallest value in } \{R_i\}_{i\in \hgS_c}$. We refer this procedure as {\it selective conditional conformal predictions} (abbreviated as SCOP in superscripts). The proposed procedure is model-agnostic, in the sense that it could wrap around any prediction algorithms with commonly used selection procedures to provide prediction intervals. We investigate the $\FCR$ control of the proposal under two general classes of selective rules: exchangeable selection and ranking-based selection. 

The main contributions of the paper are summarized as follows. Firstly, we investigate the $\FCR$-adjusted method in the setting of conformal inference and show that it can achieve $\FCR$ control if the selection rule does not depend on the calibration set. Secondly, under a unified framework and exchangeable assumption on the selection rule, we show that our method exactly controls the $\FCR$ at the target level. Thirdly, we adapt our method with an inflated threshold in choosing calibration points under the ranking-based selection beyond exchangeability. If the ranking threshold depends only on the test set, our method enjoys a distribution-free $\FCR$ control and anti-conservation guarantee. If the ranking threshold depends on both the test set and calibration set, we provide a non-asymptotic bound for the $\FCR$ control gap under mild distributional assumptions. Finally, we illustrate the easy coupling of our method with commonly used prediction algorithms, and numerical experiments show that our method exhibits more accurate control over $\FCR$ compared to existing methods while offering a narrowed prediction interval.

\section{Connections to existing works}\label{sec:related_work}

Post-selection inference on a large number of variables has attracted considerable research attention. Along the path of $\FCR$-adjusted procedure, \citet{weinstein2013selection}, \citet{zhao2020constructing} and \citet{zhao2022general} further proposed some methods to narrow the adjusted confidence intervals by incorporating more selection information. Recently, \citet{xu2022post} combined the $\FCR$-adjusted procedure with ``$e$-value'' based confidence intervals. Among some others, \citet{fithian2014optimal}, \citet{lee2016exact} and \citet{taylor2018post}  proposed constructing conditional confidence intervals for each selected variable and showed that the selective error rate can be controlled given that the selected set is equal to some deterministic subset. Specially, \citet{reid2017post} considered the inference problem after Top-K selection or Benjamini--Hochberg procedure \citep{benjamini1995controlling}. However, those methods either require some tractable conditional distribution assumptions or are only applicable for some given prediction algorithms, such as normality assumptions or LASSO model. Besides, a relevant direction is a splitting-based strategy for high-dimensional inference. The number of variables is firstly reduced to a manageable size using one part of the data, while confidence intervals or significance tests can be constructed by computing estimates in a low-dimensional region with the other part of the data and selected variables. See \citet{wasserman2009high}, \citet{rinaldo2019bootstrapping}, \citet{du2021false} and the references therein. One potential related work is \citet{chen2020valid}, in which the authors considered constructing confidence intervals for regression coefficients after removing the potential outliers from the data. Our paradigm differs substantially from those works as we focus on post-selection inference for sample selection rather than variable selection, and it is difficult to extend existing works on variable selection to the present problem due to the requirements of model or distribution assumptions. At last,  \citet{weinstein2020online} provided a nice solution to the problem of $\FCR$ control in the online setting. The authors also discussed the selective inference problem under the framework of conformal prediction, but that is different from our offline setting.

The building block of our proposed method is the conformal inference framework, which has been well studied in many settings, including non-parametric regression \citep{lei2013}, quantile regression \citep{romano2019conformalized}, high-dimensional regression \citep{lei2018distribution} and classification \citep{sadinle2019least,romano2020classification}, etc. More comprehensive reviews can be found in \citet{shafer2008tutorial} and \citet{angelopoulos2021gentle}. Conventionally, each conformal prediction interval enjoys a distribution-free marginal coverage guarantee if the data points are exchangeable. However, the exchangeability may be violated in practice, and the violation would be more severe in the post-selection conformal inference because the selection procedure might be determined by either the labeled data or the test data, or both. 
Conformal inference beyond exchangeability has attracted attention \citep{tibshirani2019conformal,lei2021conformal,candes2021conformalized}. In particular, \citet{barber2022conformal} used weighted quantiles to implement conformal inference when the predictive algorithms cannot treat data symmetrically. However, how to quantify or characterize the non-exchangeability between the selected test set and the selected calibration set when the distribution is unknown remains a challenge. To address this challenge, we introduce a new technique by carefully constructing a {virtual} post-selection calibration set in theory, which may be of independent interest for conformal prediction for non-exchangeable data.

\section{Selective conditional conformal prediction}\label{sec:conformal-pred}

\subsection{Problem statement}

Now we state the selective predictive inference problem in the conformal framework. Let {$g:\sR^d \to \sR$} be one plausible selection-score function, which is user-specified or estimated by the training data $\gD_t$. Denote $T_i=g(X_i)$ for $i \in \gC \cup \gU$. A particular selection procedure $\mathbf{S}$ is applied to $\{T_i\}_{i\in\gU}$ to find the test points of interest, and those $X_i$'s with smaller values of $T_i$ tend to be chosen by $\mathbf{S}$. Let the selected test set be $\hat{\gS}_u=\{i\in\gU: T_i\leq\hat{\tau}\}$, where $\hat{\tau}$ is the threshold determined by the selection procedure. Various procedures $\mathbf{S}$ can be chosen from different perspectives, and we summarize the selection threshold $\hat{\tau}$ into three types: (a) The threshold $\hat{\tau}$ is user-specified or independent of the calibration set $\gD_c$ and test set $\gD_u$. For example, $\hat{\tau}= \tau_0$, where $\tau_0$ is a prior value or could be obtained from a process depending only on the training set $\gD_t$; (b) The threshold $\htau$ depends only on the test scores $\{T_i\}_{i\in\gU}$. This type includes the Top-K selection which chooses the first $K$ individuals, and the quantile-based selection in which a given proportion of individuals with the smallest $T_i$ values in the test set are selected \citep{taylor2018post}; (c) The threshold $\htau$ fully or partially relies on the calibration set $\gD_c$. For example, $\htau$ is some quantile of the observed responses in $\gD_c$, or the quantile of the predicted responses in both $\gD_c$ and $\gD_u$. In particular, one may employ some multiple testing procedures with conformal $p$-values \citep{bates2021testing,jin2022selection} to achieve error rate control, such as the false discovery rate ($\FDR$) control by applying Benjamini--Hochberg procedure.

Hereafter, similar to \citet{bates2021testing} and \citet{angelopoulos2023prediction}, we treat the training set $\gD_t$ as fixed and only consider the randomness from the calibration set $\gD_c$ and test set $\gD_u$. Our goal is to construct valid conformal prediction intervals for the selected candidates in $\hgS_u$ while controlling $\FCR$ defined in \eqref{eq:FCR_def} below or around the target level $\alpha \in (0,1)$. {In Section \ref{subsec:ACP}, we first investigate adjusted conformal prediction and show that it is able to achieve fair control. Then, we introduce our proposed method and validate its effectiveness specifically for exchangeable selective procedures in Section \ref{subsec:SCOP}.}

\subsection{Adjusted conformal prediction}\label{subsec:ACP}

To control the $\FCR$ after the selection of parameters, \citet{benjamini2005false} proposed an adjusted method by constructing a more conservative confidence interval for each selected candidate. Now we generalize it to the present conformal prediction setting. Define
\[
 M^j_{\min} = \min_{t \in \sR}\LRl{|\hgS_u^{j\gets t}|:j\in \hgS_u^{j\gets t}},
 \]
where $\hgS_u^{j\gets t}$ denotes the selected subset when replacing $T_j$ with some deterministic value $t$. The $\FCR$-adjusted conformal prediction intervals amount to marginally constructing a wider prediction interval with level $(1 - \alpha_j)$ instead of the constant level $(1 - \alpha)$ in \eqref{eq:marginal_PI}, where $\alpha_j = \alpha\times M_{\min}^j/m$. That is for each $j\in\hgS_u$,
\begin{align}\label{eq:adj_PI}
    \PI_j^{\text{AD}} = \hat{\mu}(X_j) \pm Q_{\alpha_j}\LRs{\LRl{R_i}_{i\in \gC}},
\end{align}
where $Q_{\alpha_j}(\{R_i\}_{i\in \gC})$ denotes  $\lceil (1 - \alpha)(n+1)\rceil \text{-th smallest value in } \{R_i\}_{i\in \gC}$. Given the model $\hat{\mu}(\cdot)$, the adjusted prediction intervals $\PI_j^{\text{AD}}$'s are not independent of each other because they all rely on the empirical quantile obtained from $\gD_c$, and therefore the proofs in \citet{benjamini2005false} cannot be trivially extended to our setting. The following proposition demonstrates that the $\FCR$-adjusted approach can successfully control the $\FCR$ for any selection threshold that is independent of the calibration set.

\begin{proposition} \label{proposition:BY-FCR-general}
   Suppose that $\{(X_i, Y_i)\}_{i\in \gC \cup \gU}$ are independent and identically distributed random variables and the selection threshold $\htau$ is independent of the calibration set $\gD_c$. Then the $\FCR$ value of the $\FCR$-adjusted method in (\ref{eq:adj_PI}) satisfies $\FCR^{\operatorname{AD}}\leq\alpha$.
\end{proposition}

The $\FCR$-adjusted method is known to be quite conservative because it does not incorporate the selection event into the calculation \citep{weinstein2013selection}. A simple yet effective remedy is to use conditional calibration, which is our proposed framework in the next subsection.


\subsection{Selective conditional conformal prediction}\label{subsec:SCOP}

{We start by making a decomposition of the inner numerator of $\FCR$ according to the contribution of each sample in the selected set $\hat{\gS}_u$, given as  
$\sP(Y_j \not \in \PI_j \mid j\in\hat{\gS}_u)\sP(j\in\hat{\gS}_u)$. Notice that the $\FCR$ can be tightly controlled if the conditional miscoverage probability satisfies $\sP(Y_j \not \in \PI_j\mid j \in \hgS_u) \leq \alpha$, which sheds light on the construction of conditional conformal prediction intervals.}

{For split conformal inference, the \emph{marginal uncertainty} of the test residual $R_j$ is measured by the residuals of the full calibration set $\{R_i\}_{i\in \gC}$. In the context of selective inference, the {\it conditional uncertainty} of $R_j$ given the selection condition $j\in\hat{\gS}_u$ can be reliably approximated by the calibration set $\gD_c$ and the selection threshold $\htau$. To be specific, we can obtain the post-selection calibration set $\hgS_c$ with the same threshold $\hat{\tau}$, that is $\hgS_c = \LRl{i\in \gC: T_i \leq \htau}$.} Notice that $\hgS_c$ is formed via the same selection criterion with $\hgS_u$, and thus we could utilize the residuals $R_i$ for $i\in\hgS_c$ to approximately characterize the conditional uncertainty of $R_j$ for $j\in \hgS_u$. The conditional conformal prediction intervals can accordingly be constructed as \eqref{eq:SCOP_PI}.
The procedure is summarized in Algorithm \ref{alg:SCOP}.




\begin{algorithm} 
	\renewcommand{\algorithmicrequire}{{Input:}}
	\renewcommand{\algorithmicensure}{{Output:}}
	\caption{Selective conditional conformal prediction}
	\label{alg:SCOP}
	\begin{algorithmic}
		\REQUIRE Labeled set $\gD_l$, test set $\gD_u$, selection procedure $\mathbf{S}$, target $\FCR$ level $\alpha \in (0,1)$.
		
		\STATE {STEP 1} (Splitting and training). Split $\gD_l$ into training set $\gD_t$ and calibration set $\gD_c$ with equal size $n$. Fit prediction model $\hmu(\cdot)$ and score function $g$ (if needed) on the training set $\gD_t$.
		
		\STATE {STEP 2} (Selection). Compute the scores: $\{T_i\}_{i \in \gC}$ and $\{T_i \}_{i \in \gU}$. Apply the selective procedure $\mathbf{S}$ to $\{T_i\}_{i \in \gC \cup \gU}$ and obtain the threshold value $\htau$.  Obtain the post-selection subsets: $\hgS_u = \{i\in \gU: T_i \leq \htau\}$ and $\hgS_c = \{i\in \gC: T_i \leq \htau\}$.
		
		\STATE {STEP 3} (Calibration). Compute residuals: $\{R_i = |Y_i - \hmu(X_i)|: i\in \hgS_c\}$. 
		
		\STATE {STEP 4} (Construction). Construct prediction interval for each $j \in \hgS_u$ as $\PI^{\SCOP}_j = \hmu(X_j) \pm Q_{\alpha}(\{R_i\}_{i\in \hgS_c})$.
		
       \ENSURE Prediction intervals $\{\PI^{\SCOP}_j: j \in \hgS_u\}$.
	\end{algorithmic}
\end{algorithm}

The following theorem shows that Algorithm \ref{alg:SCOP} can control the $\FCR$ at $\alpha$ for exchangeable selection threshold $\hat{\tau}$, {i.e., the value of $\hat{\tau}$ is invariant to the permutation of $\{T_i\}_{i\in \gC \cup \gU}$.}

\begin{theorem}\label{thm:FCR_exchange}
Suppose $\{(X_i, Y_i)\}_{i\in \gC \cup \gU}$ are independent and identically distributed random variables, and the threshold $\hat{\tau}$ is exchangeable with respective to the $\{T_i\}_{i\in \gC \cup \gU}$. Then, for each $j\in \gU$, the conditional miscoverage probability is bounded by
\begin{align}\label{eq:selection_conditional_miscover}
    \sP\LRs{Y_j \not\in \PI_j^{\SCOP}\mid j\in \hgS_u} \leq \alpha.
\end{align}
Further, the $\FCR$ value of Algorithm \ref{alg:SCOP} is controlled at $\FCR^{\SCOP} \leq \alpha$.  If $\{R_i\}_{i\in \gC \cup \gU}$ are distinct values almost surely and $\sP(|\hgS_u| > 0) = 1$, we also have $\FCR^{\SCOP} \geq \alpha - \E\{(|\hgS_c| + 1)^{-1}\}$.
\end{theorem}

Under independent and identically distributed assumption on the data, the $\FCR$ control results match the marginal miscoverage property of the split conformal prediction intervals \citep{vovk2005algorithmic,lei2018distribution}. As long as $\htau$ is lower bounded by a constant $\tau$ such that $\sP(T_i\leq \tau)=O(1)$, the expectation of $ (|\hgS_c| + 1)^{-1}$ would be $O(n^{-1})$, which guarantees the exact $\FCR$ level in asymptotic regime. {The proof of this theorem relies on the exchangeable condition on the selection threshold $\hat{\tau}$. The simplest case is the selection with a fixed threshold. Another example is that $\hat{\tau}$ is some quantile of $\{T_i\}_{i\in{\gC}\cup {\gU}}$.}
However, many selection procedures may be excluded, such as the Top-K selection. In such cases, the threshold $\htau$ is only determined by the test data, which does not treat the data points from calibration and test sets symmetrically. We will explore the effectiveness of Algorithm \ref{alg:SCOP} for more general selection procedures in the next section.

\begin{remark}
    Several works considered approximately constructing the test-conditional prediction interval \citep{victor2021distributional,feldman2021improving}, i.e., for any $x\in \sR^d$,
\begin{align}\label{eq:full_conditional_miscover}
         \sP\LRl{Y_j\notin \PI(X_j)\mid X_j = x }\leq \alpha,
     \end{align}
     where $\PI(X_j) \equiv \PI_j$ is a prediction interval constructed upon the covariate $X_j$.
    However, it is well known that achieving ``fully'' conditional validity in \eqref{eq:full_conditional_miscover} is impossible in distribution-free regime \citep{lei2013,Barber20limits}. Our selective conditional miscoverage control in \eqref{eq:selection_conditional_miscover} is a weaker guarantee compared with \eqref{eq:full_conditional_miscover} since we only consider the selection events. We refer to Appendix B in \citet{weinstein2020online} for more discussion about these two conditional guarantees. {The post-selection calibration data in $\hgS_c$ enables us to approximate the selective conditional distribution of residuals and further leads to conditional coverage.}
\end{remark}

\section{$\FCR$ control for ranking-based selection}\label{sec:ranking-selection}
\subsection{Ranking-based selection}
In this subsection, we introduce a general class of selection rules 
named \emph{ranking-based selection}, and then we make an exchangeable adaptation for Algorithm \ref{alg:SCOP} to achieve better $\FCR$ control results for this type of selection rule. 

Suppose the selection algorithm $\mathbf{S}$ is conducted based on $\{T_i\}_{i\in\gC \cup \gU}$, and outputs a ranking threshold $\hkappa \leq m$. Then the selected subset of the test set can be rewritten as
\begin{align}\label{eq:def_Su_hat}
    \hgS_u = \LRl{j\in \gU: T_j \leq \ermT_{(\hkappa)}},
\end{align}
where $\ermT_{(\hkappa)}$ is the $\hkappa$-th smallest value in $\{T_i\}_{i\in \gU}$. If $\hkappa$ depends only on the test set, we call the procedure as \emph{test-driven} selection; if $\hkappa$ depends on both the test and calibration sets, we refer to the procedure as \emph{calibration-assisted} selection. This ranking-based procedure in \eqref{eq:def_Su_hat} incorporates many practical examples, such as Top-K selection, quantile-based selection, and step-up procedures for error rate control including Benjamini--Hochberg method \citep{Romano2006,Sarkar2007,lei2020}.

For the ranking-based selection rule, we introduce an inflated threshold $\ermT_{(\hkappa+1)}$ for conditional calibration, where the corresponding post-selection calibration set is defined as
\begin{align}\label{eq:Sc_inflate}
    \hgS_c^{+} = \LRl{i\in \gC: T_i \leq \ermT_{(\hkappa + 1)}}.
\end{align}
{Replacing the $\hgS_c$ in Algorithm \ref{alg:SCOP} with $\hgS_c^{+}$}, we get a refined procedure for the ranking-based selection, which is denoted as Algorithm \ref{alg:SCOP}+. We briefly explain why using the inflated threshold in (\ref{eq:Sc_inflate}) can yield exchangeability on selection indicators. Let us consider a simple case where the ranking threshold is fixed, that is $\hkappa = \kappa$. Slightly abusing notation, we write $\ermT_{-j,(\kappa)}$ as the $\kappa$-th smallest value in $\{T_i\}_{i\in \gU \setminus \{j\}}$.
Firstly, when $\{T_i\}_{i\in\gU}$ are almost surely distinct, for any $j\in \gU$ and $\kappa < m$, we know $\ermT_{(\kappa+1)} = \ermT_{-j,(\kappa)}$ if and only if $T_j \leq \ermT_{(\kappa)}$. Secondly, $\{T_j \leq \ermT_{(\kappa)}\} = \{T_j \leq \ermT_{-j,(\kappa)}\}$ holds for any $j \in \gU$ and $\kappa < m$, {which is also referred to as inflation of quantiles in the literature of conformal inference \citep{romano2019conformalized}.}
Based on these two facts, we can equivalently write the selection conditions in \eqref{eq:def_Su_hat} and \eqref{eq:Sc_inflate} as
\begin{align*}
    \{j\in \hgS_u, k\in \hgS_c^+\} = \LRl{T_j \leq \ermT_{-j,(\kappa)}, T_k \leq  \ermT_{-j,(\kappa)}},
\end{align*}
where $\ermT_{-j,(\kappa)}$ 
is independent of the samples $(X_j,Y_j)$ and $(X_k,Y_k)$ for $j\in\gU$ and $k\in\gC$. As a consequence, the residual pair $(R_j, R_k)$ still has exchangeability under the joint selection condition stated above. This insight from the fixed ranking threshold is a starting point for our finite-sample control results with the data-dependent threshold.


\subsection{$\FCR$ control for test-driven selection}\label{sec:test-driven-results}
In this subsection, we establish the finite-sample $\FCR$ control results for test-driven selection, where the ranking threshold $\hkappa$ depends only on the test set $\gD_u$.
We introduce the following assumption on the scores and residuals, which holds naturally when $\{(X_i, Y_i)\}_{i\in \gC \cup \gU}$ are independent and identically distributed, and the training set $\gD_t$ is fixed.

\begin{assumption}\label{assum:continuous}
Suppose both $\{T_i\}_{i\in \gC \cup \gU}$ and $\{R_i\}_{i\in \gC \cup \gU}$ are independent and identically distributed continuous random variables.
\end{assumption}

For more general cases, we expect that Algorithm \ref{alg:SCOP}+ enjoys a finite-sample $\FCR$ control. To reach this target, we impose the following assumption to decouple the dependency between the ranking threshold $\hkappa$ and the selected candidate $j\in \hgS_u$. Let $\hkappa^{j\gets t}$ be the ``virtual'' ranking threshold obtained from the selection algorithm $\mathbf{S}$ but replacing $T_j$ with some deterministic value $t$.

\begin{assumption}\label{assum:htau}
{There exists some deterministic value $t_u \in \sR$ such that $\hkappa^{j\gets t_u} = \hkappa$ holds for any $j \in \gU$. In addition, the ranking threshold satisfies $\hkappa \leq m-1$.}
\end{assumption}

For the selection rules such as quantile-based selection and Top-K selection, this assumption is satisfied with $t_u = -\infty$. It enables the leave-one-out analysis in the $\FCR$ control of Algorithm \ref{alg:SCOP}+. It is worth noticing that all test samples are selected when $\hkappa = m$, in which there is no need to construct conditional prediction intervals. Therefore, we require $\hkappa \leq m-1$ in Assumption \ref{assum:htau}.

\begin{theorem}\label{thm:finite_control}
    Suppose Assumptions \ref{assum:continuous} and \ref{assum:htau} hold. Under test-driven selection procedures, the $\FCR$ value of Algorithm \ref{alg:SCOP}+ can be controlled at
    \begin{align}\label{eq:test_driven_FCR}
        \alpha - \E\LRm{\LRl{(n+1)\ermT_{(\hkappa+1)}}^{-1}} \leq \FCR^{\SCOP} \leq \alpha.
    \end{align}
\end{theorem}

The upper bound in \eqref{eq:test_driven_FCR} shows that Algorithm \ref{alg:SCOP}+ using an inflated threshold to construct $\hgS_c^+$ in \eqref{eq:Sc_inflate}, can control $\FCR$ below the target level. More importantly, our result inherits the model-free and distribution-free guarantee from split conformal inference. The lower bound in \eqref{eq:test_driven_FCR} provides an anti-conservation guarantee for Algorithm \ref{alg:SCOP}+, where the term $\{(n+1)\ermT_{(\hkappa+1)}\}^{-1}$ measures the sample size of the selected calibration set $\hgS_c^{+}$. When the ranking threshold $\hkappa$ diverges when $m$ tends to infinity, we can ensure $\E[\{\ermT_{(\hkappa+1)}\}^{-1}] = O(1)$ and then the $\FCR$ value can be exactly controlled at the target level $\alpha$ in the asymptotic regime.

\subsection{$\FCR$ control for calibration-assisted selection}\label{sec:cal-assisted-results}

In the proof of Theorem \ref{thm:finite_control}, the virtual ranking threshold $\hkappa^{j\gets t_u}$ helps us decouple the dependence between the selected test candidate $j$ and the ranking threshold $\hkappa$. For calibration-assisted selective procedures, the analysis is more difficult because the $\hkappa$ also depends on the calibration set. It demands a more tractable virtual ranking threshold to decouple the dependence on the selected test samples and the calibration samples. To achieve this, we construct a virtual ranking threshold $\hkappa^{(j,k) \gets (t_u,t_c)}$ by replacing $T_j$ with $t_u$ and $T_k$ with $t_c$ simultaneously for $j\in \gU$ and $k\in \gC$.


\begin{assumption}\label{assum:decouple_k}
There exists some deterministic values $t_u, t_c \in \sR$ and integers $I_u,I_c \geq 0$ such that: (1) $\hkappa^{j\gets t_u} = \hkappa$ holds for any $j \in \hgS_u$ and $\hkappa\leq \hkappa^{j\gets t_u}\leq \hkappa + I_u$ for any $j\in \gU \setminus \hgS_u$; (2) $\hkappa^{j\gets t_u} \leq \hkappa^{(j,k) \gets (t_u, t_c)} \leq \hkappa^{j\gets t_u} + I_c$ holds for any $j\in \gU$ and $k\in \gC$.
In addition, there exists some $\gamma \in (0,1)$ such that $\lceil\gamma m\rceil\leq \hkappa\leq m-1$.
\end{assumption}

{
We introduce Assumption \ref{assum:decouple_k} to characterize the stability of the ranking threshold $\hkappa$ more carefully due to the complexity of calibration-assisted selection.
For $j\in \gU \setminus \hgS_u$, we use the quantity $I_u$ to bound the change in $\hkappa$ after replacing $T_j$. Further, the quantity $I_c$ is used to bound the change in $\hkappa$ after replacing $(T_j, T_k)$ with $(t_u,t_c)$. The lower bound on $\hkappa$ is used to guarantee enough sample size for 
$\hgS_c^+$ because $\ermT_{(\hkappa+1)}$ would become extremely small if $\hkappa = O_p(1)$.}


The virtual ranking threshold $\hkappa^{(j,k) \gets (t_u,t_c)}$ leads to a different post-selection calibration set and the corresponding conformal prediction interval. We use the following assumption on the joint distribution of $(R_i, T_i)$ to control the difference between the real and virtual prediction intervals.

\begin{assumption}\label{assum:joint_CDF}
    Denote $F_R(\cdot)$ and $F_T(\cdot)$ the cumulative distribution functions of $R_i$ and $T_i$, respectively. Let $F_{(R,T)}(\cdot,\cdot)$ be the joint cumulative distribution function of $(F_R(R_i), F_T(T_i))$. There exists $\rho \in (0,1)$ such that $\partial F_{(R,T)}(r,t)/\partial r \geq \rho t$ holds for any $t \in [\gamma/2, 1]$ and $r \in (0, 1)$.
\end{assumption}

\begin{theorem}\label{thm:FCR_bound_cal}
Suppose Assumptions \ref{assum:continuous}, \ref{assum:decouple_k} and \ref{assum:joint_CDF} hold.
If $n \gamma\geq 256\log n$ and $m\gamma \geq 32\log m$, the $\FCR$ value of Algorithm \ref{alg:SCOP}+ for the calibration-assisted selection procedures can be controlled at
\begin{align}\label{eq:FCR_bound_cal}
    \FCR^{\SCOP} - \alpha = O\LRl{\frac{1}{\rho \gamma^3}\LRs{\frac{I_c \log m}{ m} + \frac{\log n}{n}} + \frac{I_u \log m}{m\gamma}}.
\end{align}
\end{theorem}

This theorem shows that our method can control the $\FCR$ around the target value with a small gap under a more complex selection rule. {The dependence of \eqref{eq:FCR_bound_cal} on $m$ includes two terms $I_u\log m/m$ and $I_c\log m/m$, characterizing the difference between the real threshold and the virtual thresholds after replacing $T_j$ and $(T_j,T_k)$ respectively, i.e. $\ermT_{(\hkappa^{j\gets t_u})} - \ermT_{(\hkappa)}$ and $\ermT_{(\hkappa^{(j,k)\gets (t_u,t_c)})} - \ermT_{(\hkappa)}$.}
Since the dependence between $\hkappa$ and the calibration set can be arbitrary, it is difficult to achieve a finite-sample control like Theorem \ref{thm:finite_control} without further assumptions on $\hkappa$'s structure. If $\gamma$ and $\rho$ are fixed and $I_c \log m = o(m)$, we also have an asymptotic $\FCR$ control.

\subsection{Prediction-oriented selection with conformal $p$-values}\label{sec:conformal_p-values}

We discuss the implementation of our proposed procedure under a special application in conformal inference, the selection via multiple testing based on conformal $p$-values. The concept of the conformal $p$-value was proposed by \citet{vovk2005algorithmic}.
Like the conformal prediction interval, the conformal $p$-values enjoy model/distribution-free properties. Recently, there exist some works to apply conformal $p$-values to perform sample selection from a multiple-testing perspective. 
{\citet{bates2021testing} proposed a method to construct conformal $p$-values with data splitting and apply it to detect outliers. \citet{zhangautoms} extended that method and proposed a Jackknife implementation combined with automatic model selection. 


In particular, \citet{jin2022selection} investigated the \emph{prediction-oriented selection} problem aiming to select samples whose unobserved outcomes exceed some specified values and proposed a new method based on the conformal $p$-values to control the $\FDR$.
This problem can be viewed as the following multiple hypothesis tests: for $i\in \gU$ and some $b_0 \in \sR$,
\begin{align*}
    H_{0,i}: Y_i \geq b_0\quad \text{v.s.}\quad H_{1,i}: Y_i < b_0.
\end{align*}
One could take the score function as $g(x) = \hmu(x) - b_0$ and compute the scores as $\{T_i=g(X_i)\}_{i\in\gC\cup\gU}$.
Denote the null set of calibration samples as $\gC_0=\{i\in\gC: Y_i\geq b_0\}$.
The conformal $p$-value for each test data point can be calculated by
\begin{align}\label{eq:conformal_p_values}
    p_j = \frac{1 + |\{i\in \gC_0: T_i \leq T_j\}|}{|\gC_0|+1}, \quad \text{for }j\in \gU.
\end{align}
To control the $\FDR$ at the level $\beta \in (0,1)$, we deploy Benjamini--Hochberg procedure to $\{p_j\}_{j\in \gU}$ and obtain the rejection set $\hgS_u$. Let $p_{(r)}$ be the $r$-th smallest value in $\{p_i\}_{i\in \gU}$, then define
\begin{align}\label{eq:kappa_BH}
    \hkappa = \max \LRl{1\leq r\leq m:\ p_{(r)} \leq \frac{r \beta}{m}},
\end{align}

\begin{proposition}\label{pro:$p$-value_T-value}
For any $i\in \gU$, it holds that $\{p_i \leq p_{(\hkappa)}\} = \{T_i \leq \ermT_{(\hkappa)}\}$.
\end{proposition}

Proposition \ref{pro:$p$-value_T-value} indicates that using the conformal $p$-values in \eqref{eq:conformal_p_values} to obtain $\hgS_u$ is equivalent to using the scores $\{T_i\}_{i\in \gU}$ with the same ranking threshold $\hkappa$ in \eqref{eq:kappa_BH}, that is
\begin{align*}
    \hgS_u = \LRl{i\in \gU: p_i \leq p_{(\hkappa)}} \equiv \LRl{i\in \gU: T_i \leq \ermT_{(\hkappa)}}.
\end{align*}
Therefore, we can frame the Benjamini--Hochberg procedure with conformal $p$-values as a calibration-assisted selection in Section \ref{sec:cal-assisted-results}. By Algorithm \ref{alg:SCOP}+, we obtain the selected calibration set $\hgS_c^+$ in \eqref{eq:Sc_inflate} with the inflated threshold $\ermT_{(\hkappa+1)}$. If Assumptions \ref{assum:continuous}, \ref{assum:decouple_k} and \ref{assum:joint_CDF} are satisfied, we have the same bound in Theorem \ref{thm:FCR_bound_cal} for the $\FCR$ control gap.
\begin{remark}
    To adapt Assumptions \ref{assum:htau} and \ref{assum:decouple_k}, we can simply take $t_u = -\infty$ and $t_c = +\infty$ by replacing $T_j$ with $-\infty$ for $j \in \gU$ and $T_k$ with $+\infty$ for $k \in \gC$, respectively. By the step-up property of Benjamini--Hochberg procedure, we have $\hkappa^{j\gets t_u} = \hkappa$ for any $j\in \hgS_u$. In addition, replacing $T_k$ with $+\infty$ leads to a smaller conformal $p$-value due to the construction in \eqref{eq:conformal_p_values}. Therefore, we can guarantee $\hkappa^{(j,k)\gets (t_u,t_c)} \geq \hkappa^{j\gets t_u}$. For the Benjamini--Hochberg procedure with standard $p$-values, it has been proved that $\hkappa/m$ could converge (in probability) to a constant in $(0,1)$ when $m$ tends to infinity, even under dependence \citep{Genovese2002, storey2004strong,Ferreira2006}. Hence $\hkappa \geq \lceil \gamma m \rceil$ in Assumption \ref{assum:decouple_k} is reasonable for a fixed $\gamma$.
\end{remark}

\section{Numerical results}\label{sec:numr}

We illustrate the breadth of applicability of the proposed methods via comprehensive simulation studies. In each setting, we generate independent and identically distributed 10-dimensional covariates $X_i$ from uniform distribution $\unif([-1,1]^{10})$ and the corresponding responses as $Y_i=\mu(X_i)+\epsilon_i$. Two scenarios are considered with different models of $\mu(\cdot)$ and distributions of $\epsilon_i$'s: {Scenario A} considers a linear model $\mu(X)=X^\top\beta$ where $\beta$ is randomly sampled from $ \unif([-1,1])^{10}$. The noise is heterogeneous and follows $\epsilon\mid X\sim N(0,\{1+|\mu(X)|\}^2)$; {Scenario B} considers nonlinear model  $\mu(X)=X^{(1)}X^{(2)}+X^{(3)}-2\exp\{X^{(4)}+1\}$, where $X^{(k)}$ denotes the $k$-th element of vector $X$. The noise is from $\epsilon\sim N(0,1)$ and is independent of covariates $X$. 

We fix the labeled data size $2n=400$ and equally split it into $\gD_t$ and  $\gD_c$. The regression model $\hmu(\cdot)$ is fitted on $\gD_t$ by ordinary least-squares for Scenario A, and support vector machine for Scenario B, respectively. The latter algorithm is implemented by \texttt{R} package \texttt{ksvm} with default parameters. In both cases, the selection score is exactly the prediction value, i.e., $T_i=\hat{\mu}(X_i)$. 

The selected subset is $\hgS_u=\{i\in\gU: T_i\leq \hat{\tau}\}$, where $\hat{\tau}$ is the threshold. To illustrate the wide applicability of our proposed method, several selection thresholds $\hat{\tau}$ are considered. 
(i) {T-cal($q$)}: $q\%$-quantile of true response $Y$ in calibration set, that is $\hat{\tau}$ is $q\%$-quantile of $\{Y_i\}_{i\in \gC}$.
(ii) {T-test($q$)}: $q\%$-quantile of predicted response $\hmu(X)$ in test set, i.e. $\hat{\tau}=\ermT_{([qm/100])}$.
(iii) {T-exch($q$)}: $q\%$-quantile of predicted response $\hmu(X)$ in both calibration set and test set, that is $\hat{\tau}$ is the $q\%$-quantile of $\{T_i\}_{\gC\cup\gU}$.
(iv) {T-cons($b_0$)}: a pre-determined constant value $b_0$, i.e., $\hat{\tau}=b_0$. 
(v) {T-pos($b_0, \beta$)}: the threshold for prediction-oriented selection proposed by \citet{jin2022selection}, where one would like to select those test samples with response $Y$ smaller than $b_0$ while controlling the $\FDR$ level at $\beta=0.2$. Here $\hat{\tau}$ is computed by Benjamini--Hochberg procedure with conformal $p$-values in Section \ref{sec:conformal_p-values}.
(vi) {T-top($K$)}: the $K$-th smallest value of $\{T_i\}_{i\in\gU}$, i.e. $\htau=\ermT_{(K)}$.

\begin{figure}
\includegraphics[width = 1\linewidth]{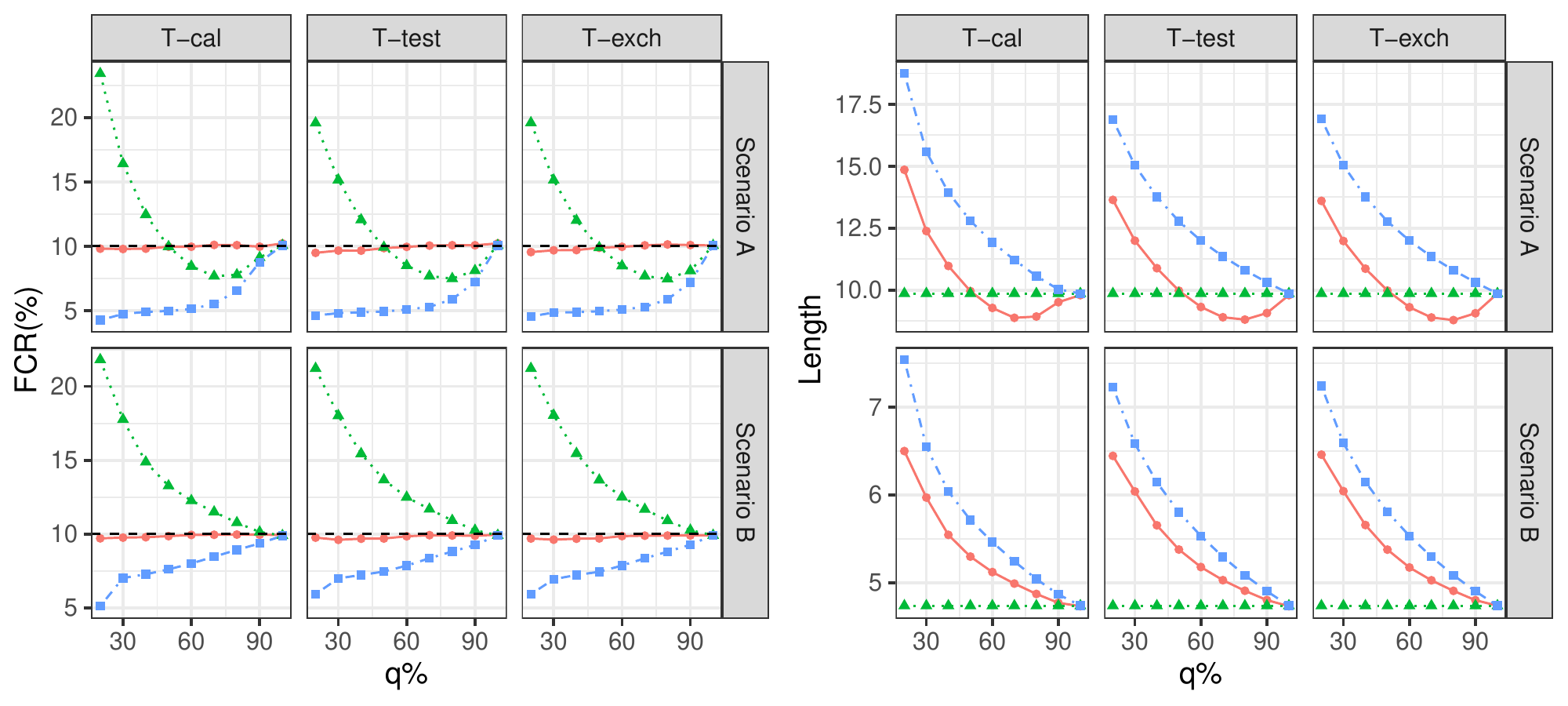}
\caption{Empirical $\FCR$ (\%) and the average length of prediction intervals for quantile based thresholds with varying quantile level $q\%$ of three methods as follows: selective conditional conformal prediction (circle dot solid red line); ordinary conformal prediction \eqref{eq:marginal_PI} (triangle dashed green line); $\FCR$-adjusted conformal prediction \eqref{eq:adj_PI}  (square dot-dash blue line). The black dashed line represents the target $\FCR$ level $10\%$.}
\label{fig:quantile}
\end{figure}

Among all the considered threshold selections, only the {T-exch($q$)} and {T-cons($b_0$)} satisfy the exchangeability with respect to $\{T_i\}_{i\in\gC\cup\gU}$. The threshold {T-top($K$)} is actually a special case of {T-test($q$)} with $K=[qm/100]$. With target $\FCR$ level $\alpha=10\%$, we apply the original Algorithm \ref{alg:SCOP} for those two exchangeable thresholds to construct prediction intervals, while the Algorithm \ref{alg:SCOP}+ is used for other thresholds. Two benchmarks are included for comparison. One is to directly construct a $(1-\alpha)$ marginal $\PI_j^{\text{M}}$ as \eqref{eq:marginal_PI} for each selected sample based on the whole calibration set. We refer to this method as ordinary conformal prediction and notice that it takes no account of the selection effects. Another one is the $\FCR$-adjusted conformal prediction, which builds $(1-\alpha|\hgS_u|/m)$ level $\PI_j^{\text{AD}}$ as \eqref{eq:adj_PI} for each selected individual. The performances are compared in terms of the $\FCR$ and the average length of prediction intervals among 1,000 repetitions.  

We firstly fix the size of test data $\gD_u$ as $m=200$ and consider three quantile-based thresholds: {T-cal($q$)}, {T-test($q$)} and {T-exch($q$)}. Figure \ref{fig:quantile} displays the estimated $\FCR$ and average length of $\PI_j$ through varying the quantile level $q\%$ from $20\%$ to $100\%$. Across all the settings, it is evident that the selective conditional conformal prediction is able to deliver quite accurate $\FCR$ control and results in more narrowed $\PI_j$. As expected, the ordinary conformal prediction yields the same lengths of $\PI_j$ under both scenarios and can only control the $\FCR$ under $q\%=100\%$, that is the situation including all the test data without any selection. This can be understood since the ordinary conformal prediction builds the marginal prediction intervals $\PI_j$ using the whole calibration set without consideration of the selection procedure and thus possesses the length of $\PI_j$ as $2Q_{\alpha}(\{R_i\}_{i\in \gC})$ in \eqref{eq:marginal_PI}. The $\FCR$-adjusted conformal prediction results in considerably conservative $\FCR$ levels and accordingly it performs not well in terms of the length of $\PI_j$ under both cases. This is not surprising as the $\FCR$-adjusted conformal prediction marginally constructs much larger $(1-\alpha|\hat{\gS}_u|/m)$-level intervals than the target level $\alpha$ to ensure the $\FCR$ control.

\begin{table}[htbp]
\centering
\caption{Comparisons of empirical $\FCR$ (\%) and the average length of $\PI_j$'s under Scenarios A and B with target $\FCR$ $\alpha=10\%$. The sample sizes of the calibration and test sets are fixed as $n=m=200$. SCOP: the proposed method; OCP: ordinary conformal interval \eqref{eq:marginal_PI} ; ACP: $\FCR$-adjusted conformal interval \eqref{eq:adj_PI}.}\label{table:other_thresholds}
\resizebox{\textwidth}{!}{
\begin{tabular}{lcccccccccccc}
    && \multicolumn{3}{c}{{T-con($b_0$)}} && \multicolumn{3}{c}{{T-pos($b_0,20\%$)}} && \multicolumn{3}{c}{{T-top($60$)}}   \\ 
    && SCOP    & OCP    & ACP    && SCOP    & OCP    & ACP    && SCOP    & OCP    & ACP    \\ 
\multirow{2}{*}{{Scenario A}}& $\FCR$ & 9.76   & 14.67  & 4.91   && 5.45    & 13.19& 2.17    && 9.73 & 15.26 & 4.90   \\
&Length  & 11.83   & 9.91   & 14.87  && 16.06   & 9.91   & 22.57 && 12.09 & 9.91 & 15.10   \\                                         
\multirow{2}{*}{{Scenario B}}& $\FCR$ & 9.84    & 16.99  & 7.04   && 9.93    & 15.83  & 7.09 && 9.80 & 17.71& 6.98    \\
&Length  & 5.87    & 4.72   & 6.41   && 5.68    & 4.72   & 6.23  && 5.95 & 4.72 & 6.54    \\ 
\end{tabular}}
\end{table}

In Table \ref{table:other_thresholds}, we present the results of the remaining three thresholds, including {T-cons($b_0$)}, {T-pos($b_0,\beta$)} and {T-top($K$)}. 
Here, we fix the constant $b_0$ for both {T-cons($b_0$)} and {T-pos($b_0, \beta$)} at $-1$ in Scenario A and $-8$ in Scenario B, 
and choose the target $\FDR$ level $\beta=20\%$ for {T-pos($b_0, \beta$)} and $K=60$ for {T-top($K$)}. It can be seen that our method achieves $\FCR$ levels 
close to the nominal level, and it also provides a satisfactory narrowed prediction interval under both scenarios. The ordinary conformal prediction \eqref{eq:marginal_PI} leads to different $\FCR$ values but the same average length of $\PI_j$ by different selection thresholds. This corroborates the insight that the ordinary conformal prediction is unable to give a valid coverage guarantee for the selected samples. In contrast, the $\FCR$-adjusted conformal prediction \eqref{eq:adj_PI} yields overly conservative $\FCR$ values, and in turn, its $\PI_j$ lengths would be considerably inflated. 

Similar conclusions can be drawn from additional results in the Supplementary Material.

\begin{figure}[tb]
\includegraphics[width = 1\linewidth]{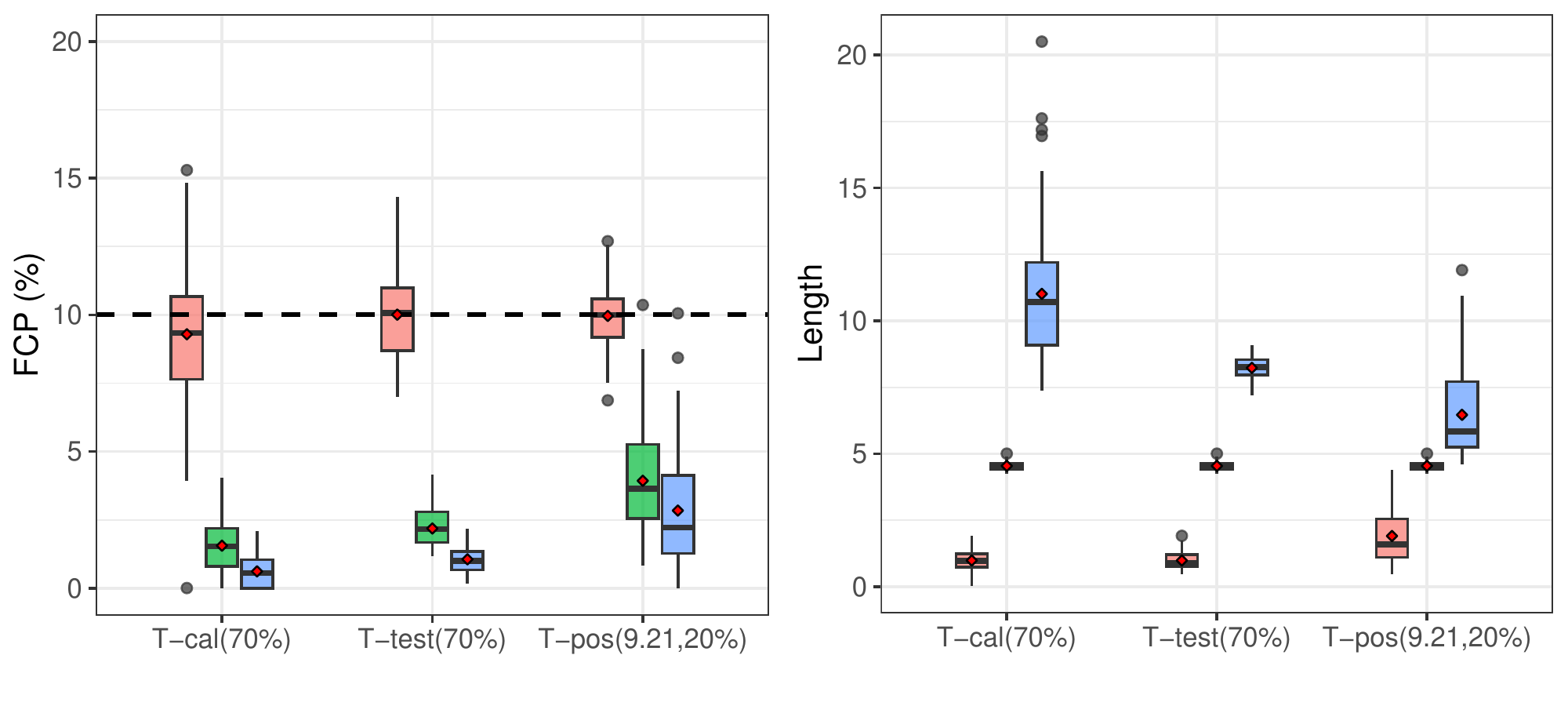}
\caption{Boxplots of the values of $\FCP$ (\%) and length of $\PI_j$ for the drug discovery example by our method (in red), ordinary conformal prediction \eqref{eq:marginal_PI} (in green) and $\FCR$-adjusted conformal prediction \eqref{eq:adj_PI} (in blue). The black dashed line represents the target $\FCR$ level $10\%$ and the red rhomboid dot denotes the average value. }
\label{fig:drug}
\end{figure}

\section{Real data application}\label{sec:real_data}

Early stages of drug discovery aim at identifying those drug-target pairs with high binding affinity of a specific target from a large pool \citep{santos2017comprehensive}. It is important to provide reliable prediction intervals for selecting those promising drug-target pairs for further clinical trials \citep{huang2022artificial}.   
In this example, we apply the proposed method to construct prediction intervals of binding affinities for those drug-target pairs with high binding affinity predictions, while ensuring $\FCR$ control. We consider the DAVIS dataset \citep{davis2011comprehensive}, which contains $25,772$ drug-target pairs. Each pair includes the binding affinity, the structural information of the drug compound, and the amino acid sequence of the target protein; the drugs and targets are firstly encoded into numerical features through \texttt{Python} library \textit{DeepPurpose} \citep{huang2020deeppurpose}, and the responses are taken as the log-scale affinities. We randomly sample $2,000$ observations as the calibration set and another $2,000$ observations as the test set and use the remaining ones as the training set to fit a small neural network model with $3$ hidden layers and $5$ epochs. 

Our goal is to build prediction intervals of drug-target pairs on the test set by selecting their predicted affinities that exceed some specific threshold. Here we consider three different thresholds: {T-cal($70\%$)}, the 70\%-quantile of true responses in the calibration set; {T-test($70\%$)}, the 70\%-quantile of predicted affinities in the test set; {T-pos($9.21, 20\%$)}, which selects those pairs with affinities larger than $9.21$ while controlling $\FDR$ at $0.2$ level. The target $\FCR$ level is $\alpha=10\%$.

To evaluate the performance of our method, we also consider building prediction intervals for those selected candidates by ordinary conformal prediction {\eqref{eq:marginal_PI}} and $\FCR$-adjusted conformal prediction {\eqref{eq:adj_PI}}. Figure \ref{fig:drug} shows the boxplots of the false coverage-statement proportion ($\FCP$), which is defined as
\begin{equation*}
    \FCP = \frac{|\{j\in \hat{\gS}_{u}: \ Y_j \not\in \PI_j\}|}{\max\{|\hat{\gS}_u|, 1\}},
\end{equation*}
and the average length of $\PI_j$ based on three methods among 100 runs. As illustrated, our method has stable $\FCP$ close to the nominal level across all the threshold selections. In comparison, both ordinary conformal prediction and $\FCR$-adjusted conformal prediction result in conservative $\FCP$ levels. The lengths of prediction intervals constructed by ordinary conformal prediction and $\FCR$-adjusted conformal prediction are much broader compared to those of our method. In fact, the responses of log-scale affinities truly range at $(-5,10)$, and thus those broader prediction intervals would barely provide useful information for further clinical trials.

We also consider another application in the house price prediction in the Supplementary Material, and our proposed method also demonstrates superior performance in terms of both $\FCR$ control and narrower prediction intervals, making it promising in practice. 


\section{Discussion}\label{sec:conclu}
We conclude the paper with three remarks. Firstly, we focus mainly on using the residuals as nonconformity scores for the construction of prediction intervals. Our framework can be readily extended to more general scores such as the one based on quantile regression \citep{romano2019conformalized} or distributional regression \citep{victor2021distributional}. To make it clear, we define a general nonconformity score function as $R(\cdot, \cdot): \sR^d \times \sR \to \sR$. The conditional prediction interval \eqref{eq:SCOP_PI} can be generalized to
\begin{align}
    \PI^{\SCOP}_j = \LRl{y\in \sR: R(X_j,y)\leq Q_{\alpha}\LRs{\{R(X_i,Y_i)\}_{i\in \hgS_c}}}.\nonumber
\end{align}
All the theoretical results remain intact under the same conditions. Secondly, as split conformal introduces extra randomness from data splitting and reduces the effectiveness of training models, we may consider the implementation of Algorithm \ref{alg:SCOP} via Jackknife+ \citep{barber2021predictive} or de-randomized inference \citep{ren2022derandomised,bashari2023derandomized} to refine the prediction intervals, and the theoretical guarantee requires further investigation. 
Thirdly, there are some methods proposed to narrow the post-selection intervals \citep{zhao2020constructing,zhao2022general}. How to incorporate their approaches to our selective conformal prediction warrants future research.

\section*{Acknowledgement}
The authors thank the Editor, Associate Editor, and two anonymous referees for their many helpful comments that have resulted in significant improvements in the article. All authors contributed equally to the paper. This research was supported by the China National Key R\&D Program (Grant Nos. 2022YFA1003703, 2022YFA1003800), NNSF of China Grants (Nos. 11925106, 12101398, 12231011, 11931001), and Special Fund Project for the Philosophy and Social Science Laboratory of Chinese Ministry of Education.

\bibliography{references}

\newpage
\appendix

\begin{center}
    \Large{\bf Supplementary Material for ``Selective conformal inference with FCR control''}
\end{center}
\allowdisplaybreaks

\section{Preliminaries}
\subsection{Notations}
The index sets of the calibration set and test set are $\gC$ and $\gU$, respectively. Denote the samples $Z_i = (X_i,Y_i)$ for $i\in \gC \cup \gU$. The selected test set and calibration set are $\hgS_u$ and $\hgS_c$, respectively. In this appendix, we introduce the following new notations:
\begin{itemize}
    \item \emph{Order statistics in subsets.} For any $\gS_c \subseteq \gC$, we write $\ermR_{(\ell)}^{\gS_c}$ as the $\ell$th smallest values in $\ermR^{\gS_c} = \{R_i: i\in \gS_c\}$; For any $\gS_u \subseteq \gU$, we denote $\ermT_{(\ell)}^{\gS_u}$ the $\ell$th smallest values in $\ermT^{\gS_u} = \{T_i: i\in \gS_u\}$.

    \item \emph{Virtual ranking thresholds.} For the ranking threshold in Section \ref{sec:ranking-selection}, we write the $\hkappa^{(j)} = \hkappa^{j \gets t_u}$ and $\hkappa^{(j,k)} = \hkappa^{(j,k) \gets (t_u,t_c)}$.

    \item \emph{Selected calibration set.} For any $t\in \sR$, we define the selected calibration set as $\hgS_c(t) = \{i\in \gC: T_i \leq t\}$.

    \item \emph{Add-one set and drop-one set.} For any $j\in \gU$, we write $\gC_{+j} = \gC \cup \{j\}$ and $\gU_{-j} = \gU \setminus \{j\}$. Corresponding, denote the samples within the two sets as $\gD_{\gC_{+j}} = \{Z_i\}_{i\in \gC_{+j}}$ and $\gD_{\gU_{-j}} = \{Z_i\}_{i\in \gU_{-j}}$.

    \item \emph{Conditional probability and expectation.} For any subset $\gS \subseteq \gC \cup \gU$, we write
    \begin{align*}
        \sP_{\gS}(\cdot) = \sP\LRs{\cdot \mid \{(X_i,Y_i)\}_{i\in \gS}},\quad \E_{\gS}\LRm{\cdot} = \E\LRm{\cdot\mid \{(X_i,Y_i)\}_{i\in \gS}}.
    \end{align*}
\end{itemize}

Recall that the score $T_i = g(X_i)$ and the residue $R_i = |Y_i - \hmu(X_i)|$, hence the values of $T_i$ and $R_i$ are determined by the value of $(X_i,Y_i)$. Then we know that
\begin{align*}
    Q_{\alpha}\LRs{\LRl{R_i}_{i\in \gC}} = \ermR_{(\lceil (1-\alpha) (|\gC|+1) \rceil )}^{\gC},\quad Q_{\alpha}\LRs{\LRl{R_i}_{i\in \hgS_c}} = \ermR_{(\lceil (1-\alpha) (|\hgS_c|+1) \rceil)}^{\hgS_c}.
\end{align*}
Given a positive integer $n$, we use $[n]$ to denote the index set $\{1,\ldots,n\}$.
For the ranking-based selection in Section \ref{sec:ranking-selection}, $\hkappa \in [m]$ is a data-driven threshold that depends on the test set and calibration set. The selected test set of Algorithm \ref{alg:SCOP} is given by
\begin{align*}
    \hgS_u = \LRl{i\in \gU: T_i \leq \ermT_{(\hkappa)}^{\gU}},
\end{align*}
and the selected test set of Algorithm \ref{alg:SCOP}+ is
\begin{align*}
    \hgS_c^{+} = \LRl{i\in \gC: T_i \leq \ermT_{(\hkappa + 1)}^{\gU}}.
\end{align*}

Note that the same selected set $\hgS_u$ and $\hgS_c$ can be obtained if one applies some ranking-based selection procedure to the transformed scores $\{F_T(T_i)\}_{i\in \gU}$ instead of the original scores $\{T_i\}_{i\in\gU}$. Also, the transformed residuals $\{F_R(R_i)\}_{i\in\gC\cup\gU}$ keep the original orders as the residuals $\{R_i\}_{i\in\gC\cup\gU}$. Therefore, without loss of generality, we can assume that $T_i \sim \unif([0,1])$ and $R_i \sim \unif([0,1])$ for $i \in \gC \cup \gU$ in the following assumption and theoretical result. We emphasize that the unknown cumulative distribution functions $F_R(\cdot)$ and $F_T(\cdot)$ are not needed in Algorithm \ref{alg:SCOP}+.

\subsection{Auxiliary lemmas}\label{appen:aux_lemmas}

The following lemma comes from the basic property of sample quantile, which often appears in the conformal inference literature \citep{vovk2005algorithmic,lei2018distribution,romano2019conformalized,barber2021predictive,barber2022conformal}. Here we restate it in the deterministic form.
\begin{lemma}\label{lemma:quantile_inflation}
Let $\rvx_{(\lceil n(1-\alpha)\rceil)}$ is the $\lceil n(1-\alpha)\rceil$-smallest value in $\{\rvx_i \in \sR: i \in [n]\}$. Then for any $\alpha \in (0,1)$, it holds that
\begin{align*}
    \frac{1}{n}\sum_{i=1}^n \Indicator{\rvx_i > \rvx_{(\lceil n(1-\alpha)\rceil)}}\leq \alpha.
\end{align*}
If all values in $\{\rvx_i: i \in [n]\}$ are distinct, it also holds that
\begin{align*}
    \frac{1}{n}\sum_{i=1}^n \Indicator{\rvx_i > \rvx_{(\lceil n(1-\alpha)\rceil)}}\geq \alpha - \frac{1}{n},
\end{align*}
\end{lemma}

Next lemma is a variant of the inflation of quantiles (see Lemma 2 in \citet{romano2019conformalized}).
\begin{lemma}\label{lemma:order_drop_one}
For almost surely distinct random variables $\rvx_1,...,\rvx_n$, let $\rvx_{(1)}\leq \ldots \leq \rvx_{(n)}$ be order statistics of $\{\rvx_i: i\in [n]\}$, and $\rvx_{(1)}^{[n]\setminus\{j\}} \leq \ldots \rvx_{(n-1)}^{[n]\setminus\{j\}}$ be the order statistics of $\{\rvx_i: i\in [n]\setminus\{j\}\}$, then for any $r \in [n-1]$ we have: 
\begin{itemize}
    \item[(1)] $\rvx_{(r)}^{[n]\setminus\{j\}} = \rvx_{(r)}$ if $\rvx_j > \rvx_{(r)}$ and $\rvx_{(r)}^{[n]\setminus\{j\}} = \rvx_{(r+1)}$ if $\rvx_j \leq \rvx_{(r)}$.
    \item[(2)] $\LRl{\rvx_{j} \leq \rvx_{(r)}} = \{\rvx_{j} \leq \rvx_{(r)}^{[n]\setminus\{j\}}\}$.
\end{itemize}
\end{lemma}

    
The next two lemmas are corollaries of the well-known spacing representation of consecutive random variables (c.f. Lemma \ref{lemma:spacing}), and the proofs can be found in \ref{proof:lemma:uniform_spacing_and_ratio} and \ref{proof:lemma:Sc_max_spacing}.

\begin{lemma}\label{lemma:uniform_spacing_and_ratio}
    Let $U_1,\cdots ,U_n \stackrel{i.i.d.}{\sim}\unif([0,1])$, and $U_{(1)}\leq U_{(2)}\leq \cdots \leq U_{(n)}$ be their order statistics. For any absolute constant $C \geq 1$, it holds that
    \begin{align}\label{eq:max_spacing}
        \sP\LRs{\max_{0\leq \ell \leq n-1}\LRl{U_{(\ell+1)} - U_{(\ell)}} \geq \frac{1}{1 - 2\sqrt{\frac{C\log n}{n+1}}}\frac{2C\log n}{n+1}} \leq 2 n^{-C},
    \end{align}
    and
    \begin{align}\label{eq:max_scaled_order_stat}
        \sP\LRs{\max_{1\leq k \leq n} \frac{U_{(k)}}{k} \geq \frac{1}{1 - 2\sqrt{\frac{C\log n}{n+1}}} \frac{2C \log n}{n+1}} \leq 2n^{-C}.
    \end{align}
\end{lemma}

\begin{lemma}\label{lemma:Sc_max_spacing}
Let $\hgS_c(t) = \{i\in \gC: T_i \leq t\}$. If $\frac{d}{d r}F_{(R,T)}(r, t) \geq \rho t$ holds, then for any absolute constant $C \geq 1$, we have
\begin{align*}
    \sP\LRs{\max_{0\leq \ell \leq |\hgS_c(t)|-1}\LRl{\ermR_{(\ell+1)}^{\hgS_c(t)} - \ermR_{(\ell)}^{\hgS_c(t)}} \geq  \frac{1}{\rho}\frac{1}{1 - 2\sqrt{\frac{C\log n}{|\hgS_c(t)|+1}}}\frac{2C\log n}{|\hgS_c(t)|+1} }\leq 2n^{-C}.
\end{align*}
\end{lemma}

Lemma \ref{lemma:Sc_interval_concentration} is used to bound the change in the size of the selected calibration set after changing the threshold. We defer the proof to Section \ref{proof:lemma:Sc_interval_concentration}.

\begin{lemma}\label{lemma:Sc_interval_concentration}
Let $Z_i = \Indicator{t_1 < T_i \leq t_2} - (t_2 - t_1)$ for some fixed $0\leq t_1 < t_2 \leq 1$. Then for any absolute constant $C \geq 1$, we have
\begin{align*}
    \sP\LRs{\frac{1}{n}\LRabs{\sum_{i\in \gC} Z_i} \geq  2\sqrt{\frac{e C \log n}{n} \cdot (t_2-t_1)} + \frac{2e C\log n}{n} }\leq 2n^{-C}.
\end{align*}
\end{lemma}

Lemma \ref{lemma:hgSc_size_lower_bound} is used to lower bound the size of the selected calibration set with arbitrary threshold $t \in (0,1)$, and Lemma \ref{lemma:denominator_lower_bound} is used to guarantee the threshold in selective conditional conformal predictions will be far away from 0 under Assumption \ref{assum:decouple_k}. The proofs can be found in Section \ref{proof:lemma:hgSc_size_lower_bound} and \ref{proof:lemma:denominator_lower_bound}.

\begin{lemma}\label{lemma:hgSc_size_lower_bound}
Let $\hgS_c(t) = \{i\in \gC: T_i \leq t\}$. For any $C \geq 1$, if $8C \log n / (nt) \leq 1$, we have
\begin{align*}
    \sP\LRs{|\hgS_c(t)| \geq \frac{n\cdot t}{2} } \leq 2 n^{-C}.
\end{align*}
\end{lemma}

\begin{lemma}\label{lemma:denominator_lower_bound}
For any fixed $\gamma \in (0,1)$, any $j \in \gU$ and any $C \geq 1$, if $64 C \log m \leq \lceil \gamma m \rceil$, we have
\begin{align*}
    \sP\LRs{\ermT_{(\lceil \gamma m \rceil)}^{\gU_{-j}} \leq \frac{\gamma}{2}} \leq 2m^{-C}.
\end{align*}
\end{lemma}

\setcounter{equation}{0}
\def\theequation{B.\arabic{equation}}

\section{Proof of the results in Section \ref{sec:conformal-pred}}

\subsection{Proof of Proposition \ref{proposition:BY-FCR-general}}\label{sec:adjusted_con}

\textbf{Proposition \ref{proposition:BY-FCR-general} restated.}
\textit{Suppose that $\{(X_i, Y_i)\}_{i\in \gC \cup \gU}$ are independent and identically distributed random variables and the selection threshold $\htau$ is independent of the calibration set $\gD_c$. Then the $\FCR$ value of the $\FCR$-adjusted method in (\ref{eq:adj_PI}) satisfies $\FCR^{\operatorname{AD}}\leq\alpha$.}

\begin{proof}
For simplicity, we assume $\gD_t$ is fixed. Let $A_{v,r}$ be the event: $r$ prediction intervals are constructed, and $v$ of these do not cover the corresponding true responses. Let ${\text{NPI} }_j$ denote the event that $\{Y_j \not \in \PI^{\text{AD} }(X_j)\}$. By Lemma 1 in \citet{benjamini2005false}, it holds that
\begin{align*}
    \sP\LRs{A_{v,r}} = \frac{1}{v}\sum_{j=1}^m \sP(A_{v,r},\  {\text{NPI} }_j).
\end{align*}
Note that $\cup_{v=1}^r A_{v,r}$ is a disjoint union of events such that $|\hgS_u| = r$ and $|\mathcal{N}_u| = v$, where $\gN_{u}= \{j\in \widehat{\gS}_{u}: \ Y_j \not\in \PI^{\text{AD} }_j\}$. By the definition of $\FCR$, we have
\begin{align}\label{eq:BY-decompose}
    \FCR &= \sum_{r=1}^m\sum_{v=1}^r \frac{v}{r}\sP\LRs{A_{v,r}}\nonumber\\
    &= \sum_{r=1}^m \sum_{v=1}^r\frac{1}{r}\sum_{j=1}^m \sP\LRs{A_{v,r}, {\text{NPI} }_j}\nonumber\\
    &= \sum_{r=1}^m \sum_{j=1}^m\frac{1}{r}\sP\LRs{|\hgS_u| = r,\ {\text{NPI} }_j}.
\end{align}
For each $j\in \hgS_u$, we define the events $\gM_k^{(j)} = \LRl{ M_{\min}^{(j)} = k}$ for $k=1,...,m$.
Recalling the definition of $M_{\min}^{(j)}= \min_{t}\LRl{|\hgS_u^{j\gets t}|:j\in \hgS_u^{j\gets t}}$, we have $M_{\min}^{(j)} \leq |\hgS_u|$ if $j \in \hgS_u$. Following the proof of Theorem 1 in \citet{benjamini2005false} and using the decomposition \eqref{eq:BY-decompose}, we have
\begin{align*}
    \FCR &= \sum_{r=1}^m \sum_{j=1}^m\frac{1}{r}\sum_{l=1}^m\sP\LRs{|\hgS_u| = r,\ j\in \hgS_u, \gM_l^{(j)}, Y_j \not \in \PI^{\text{AD} }_j(X_j) }\\
    &= \sum_{r=1}^m \sum_{j=1}^m\sum_{l=1}^{r} \frac{1}{r}\sP\LRs{|\hgS_u| = r,\ j\in \hgS_u, \gM_l^{(j)}, Y_j \not \in  \PI^{\text{AD} }_j(X_j)}\\
    &\Eqmark{i}{\leq} \sum_{j=1}^m\sum_{r=1}^m\sum_{l=1}^{r} \frac{1}{l}\sP\LRs{|\hgS_u| = r,\ j\in \hgS_u, \gM_l^{(j)}, Y_j \not \in \PI^{\text{AD} }_j(X_j)}\\
    &\Eqmark{ii}{=} \sum_{j=1}^m\sum_{l=1}^m\frac{1}{l}\sum_{r=l}^{m} \sP\LRs{|\hgS_u| = r,\ j\in \hgS_u, \gM_l^{(j)}, Y_j \not \in  \PI^{\text{AD} }_j(X_j)}\\
    &=\sum_{j=1}^m\sum_{k=1}^m\frac{1}{k} \sP\LRs{ \gM_k^{(j)}, Y_j \not \in  \PI^{\text{AD} }_j(X_j)}\\
    &\Eqmark{iii}{=} \sum_{j=1}^m\sum_{k=1}^m\frac{1}{k} \E\LRm{\Indicator{\gM_k^{(j)}}\sP\LRl{Y_j \not \in  \PI^{\text{AD} }_j(X_j) \mid M_{\min}^{(j)}}},\\
    &\Eqmark{iv}{\leq} \sum_{j=1}^m\sum_{k=1}^m\frac{1}{k} \E\LRm{\Indicator{\gM_k^{(j)}} \frac{\alpha M_{\min}^{(j)}}{m}}\\
    &= \sum_{j=1}^m\sum_{k=1}^m\frac{1}{k} \E\LRm{\Indicator{\gM_k^{(j)}} \frac{\alpha k}{m}}\\
    &= \alpha,
\end{align*} 
where $(i)$ holds due to $M_{\min}^{(j)} \leq |\hgS_u|$; $(ii)$ follows from the interchange of summations over $l$ and $r$; $(iii)$ holds since $M_{\min}^{(j)}$ is independent of the calibration set $\gD_c$ and the sample $j$; and $(iv)$ is true because the marginal coverage guarantees that for any $j\in [m]$,
\begin{align}\label{eq:marginal_coverage_prob}
    \sP\LRl{Y_j \not \in  \PI^{\text{AD} }_j(X_j) \mid M_{\min}^{(j)}} \leq \alpha_j^* = \frac{\alpha M_{\min}^{(j)}}{m}.
\end{align}
Here we emphasize that the miscoverage probability $\sP(Y_j \not \in  \PI^{\text{AD}}_j(X_j))$ in $(iv)$ does not depend on the selection condition since the summation is over $[m]$. Consequently, we may utilize the exchangeability between the sample $j$ and the calibration set to verify \eqref{eq:marginal_coverage_prob}.
\end{proof}

\subsection{Proof of Theorem \ref{thm:FCR_exchange}}
\textbf{Theorem \ref{thm:FCR_exchange} restated.}
\textit{Suppose $\{(X_i, Y_i)\}_{i\in \gC \cup \gU}$ are independent and identically distributed random variables, and the threshold $\hat{\tau}$ is exchangeable with respective to the $\{T_i\}_{i\in \gC \cup \gU}$. Then, for each $j\in \gU$, the conditional miscoverage probability is bounded by
\begin{align}
    \sP\LRs{Y_j \not\in \PI_j^{\SCOP}\mid j\in \hgS_u} \leq \alpha.\nonumber
\end{align}
Further, the $\FCR$ value of Algorithm \ref{alg:SCOP} is controlled at $\FCR^{\SCOP} \leq \alpha$.  If $\{R_i\}_{i\in \gC \cup \gU}$ are distinct values almost surely and $\sP(|\hgS_u| > 0) = 1$, we also have $\FCR^{\SCOP} \geq \alpha - \E\{(|\hgS_c| + 1)^{-1}\}$.}

We first introduce two lemmas used in the proof of Theorem \ref{thm:FCR_exchange}, whose proofs are deployed in Section \ref{proof:lemma:zero_gap_exchange_tau_1} and \ref{proof:lemma:zero_gap_exchange_tau_2}, respectively.
\begin{lemma}\label{lemma:zero_gap_exchange_tau_1}
Under the conditions of Theorem \ref{thm:FCR_exchange}, it holds that
\begin{align}
    &\E\LRm{\frac{\Indicator{R_j > \ermR^{\hgS_c \cup \{j\}}_{(\lceil (1 -\alpha)(|\hgS_c|+1)\rceil)},T_k \leq \htau,T_j \leq \htau}}{\sum_{i\in \gC_{+j}} \Indicator{T_i \leq \htau}}}\nonumber\\
    =&\E\LRm{\frac{\Indicator{R_k > \ermR^{\hgS_c \cup \{j\}}_{(\lceil (1 -\alpha)(|\hgS_c|+1)\rceil)},T_k \leq \htau,T_j \leq \htau}}{\sum_{i\in \gC_{+j}} \Indicator{T_i \leq \htau}}}.\nonumber
\end{align}
\end{lemma}

\begin{lemma}\label{lemma:zero_gap_exchange_tau_2}
    Under the conditions of Theorem \ref{thm:FCR_exchange}, for any $j\in \gS_u \subseteq \gU$, it holds that
    \begin{align*}
        &\E\LRm{\frac{\Indicator{R_j > \ermR^{\hgS_c \cup \{j\}}_{(\lceil (1 -\alpha)(|\hgS_c|+1)\rceil)},T_k \leq \htau,\hgS_u = \gS_u}}{\sum_{i\in \gC_{+j}} \Indicator{T_i \leq \htau}}}\\
        =&\E\LRm{\frac{\Indicator{R_k > \ermR^{\hgS_c \cup \{j\}}_{(\lceil (1 -\alpha)(|\hgS_c|+1)\rceil)},T_k \leq \htau,\hgS_u = \gS_u}}{\sum_{i\in \gC_{+j}} \Indicator{T_i \leq \htau}}}.
    \end{align*}
\end{lemma}

\begin{proof}
For any $j\in \hgS_u$, the prediction interval in Algorithm \ref{alg:SCOP} is given by
\begin{align*}
    \PI_j^{\SCOP} = \hmu(X_j) \pm \ermR_{(\lceil (1 -\alpha)(|\gS_c|+1)\rceil)}^{\hgS_c}.
\end{align*}
Applying Lemma \ref{lemma:order_drop_one}(2), we know that
\begin{align}
    \LRl{Y_j \not\in \PI_j^{\SCOP}} = \LRl{R_j > \ermR^{\hgS_c}_{(\lceil (1 -\alpha)(|\hgS_c|+1)\rceil)}} = \LRl{R_j > \ermR^{\hgS_c \cup \{j\}}_{(\lceil (1 -\alpha)(|\hgS_c|+1)\rceil)}}.\nonumber
\end{align}
For simplicity, we write $\ermR^{\hgS_c \cup \{j\}} = \ermR^{\hgS_c \cup \{j\}}_{(\lceil (1 -\alpha)(|\hgS_c|+1)\rceil)}$.
Invoking the upper bound of Lemma \ref{lemma:quantile_inflation} to $\hgS_c \cup \{j\}$, we have
\begin{align}\label{eq:cond_individual_miscover}
    &\sP\LRs{Y_j \not\in \PI_j^{\SCOP} \mid j\in \hgS_u} = \sP\LRs{R_j > \ermR^{\hgS_c \cup \{j\}} \mid T_j \leq \htau}\nonumber\\
    &\qquad\leq \alpha + \E\LRm{\frac{1}{|\hgS_c|+1}\sum_{k\in \hgS_c}\Indicator{R_j > \ermR^{\hgS_c \cup \{j\}}}-\Indicator{R_k > \ermR^{\hgS_c \cup \{j\}}} \mid T_j \leq \htau}\nonumber\\
    &\qquad= \alpha + \sum_{k\in \gC} \E\LRm{\frac{\Indicator{T_k \leq \htau}}{|\hgS_c|+1} \LRs{\Indicator{R_j > \ermR^{\hgS_c \cup \{j\}}}-\Indicator{R_k > \ermR^{\hgS_c \cup \{j\}}}} \mid T_j \leq \htau}\nonumber\\
    &\qquad\Eqmark{i}{=} \alpha + \sum_{k\in \gC}\E\LRm{\frac{\Indicator{R_j > \ermR^{\hgS_c \cup \{j\}},T_k \leq \htau}}{\sum_{i\in \gC_{+j}} \Indicator{T_i \leq \htau}} - \frac{\Indicator{R_k > \ermR^{\hgS_c \cup \{j\}},T_k \leq \htau}}{\sum_{i\in \gC_{+j}} \Indicator{T_i \leq \htau}}\mid T_j \leq \htau}\nonumber\\
    &\qquad= \alpha + \sum_{k\in \gC}\frac{ \E\LRm{\frac{\Indicator{R_j > \ermR^{\hgS_c \cup \{j\}},T_k \leq \htau,T_j \leq \htau}}{\sum_{i\in \gC_{+j}} \Indicator{T_i \leq \htau}}} -\E\LRm{\frac{\Indicator{R_k > \ermR^{\hgS_c \cup \{j\}},T_k \leq \htau,T_j \leq \htau}}{\sum_{i\in \gC_{+j}} \Indicator{T_i \leq \htau}}}}{\sP\LRs{T_j \leq \htau}}\nonumber\\
    &\qquad\Eqmark{ii}{=} \alpha,
\end{align}
where $(i)$ holds due to $\sum_{i \in \gC}\Indicator{T_i \leq \htau} + 1 = \sum_{i\in \gC_{+j}} \Indicator{T_i \leq \htau}$ under the event $\Indicator{T_i \leq \htau}$; and $(ii)$ follows from Lemma \ref{lemma:zero_gap_exchange_tau_1}. 
For any given non-empty subsets $\gS_u \subseteq \gU$ and $\gS_c \subseteq \gC$, and for any $j\in \gS_u$, if it holds that
\begin{equation}\label{eq:coverage_gap}
    \begin{aligned}
    \alpha - \E\LRm{\frac{1}{|\hgS_c|+1} \mid \hgS_u = \gS_u} \leq \sP\LRs{Y_j \not \in \PI_j^{\SCOP}\mid \hgS_u = \gS_u} \leq \alpha,
    \end{aligned}
\end{equation}
Then we can bound the $\FCR$ value as
\begin{align*}
    \FCR &= \E\LRm{\frac{\sum_{j\in \hgS_u}\Indicator{Y_j \not\in \PI_j^{\SCOP}}}{|\hgS_u| \vee 1}}\nonumber\\
    &=\sum_{\gS_u \subseteq \gU}\E\LRm{\frac{\sum_{j\in \hgS_u}\Indicator{Y_j \not\in \PI_j^{\SCOP}}}{|\hgS_u| \vee 1}\mid \hgS_u = \gS_u} \sP\LRs{\hgS_u = \gS_u}\nonumber\\
    &=\sum_{\gS_u \subseteq \gU, \gS_u \neq \emptyset}\frac{1}{|\gS_u|}\sum_{j\in \gS_u}\sP\LRs{Y_j\not\in \PI_j^{\SCOP} \mid \hgS_u = \gS_u} \sP\LRs{\hgS_u = \gS_u}\nonumber\\
    &\leq \alpha \cdot\sum_{\gS_u \subseteq \gU, \gS_u \neq \emptyset} \sP\LRs{\hgS_u = \gS_u} \nonumber\\
    &\leq \alpha,
\end{align*}
where the second last inequality holds due to the right-hand side of \eqref{eq:coverage_gap}. 
Similarly, using the left hand side of \eqref{eq:coverage_gap} we can also obtain that $\FCR \geq \alpha - \E[(|\hgS_c| + 1)^{-1}]$ if $\sP(|\hgS_u| > 0) = 1$.
Therefore, it suffices to verify \eqref{eq:coverage_gap} for Algorithm \ref{alg:SCOP} under the exchangeable assumption.
Similar to \eqref{eq:cond_individual_miscover}, for any $j\in \gS_u$, we have
\begin{align*}
    &\sP\LRs{Y_j \not\in \PI_j^{\SCOP} \mid \hgS_u = \gS_u}\nonumber\\ 
    &\leq \alpha + \sum_{k\in \gC}\frac{\E\LRm{\frac{\Indicator{R_j > \ermR^{\hgS_c \cup \{j\}},T_k \leq \htau,\hgS_u = \gS_u}}{\sum_{i\in \gC_{+j}} \Indicator{T_i \leq \htau}}} -\E\LRm{\frac{\Indicator{R_k > \ermR^{\hgS_c \cup \{j\}},T_k \leq \htau,\hgS_u = \gS_u}}{\sum_{i\in \gC_{+j}} \Indicator{T_i \leq \htau}}}}{\sP\LRs{\hgS_u = \gS_u}}\nonumber\\
    &=\alpha,
\end{align*}
where the last equality holds due to Lemma \ref{lemma:zero_gap_exchange_tau_2}. 
Now using the lower bound of Lemma \ref{lemma:quantile_inflation}, we can also have
\begin{align*}
    \sP\LRs{Y_j \not\in \PI_j^{\SCOP} \mid j\in \hgS_u} \geq \alpha - \E\LRm{\frac{1}{|\hgS_c| + 1} \mid \hgS_u = \gS_u}.
\end{align*}
Therefore, we have verified the relation \eqref{eq:coverage_gap}.
\end{proof}

\subsection{Proof of Lemma \ref{lemma:zero_gap_exchange_tau_1}}\label{proof:lemma:zero_gap_exchange_tau_1}
\begin{proof}
First, we define a new set as
\begin{align*}
    \hgS_{c,+j} = \LRl{i \in \gC_{+j}: T_i \leq \htau}.
\end{align*}
Then $\hgS_c \cup \{j\} = \hgS_{c,+j}$ and $|\hgS_c|+1 = |\hgS_{c,+j}|$ hold under the event $\{T_j \leq \htau\}$. Notice that it suffices to show the following relation
\begin{align}\label{eq:zero_gap_transition_Cj}
    &\E\LRm{\frac{\Indicator{R_j > \ermR^{\hgS_{c,+j}}_{(\lceil (1 -\alpha)|\hgS_{c,+j}|\rceil)},T_j \leq \htau, T_k \leq \htau}}{\sum_{i\in \gC_{+j}} \Indicator{T_i \leq \htau}}}\nonumber\\
    = &\E\LRm{\frac{\Indicator{R_k > \ermR^{\hgS_{c,+j}}_{(\lceil (1 -\alpha)|\hgS_{c,+j}|\rceil)},T_j \leq \htau, T_k \leq \htau}}{\sum_{i\in \gC_{+j}} \Indicator{T_i \leq \htau}}}.
\end{align}
Let $z = [z_1,\ldots,z_{n+1}]$. For any $j \in \gU$, define the event
\begin{align*}
    \mathcal{A}(z; \gD_{\gC_{+j}}) = \LRl{[\{Z_i\}_{i\in \gC_{+j}}] = [z_1,\ldots,z_{n+1}]}.
\end{align*}
Denote the corresponding unordered scores and residuals by $[t_1,\ldots,t_{n+1}]$ and $[r_1,\ldots,r_{n+1}]$ under $\mathcal{A}(z, \gD_{\gC_{+j}})$. Since $\htau$ is exchangeable with respective to $\{T_i\}_{i\in \gC \cup \gU}$, then given $\gD_{\gU_{-j}}$ we can write $\htau \mid \mathcal{A}(z; \gD_{\gC_{+j}}) = \tau(z, \gD_{\gU_{-j}})$, which is fixed and is a function of $z$ and $\gD_{\gU_{-j}}$. It means that the following two unordered sets
\begin{align*}
    \LRm{\LRl{\Indicator{T_i \leq \htau}}_{i\in \gC_{+j}}} \mid \mathcal{A}(z; \gD_{\gC_{+j}}) = \LRm{\LRl{\Indicator{t_i \leq \tau(z, \gD_{\gU_{-j}})}}_{i \in [n+1]}},\\
    \LRm{\LRl{R_i\Indicator{T_i \leq \htau}}_{i\in \gC_{+j}}} \mid \mathcal{A}(z; \gD_{\gC_{+j}}) = \LRm{\LRl{r_i\Indicator{t_i \leq \tau(z, \gD_{\gU_{-j}})}}_{i \in [n+1]}},
\end{align*}
are fixed and depend only on $z$ and the data $\gD_{\gU_{-j}}$. In addition, the quantile $\ermR^{\hgS_{c,+j}}_{(\lceil (1 -\alpha)|\hgS_{c,+j}|\rceil)}$ depends only on $[\LRl{R_i\Indicator{T_i \leq \htau}}_{i\in \gC_{+j}}]$ and $\sum_{i\in \gC_{+j}} \Indicator{T_i \leq \htau}$, hence we write
\begin{align*}
    \ermR^{\hgS_{c,+j}}_{(\lceil (1 -\alpha)|\hgS_{c,+j}|\rceil)} \mid \mathcal{A}(z; \gD_{\gC_{+j}}) = \ermR\LRs{z, \gD_{\gU_{-j}}}.
\end{align*}
From now on, we suppress the dependency of notations on $\gD_{\gU_{-j}}$ on $\tau(z, \gD_{\gU_{-j}})$ and $\ermR\LRs{z, \gD_{\gU_{-j}}}$ when the data $\gD_{\gU_{-j}}$ is given. Given the unordered set $z$, we define the following index subset
\begin{align*}
    \Omega(z) = \LRl{(i_1,i_2) \subseteq [n+1]: r_{i_1} > \ermR(z), t_{i_1} > \tau(z), t_{i_2} \leq \tau(z)},
\end{align*}
which is $\sigma(\gD_{\gU_{-j}})$-measurable and independent of the event $\mathcal{A}(z; \gD_{\gC_{+j}})$. Using the exchangeability of the samples $\{Z_i\}_{i\in \gC_{+j}}$,  we can guarantee
\begin{align}\label{eq:jk_exchange_cond_Az}
    &\sP_{\gU_{-j}}\LRl{R_j > \ermR\LRs{z}, T_j > \tau\LRs{z}, T_k > \tau\LRs{z} \mid \mathcal{A}(z; \gD_{\gC_{+j}})}\nonumber\\
    &= \sum_{(i_1,i_2) \in \Omega(z)} \sP_{\gU_{-j}}\LRl{Z_j = z_{i_1},Z_k = z_{i_2} \mid \mathcal{A}(z; \gD_{\gC_{+j}})}\nonumber\\
    &= \sum_{(i_1,i_2) \in \Omega(z)} \sP_{\gU_{-j}}\LRl{Z_j = z_{i_2},Z_k = z_{i_1} \mid \mathcal{A}(z; \gD_{\gC_{+j}})}\nonumber\\
    &= \sP_{\gU_{-j}}\LRl{R_k > \ermR\LRs{z}, T_j > \tau\LRs{z}, T_k > \tau\LRs{z} \mid \mathcal{A}(z; \gD_{\gC_{+j}})}.
\end{align}
It follows that
\begin{align}
    &\E_{\gU_{-j}}\LRm{\frac{\Indicator{R_j > \ermR^{\hgS_{c,+j}}_{(\lceil (1 -\alpha)|\hgS_{c,+j}|\rceil)},T_j \leq \htau, T_k \leq \htau}}{\sum_{i\in \gC_{+j}} \Indicator{T_i \leq \htau}} \mid \mathcal{A}(z; \gD_{\gC_{+j}})}\nonumber\\
    &\qquad = \E_{\gU_{-j}}\LRm{\frac{\Indicator{R_j > \ermR\LRs{z}, T_j > \tau\LRs{z}, T_k > \tau\LRs{z}}}{\sum_{i\in \gC_{+j}} \Indicator{t_i \leq \tau(z)}} \mid \mathcal{A}(z; \gD_{\gC_{+j}})}\nonumber\\
    &\qquad =  \frac{\sP_{\gU_{-j}}\LRl{R_j > \ermR\LRs{z}, T_j > \tau\LRs{z}, T_k > \tau\LRs{z} \mid \mathcal{A}(z; \gD_{\gC_{+j}})}}{\sum_{i\in \gC_{+j}} \Indicator{t_i \leq \tau(z)}} \nonumber\\
    &\qquad = \frac{\sP_{\gU_{-j}}\LRl{R_k > \ermR\LRs{z}, T_j > \tau\LRs{z}, T_k > \tau\LRs{z} \mid \mathcal{A}(z; \gD_{\gC_{+j}})}}{\sum_{i\in \gC_{+j}} \Indicator{t_i \leq \tau(z)}}\nonumber\\
    &\qquad = \E_{\gU_{-j}}\LRm{\frac{\Indicator{R_k > \ermR^{\hgS_{c,+j}}_{(\lceil (1 -\alpha)|\hgS_{c,+j}|\rceil)},T_j \leq \htau, T_k \leq \htau}}{\sum_{i\in \gC_{+j}} \Indicator{T_i \leq \htau}} \mid \mathcal{A}(z; \gD_{\gC_{+j}})}.\nonumber
\end{align}
Through marginalizing over $\mathcal{A}(z; \gD_{\gC_{+j}})$, together with tower's rule, we can verify \eqref{eq:zero_gap_transition_Cj} immediately.
\end{proof}

\subsection{Proof of Lemma \ref{lemma:zero_gap_exchange_tau_2}}\label{proof:lemma:zero_gap_exchange_tau_2}
\begin{proof}
For any $j \in \gS_u$, we follow the nations defined in Section \ref{proof:lemma:zero_gap_exchange_tau_1}. It suffices to show that for any $j \in \gS_u$,
\begin{align}\label{eq:zero_gap_transition_Cj_Su}
    &\E\LRm{\frac{\Indicator{R_j > \ermR^{\hgS_{c,+j}}_{(\lceil (1 -\alpha)|\hgS_{c,+j}|\rceil)}, T_k \leq \htau, \hgS_u = \gS_u}}{\sum_{i\in \gC_{+j}} \Indicator{T_i \leq \htau}}}\nonumber\\
    = &\E\LRm{\frac{\Indicator{R_k > \ermR^{\hgS_{c,+j}}_{(\lceil (1 -\alpha)|\hgS_{c,+j}|\rceil)},T_k \leq \htau, \hgS_u = \gS_u}}{\sum_{i\in \gC_{+j}} \Indicator{T_i \leq \htau}}}.
\end{align}
Given the data $\gD_{\gU_{-j}}$, we can write the selection event as
\begin{align*}
     &\LRl{\hgS_u = \gS_u, \mathcal{A}(z; \gD_{\gC_{+j}})} = \LRl{\mathcal{A}(z; \gD_{\gC_{+j}}), T_j \leq \htau, \bigcap_{i\in \gS_u \setminus \{j\}} \LRl{T_i \leq \htau}, \bigcap_{i\in \gU\setminus \gS_u } \LRl{T_i > \htau}}\\
     &\qquad = \LRl{\mathcal{A}(z; \gD_{\gC_{+j}}), T_j \leq \tau(z; \gD_{\gU_{-j}}), \bigcap_{i\in \gS_u \setminus \{j\}} \LRl{T_i \leq \tau(z; \gD_{\gU_{-j}})}, \bigcap_{i\in \gU\setminus \gS_u } \LRl{T_i > \tau(z; \gD_{\gU_{-j}})}}\\
     &\qquad = \Big\{\mathcal{A}(z; \gD_{\gC_{+j}}), T_j \leq \tau(z; \gD_{\gU_{-j}}), \gE\LRs{z; \gD_{\gU_{-j}}}\Big\},
\end{align*}
where 
\begin{align*}
    \gE\LRs{z; \gD_{\gU_{-j}}} = \LRs{\bigcap_{i\in \gS_u \setminus \{j\}} \LRl{T_i \leq \tau(z; \gD_{\gU_{-j}})}} \bigcap \LRs{\bigcap_{i\in \gU\setminus \gS_u} \LRl{T_i > \tau(z; \gD_{\gU_{-j}})}}.
\end{align*}
For ease of presentation, we suppress the dependency of notations on $\gD_{\gU_{-j}}$ for now. Then we have
\begin{align}
    &\E_{\gU_{-j}}\LRm{\frac{\Indicator{R_j > \ermR^{\hgS_{c,+j}}_{(\lceil (1 -\alpha)|\hgS_{c,+j}|\rceil)}, T_k \leq \htau, \hgS_u = \gS_u}}{\sum_{i\in \gC_{+j}} \Indicator{T_i \leq \htau}} \mid \mathcal{A}(z; \gD_{\gC_{+j}})}\nonumber\\
    &= \E_{\gU_{-j}}\LRm{\frac{\Indicator{R_j > \ermR(z), T_j \leq \tau(z), T_k \leq \tau(z), \gE\LRs{z}}}{\sum_{i\in \gC_{+j}} \Indicator{t_i \leq \tau(z)}} \mid \mathcal{A}(z; \gD_{\gC_{+j}})}\nonumber\\
    &\Eqmark{i}{=} \frac{\Indicator{\gE\LRs{z}}}{\sum_{i\in \gC_{+j}} \Indicator{t_i \leq \tau(z)}}\cdot \sP_{\gU_{-j}}\LRl{R_j > \ermR(z), T_j \leq \tau(z), T_k \leq \tau(z) \mid \mathcal{A}(z; \gD_{\gC_{+j}})}\nonumber\\
    &\Eqmark{ii}{=} \frac{\Indicator{\gE\LRs{z}}}{\sum_{i\in \gC_{+j}} \Indicator{t_i \leq \tau(z)}}\cdot \sP_{\gU_{-j}}\LRl{R_k > \ermR(z), T_j \leq \tau(z), T_k \leq \tau(z) \mid \mathcal{A}(z; \gD_{\gC_{+j}})}\nonumber\\
    &= \E_{\gU_{-j}}\LRm{\frac{\Indicator{R_k > \ermR^{\hgS_{c,+j}}_{(\lceil (1 -\alpha)|\hgS_{c,+j}|\rceil)}, T_k \leq \htau, \hgS_u = \gS_u}}{\sum_{i\in \gC_{+j}} \Indicator{T_i \leq \htau}} \mid \mathcal{A}(z; \gD_{\gC_{+j}})},\nonumber
\end{align}
where $(i)$ holds $\Indicator{\gE\LRs{z}}$ is $\sigma(\gD_{\gU_{-j}})$-measurable and independent of $\mathcal{A}(z; \gD_{\gC_{+j}})$; and $(ii)$ holds due to \eqref{eq:jk_exchange_cond_Az}. Marginalizing over $\mathcal{A}(z; \gD_{\gC_{+j}})$, together with tower's rule, we can obtain \eqref{eq:zero_gap_transition_Cj_Su} immediately.
\end{proof}

\setcounter{equation}{0}
\def\theequation{C.\arabic{equation}}

\section{Proof of the results in Section \ref{sec:ranking-selection}}
\subsection{Proof of Theorem \ref{thm:finite_control}}
\textbf{Theorem \ref{thm:finite_control} restated.}
\textit{Suppose Assumptions \ref{assum:continuous} and \ref{assum:htau} hold. Under test-driven selection procedures, the $\FCR$ value of Algorithm \ref{alg:SCOP}+ can be controlled at
    \begin{align}
        \alpha - \E\LRm{\LRl{(n+1)\ermT_{(\hkappa+1)}}^{-1}} \leq \FCR^{\SCOP} \leq \alpha.\nonumber
    \end{align}}

In this proof, we define the selected calibration set respective to the threshold $t$ as
\begin{align*}
    \hgS_c(t) = \LRl{i\in \gC: T_i \leq t}.
\end{align*}
Then we can rewrite the $\hgS_c^+$ as
\begin{align}
    \hgS_c^+ \equiv \LRl{i\in \gC: T_i \leq \ermT_{(\hkappa+1)}^{\gU}} \equiv \hgS_c(\ermT_{(\hkappa+1)}^{\gU}).\nonumber
\end{align}
Similarly, the prediction interval can be rewritten as
\begin{align}\label{eq:SCOP_PI_j}
    \PI_j^{\SCOP} &= \hmu(X_j) \pm \ermR^{\hgS_c^+}_{(\lceil (1 -\alpha)(|\hgS_c|+1)\rceil)} \equiv \hmu(X_j) \pm \ermR^{\hgS_c(\ermT_{(\hkappa+1)}^{\gU})}_{(\lceil (1 -\alpha)(|\hgS_c(\ermT_{(\hkappa+1)}^{\gU})|+1)\rceil)}.
\end{align}
To simplify notations, in this subsection, we write $\ermR^{\hgS_c(t) \cup \{j\}} = \ermR^{\hgS_c(t) \cup \{j\}}_{(\lceil (1 -\alpha)(|\hgS_c(t)|+1)\rceil)}$ for any $t$.
The following lemma is needed and its proof can be found in Section \ref{proof:lemma:zero_gap_ranking}.
\begin{lemma}\label{lemma:zero_gap_ranking}
Under the conditions of Theorem \ref{thm:finite_control}, it holds that for each $k\in \gC$
\begin{align}\label{eq:zero_gap}
    &\E_{\gU_{-j}}\LRm{\frac{\delta_{j,k}(\ermT_{(\kappa)}^{\gU_{-j}})}{|\hgS_c(\ermT_{(\kappa)}^{\gU_{-j}})| + 1}\Indicator{T_j\leq \ermT_{(\kappa)}^{\gU_{-j}}, T_k\leq \ermT_{(\kappa)}^{\gU_{-j}}} } = 0,
\end{align}
where
\begin{align*}
    \delta_{j,k}(\ermT_{(\kappa)}^{\gU_{-j}}) = \Indicator{R_j > \ermR^{\hgS_c(\ermT_{(\kappa)}^{\gU_{-j}})\cup \{j\}}} - \Indicator{R_k > \ermR^{\hgS_c(\ermT_{(\kappa)}^{\gU_{-j}})\cup \{j\}}}. 
\end{align*}
\end{lemma}

\begin{proof}
Invoking Lemma \ref{lemma:quantile_inflation}, we have the following relation
\begin{align}\label{eq:quantile_inflation_j_zero}
    \alpha - \frac{1}{|\hgS_c(\ermT_{(\hkappa+1)}^{\gU})|+1}\leq \frac{1}{|\hgS_c(\ermT_{(\hkappa+1)}^{\gU})| + 1}\sum_{i \in \hgS_c(\ermT_{(\hkappa+1)}^{\gU}) \cup \{j\}} \Indicator{R_i > \ermR^{\hgS_c(\ermT_{(\hkappa+1)}^{\gU})\cup \{j\}}} \leq \alpha.
\end{align}
\textbf{Upper bound.}
By \eqref{eq:SCOP_PI_j} and Lemma \ref{lemma:order_drop_one}(2), we have $\LRl{Y_j\not \in \PI_j^{\SCOP}} = \LRl{R_j > \ermR^{\hgS_c(\ermT_{(\hkappa+1)}^{\gU})\cup \{j\}}}$.
Rearranging the inequality in the right-hand side of \eqref{eq:quantile_inflation_j_zero} gives
\begin{align}\label{eq:conformal_miscover_gap_upper}
    \Indicator{Y_j\not \in \PI_j^{\SCOP}}&\leq \alpha + \Indicator{R_j > \ermR^{\hgS_c(\ermT_{(\hkappa+1)}^{\gU})\cup \{j\}}} - \frac{\sum_{i \in \hgS_c(\ermT_{(\hkappa+1)}^{\gU})\cup \{j\}}\Indicator{R_i > \ermR^{\hgS_c(\ermT_{(\hkappa+1)}^{\gU})\cup \{j\}}}}{|\hgS_c(\ermT_{(\hkappa+1)}^{\gU})| + 1}\nonumber\\
    &= \alpha - \frac{\sum_{k \in \hgS_c(\ermT_{(\hkappa+1)}^{\gU})} \LRm{\Indicator{R_k > \ermR^{\hgS_c(\ermT_{(\hkappa+1)}^{\gU})\cup \{j\}}} - \Indicator{R_j > \ermR^{\hgS_c(\ermT_{(\hkappa+1)}^{\gU})\cup \{j\}}}}}{|\hgS_c(\ermT_{(\hkappa+1)}^{\gU})| + 1}\nonumber\\
    &= \alpha - \frac{1}{|\hgS_c(\ermT_{(\hkappa+1)}^{\gU})| + 1}\sum_{k \in \hgS_c(\ermT_{(\hkappa+1)}^{\gU})} \delta_{j,k}(\ermT_{(\hkappa+1)}^{\gU}).
\end{align}
Because the assumption $\hkappa \leq m-1$, we can upper bound $\FCR$ by
\begin{align}
    \FCR &= \sum_{\kappa = 1}^{m-1} \sum_{j\in \gU} \frac{1}{\kappa}\E\LRm{\Indicator{Y_j \not \in \PI_j^{\SCOP}}\Indicator{\hkappa = \kappa} \Indicator{j\in \hgS_u}}\nonumber\\
    &\Eqmark{i}{\leq} \sum_{\kappa = 1}^{m-1} \sum_{j\in \gU} \frac{1}{\kappa}\E\LRm{\Indicator{\hkappa = \kappa} \Indicator{T_j \leq \ermT_{(\kappa)}^{\gU}} \LRl{\alpha - \frac{\sum_{k \in \hgS_c(\ermT_{(\kappa+1)}^{\gU})} \delta_{j,k}(\ermT_{(\kappa+1)}^{\gU})}{|\hgS_c(\ermT_{(\kappa+1)}^{\gU})| + 1}}}\nonumber\\
    &\Eqmark{ii}{=} \alpha - \sum_{\kappa = 1}^{m-1} \sum_{j\in \gU} \frac{1}{\kappa}\E\LRm{\Indicator{\hkappa^{(j)} = \kappa} \Indicator{T_j \leq \ermT_{(\kappa)}^{\gU}} \frac{\sum_{k \in \hgS_c(\ermT_{(\kappa+1)}^{\gU})} \delta_{j,k}(\ermT_{(\kappa+1)}^{\gU})}{|\hgS_c(\ermT_{(\kappa+1)}^{\gU})| + 1}}\nonumber\\
    &\Eqmark{iii}{=} \alpha - \sum_{\kappa = 1}^{m-1} \sum_{j\in \gU} \frac{1}{\kappa}\E\LRm{\Indicator{\hkappa^{(j)} = \kappa} \Indicator{T_j \leq \ermT_{(\kappa)}^{\gU}} \frac{\sum_{k \in \hgS_c(\ermT_{(\kappa)}^{\gU_{-j}})} \delta_{j,k}(\ermT_{(\kappa)}^{\gU_{-j}})}{|\hgS_c(\ermT_{(\kappa)}^{\gU_{-j}})| + 1}}\nonumber\\
    &\Eqmark{iv}{=} \alpha - \sum_{\kappa = 1}^{m-1} \sum_{j\in \gU} \frac{1}{\kappa}\E\LRm{\Indicator{\hkappa^{(j)} = \kappa} \Indicator{T_j \leq \ermT_{(\kappa)}^{\gU_{-j}}} \frac{\sum_{k \in \hgS_c(\ermT_{(\kappa)}^{\gU_{-j}})} \delta_{j,k}(\ermT_{(\kappa)}^{\gU_{-j}})}{|\hgS_c(\ermT_{(\kappa)}^{\gU_{-j}})| + 1}}\nonumber\\
    &= \alpha - \sum_{\kappa = 1}^{m-1} \sum_{j\in \gU} \frac{1}{\kappa}\E\LRm{\Indicator{\hkappa^{(j)} = \kappa} \sum_{k\in \gC}\Indicator{T_j \leq \ermT_{(\kappa)}^{\gU_{-j}}, T_k \leq \ermT_{(\kappa)}^{\gU_{-j}}} \frac{ \delta_{j,k}(\ermT_{(\kappa)}^{\gU_{-j}})}{|\hgS_c(\ermT_{(\kappa)}^{\gU_{-j}})| + 1}}\nonumber\\
    &\Eqmark{v}{=} \alpha\nonumber,
\end{align}
where $(i)$ holds due to \eqref{eq:conformal_miscover_gap_upper}; $(ii)$ holds due to $\hkappa^{(j)} = \hkappa$ (c.f. Assumption \ref{assum:htau}) and $\sum_{j\in \gU}\mathds{1}\{j\in \hgS_u\} = \hkappa$; $(iii)$ follows from $\ermT_{(\kappa+1)}^{\gU} = \ermT_{(\kappa)}^{\gU_{-j}}$ when $T_j \leq \ermT_{(\kappa)}^{\gU}$ by Lemma \ref{lemma:order_drop_one}(1); $(iv)$ is true since $\{T_j \leq \ermT_{(\kappa)}^{\gU}\} = \{T_j \leq \ermT_{(\kappa)}^{\gU_{-j}}\}$ by Lemma \ref{lemma:order_drop_one}(2); and $(v)$ holds since Lemma \ref{lemma:zero_gap_ranking} and $\Indicator{\hkappa^{(j)} = \kappa}$ depends only on $\gD_{\gU_{-j}}$. 
\newline
\textbf{Lower bound.}
Rearranging the inequality in the left-hand side of \eqref{eq:quantile_inflation_j_zero} gives
\begin{align}\label{eq:conformal_miscover_gap_lower}
    \Indicator{j\not \in \PI_j^{\SCOP}} &\geq \alpha - \frac{1}{|\hgS_c(\ermT_{(\hkappa+1)}^{\gU})| + 1} - \frac{1}{|\hgS_c(\ermT_{(\hkappa+1)}^{\gU})| + 1}\sum_{k \in \hgS_c(\ermT_{(\hkappa+1)}^{\gU}) } \delta_{j,k}(\ermT_{(\hkappa+1)}^{\gU}).
\end{align}
Now, using the lower bound \eqref{eq:conformal_miscover_gap_lower} and Lemma \ref{lemma:zero_gap_ranking}, we can have
\begin{align}\label{eq:FCR_lower_expand}
    \FCR &\geq \sum_{\kappa = 1}^{m-1} \sum_{j\in \gU} \frac{1}{\kappa}\E\LRm{\Indicator{\hkappa = \kappa} \Indicator{j\in \hgS_u} \LRl{\alpha - \frac{1 + \sum_{k \in \hgS_c(\ermT_{(\kappa+1)}^{\gU})} \delta_{j,k}(\ermT_{(\kappa+1)}^{\gU})}{|\hgS_c(\ermT_{(\kappa+1)}^{\gU})| + 1} }}\nonumber\\
    &= \alpha - \E\LRm{\frac{1}{|\hgS_c(\ermT_{(\hkappa+1)}^{\gU})|+1}}.
\end{align}
Notice that given the test data $\gD_{u}$, $|\hgS_c(\ermT_{(\hkappa+1)}^{\gU})| = \sum_{i\in \gC} \mathds{1}\{T_i \leq \ermT_{(\hkappa+1)}^{\gU}\} \sim \operatorname{Binomial}(n, p)$, where $p = \ermT_{(\hkappa+1)}^{\gU}$. Hence we have the following result
\begin{align*}
    \E\LRm{\frac{1}{|\hgS_c(\ermT_{(\hkappa+1)}^{\gU})|+1} \mid \gD_u} &= \sum_{s = 0}^n \frac{1}{s+1} \frac{n!}{s!(n-s)!} p^s(1 - p)^{n-s}\\
    &= \sum_{s = 0}^n \frac{1}{(n+1)p} \frac{(n+1)!}{(s+1)!(n-s)!} p^{s+1}(1 - p)^{n-s} \\
    &= \frac{1}{(n+1)p} \LRl{1 - \LRs{1 - p}^{n+1}}.
\end{align*}
Substituting it into \eqref{eq:FCR_lower_expand}, we can get the desired lower bound.
\end{proof}

\subsection{Proof of Lemma \ref{lemma:zero_gap_ranking}} \label{proof:lemma:zero_gap_ranking}
\begin{proof}
Let us first define a new set $\hgS_{c,+j} = \LRl{i\in \gC_{+j}: T_i \leq \ermT_{(\kappa)}^{\gU_{-j}}}$.
Then we know $\hgS_{c,+j} = \hgS_c(\ermT_{(\kappa)}^{\gU_{-j}}) \cup \{j\}$ and $|\hgS_{c,+j}| = |\hgS_c(\ermT_{(\kappa)}^{\gU_{-j}})|+1$ under the event $T_j \leq \ermT_{(\kappa)}^{\gU_{-j}}$. Therefore, it suffices to show
\begin{align}\label{eq:zero_gap_conclusion_trans}
    &\E_{\gU_{-j}}\LRm{\frac{\Indicator{R_j > \ermR_{\lceil (1-\alpha)|\hgS_{c,+j}|\rceil}^{\hgS_{c,+j}}, T_j\leq \ermT_{(\kappa)}^{\gU_{-j}}, T_k\leq \ermT_{(\kappa)}^{\gU_{-j}}}}{|\hgS_{c,+j}|}}\nonumber\\
    =& \E_{\gU_{-j}}\LRm{\frac{\Indicator{R_k > \ermR_{\lceil (1-\alpha)|\hgS_{c,+j}|\rceil}^{\hgS_{c,+j}}, T_j\leq \ermT_{(\kappa)}^{\gU_{-j}}, T_k\leq \ermT_{(\kappa)}^{\gU_{-j}}}}{|\hgS_{c,+j}|}}.
\end{align}
Denote $Z_i = (X_i,Y_i)$ for $i\in \gC \cup \gU$ and define the event
\begin{align*}
    \mathcal{A}(z; \gD_{\gC_{+j}}) = \LRl{[\{Z_i\}_{i\in \gC_{+j}}] = [z_1,\ldots,z_{n+1}]},
\end{align*}
where $z = [z_1,\ldots,z_{n+1}]$.
Given the data $\gD_{\gU_{-j}}$, the quantity $\ermT_{(\kappa)}^{\gU_{-j}}$ is fixed.  Under the event $\mathcal{A}(z; \gD_{\gC_{+j}})$, the following two unordered sets are fixed
\begin{align*}
    &\LRm{\LRl{\Indicator{T_i \leq \ermT_{(\kappa)}^{\gU_{-j}}}}_{i\in \gC_{+j}}} \mid \mathcal{A}(z; \gD_{\gC_{+j}}) = \LRm{\Indicator{t_1 \leq \ermT_{(\kappa)}^{\gU_{-j}}},\ldots,\Indicator{t_{n+1} \leq \ermT_{(\kappa)}^{\gU_{-j}}}},\\
    &\LRm{\LRl{R_i\Indicator{T_i \leq \ermT_{(\kappa)}^{\gU_{-j}}}}_{i\in \gC_{+j}}} \mid \mathcal{A}(z; \gD_{\gC_{+j}}) = \LRm{r_1 \Indicator{t_1 \leq \ermT_{(\kappa)}^{\gU_{-j}}},\ldots,r_{n+1} \Indicator{t_{n+1} \leq \ermT_{(\kappa)}^{\gU_{-j}}}}.
\end{align*}
It means that $|\hgS_{c,+j}|$ and the quantile $\ermR_{(\lceil(1-\alpha)|\hgS_{c,+j}|\rceil)}^{\hgS_{c,+j}}$ are also fixed and are functions of $z$ and $\ermT_{(\kappa)}^{\gU_{-j}}$. Then \eqref{eq:zero_gap_conclusion_trans} can be verified through similar analysis in Section \ref{proof:lemma:zero_gap_exchange_tau_1}.
\end{proof}

\subsection{Proof of Theorem \ref{thm:FCR_bound_cal}}

\textbf{Theorem \ref{thm:FCR_bound_cal} restated.}
 \textit{Suppose Assumptions \ref{assum:continuous}, \ref{assum:decouple_k} and \ref{assum:joint_CDF} hold.
If $n \gamma\geq 256\log n$ and $m\gamma \geq 32\log m$, the $\FCR$ value of Algorithm \ref{alg:SCOP}+ for the calibration-assisted selection procedures can be controlled at
\begin{align}
    \FCR^{\SCOP} - \alpha = O\LRl{\frac{1}{\rho \gamma^3}\LRs{\frac{I_c \log m}{ m} + \frac{\log n}{n}} + \frac{I_u \log m}{m\gamma}}.\nonumber
\end{align}}

\begin{proof}
According to Assumption \ref{assum:htau} such that $\hkappa \leq m-1$, we can upper bound $\FCR$ by
\begin{align}\label{eq:FCR_upper_M_12j}
    \FCR &= \sum_{\kappa = 1}^{m-1} \sum_{j\in \gU} \frac{1}{\kappa}\E\LRm{\Indicator{Y_j \not \in \PI_j^{\SCOP}}\Indicator{\hkappa = \kappa} \Indicator{j\in \hgS_u}}\nonumber\\
    &\leq \alpha - \sum_{\kappa = 1}^{m-1} \sum_{j\in \gU} \frac{1}{\kappa}\E\LRm{\Indicator{\hkappa = \kappa} \Indicator{j\in \hgS_u} \frac{\sum_{k \in \hgS_c(\ermT_{(\kappa+1)}^{\gU})} \delta_{j,k}(\ermT_{(\kappa+1)}^{\gU})}{|\hgS_c(\ermT_{(\kappa+1)}^{\gU})| + 1}}\nonumber\\
    &\Eqmark{i}{=} \alpha - \sum_{\kappa = 1}^{m-1} \sum_{j\in \gU} \frac{1}{\kappa}\E\LRm{\Indicator{\hkappa^{(j)} = \kappa} \Indicator{j\in \hgS_u} \frac{\sum_{k \in \hgS_c(\ermT_{(\kappa+1)}^{\gU})} \delta_{j,k}(\ermT_{(\kappa+1)}^{\gU})}{|\hgS_c(\ermT_{(\kappa+1)}^{\gU})| + 1}}\nonumber\\
    &\Eqmark{ii}{\leq} \alpha - \sum_{\kappa = 1}^{m-1} \sum_{j\in \gU} \frac{1}{\kappa}\E\LRm{\Indicator{\hkappa^{(j)} = \kappa} \Indicator{T_j \leq \ermT_{(\kappa)}^{\gU}} \frac{\sum_{k \in \hgS_c(\ermT_{(\kappa+1)}^{\gU})} \delta_{j,k}(\ermT_{(\kappa+1)}^{\gU})}{|\hgS_c(\ermT_{(\kappa+1)}^{\gU})| + 1}}\nonumber\\
    &\quad + \sum_{\kappa = 1}^{m-1} \sum_{j\in \gU} \frac{1}{\kappa}\E\LRm{\Indicator{\hkappa^{(j)} = \kappa} \Indicator{\ermT_{(\hkappa)}^{\gU} < T_j \leq \ermT_{(\hkappa^{(j)})}^{\gU}} }\nonumber\\
    &\Eqmark{iii}{\leq} \alpha - \sum_{j\in \gU} \underbrace{\sum_{\kappa = 1}^{m-1}\frac{1}{\kappa}\E\LRm{\Indicator{\hkappa^{(j)} = \kappa} \Indicator{T_j \leq \ermT_{(\kappa)}^{\gU_{-j}}} \frac{\sum_{k \in \hgS_c(\ermT_{(\kappa)}^{\gU_{-j}})} \delta_{j,k}(\ermT_{(\kappa)}^{\gU_{-j}})}{|\hgS_c(\ermT_{(\kappa)}^{\gU_{-j}})| + 1}}}_{\gM_{1,j}}\nonumber\\
    &\quad +  \sum_{j\in \gU} \underbrace{\sum_{\kappa = 1}^{m-1}\frac{1}{\kappa}\E\LRm{\Indicator{\hkappa^{(j)} = \kappa} \Indicator{\ermT_{(\hkappa^{(j)}-I_u)}^{\gU} < T_j \leq \ermT_{(\hkappa^{(j)})}^{\gU}} }}_{\gM_{2,j}},
\end{align}
where $(i)$ holds due to Assumption \ref{assum:decouple_k} such that $\hkappa^{(j)} = \hkappa$ if $j \in \hgS_u$; $(ii)$ follows from $|\delta_{j,k}(\ermT_{(\kappa+1)}^{\gU})| \leq 1$ and $\{j\in \hgS_u\} = \{T_j \leq \ermT_{(\hkappa)}^{\gU}\}$; and $(iii)$ holds since $\hkappa^{(j)} \leq \hkappa + I_u$ by Assumption \ref{assum:decouple_k}.
Applying Lemma \ref{lemma:uniform_spacing_and_ratio} gives
\begin{align}
    &\sP\LRs{\ermT_{(\hkappa^{(j)})}^{\gU_{-j}} - \ermT_{(\hkappa^{(j)} - I_u)}^{\gU_{-j}} \geq I_u \cdot \frac{1}{1 - 2\sqrt{\frac{C\log m}{m}}} \frac{2C \log m}{m}}\nonumber\\
    &\qquad \leq \sP\LRs{\max_{0\leq \ell \leq m-1}\LRl{\ermT_{(\ell+1)}^{\gU_{-j}} - \ermT_{(\ell)}^{\gU_{-j}}} \geq \frac{1}{1 - 2\sqrt{\frac{C\log m}{m}}} \frac{2C \log m}{m}} \leq 2m^{-C}.\nonumber
\end{align}
Now taking $C=2$ and using $\hkappa^{(j)}\geq \hkappa \geq \lceil \gamma m \rceil$ in Assumption \ref{assum:decouple_k}, we can bound $\gM_{2,j}$ through
\begin{align}\label{eq:M_2j_upper}
    \gM_{2,j} &= \sum_{\kappa = 1}^{m-1}\frac{1}{\kappa}\E\LRm{\Indicator{\hkappa^{(j)} = \kappa} \Indicator{\ermT_{(\hkappa^{(j)}-I_u)}^{\gU_{-j}} < T_j \leq \ermT_{(\hkappa^{(j)})}^{\gU_{-j}}} }\nonumber\\
    &= \E\LRm{\frac{\ermT_{(\hkappa^{(j,k)})}^{\gU_{-j}} - \ermT_{(\hkappa^{(j)}-I_u)}^{\gU_{-j}}}{\hkappa^{(j)}}}\nonumber\\
    &\leq \frac{8I_u \log m}{m^2 \gamma} + 2m^{-2}\nonumber\\
    &\leq \frac{16 I_u \log m}{m^2 \gamma}.
\end{align}
where we used Lemma \ref{lemma:order_drop_one} since $\kappa \leq m-1$ and the fact $32\log m \leq m$.

In the following proof, for $j\in \gU$ and $k\in \gC$, we will use the following notations: 
\begin{itemize}
    \item $\E_{-j}[\cdot] = \E[\cdot \mid \{Z_i\}_{i \in \gC \cup (\gU \setminus \{j\})}]$;
    \item $\E_{-k}[\cdot]=\E_{\cdot}[\cdot \mid \{Z_i\}_{i \in \gU \cup (\gC \setminus \{k\})}]$;
    \item $\E_{-(j,k)}[\cdot] = \E[\cdot \mid \{Z_i\}_{i\in (\gC \cup \gU) \setminus \{j,k\}}]$;
    \item $\hgS_{c_{-k}}(\kappa) = \{i\in \gC\setminus \{k\}: T_i \leq \ermT_{(\kappa)}^{\gU_{-j}}\}$;
    \item $\hgS_{c}(\kappa)\equiv \hgS_{c}(\ermT_{(\kappa)}^{\gU_{-j}})$ and $\delta_{j,k}(\kappa)\equiv \delta_{j,k}(\ermT_{(\kappa)}^{\gU_{-j}})$.
\end{itemize}
Next we may decompose $\gM_{1,j}$ by
\begin{align}\label{eq:cal-assis-upper}
    \gM_{1,j}
    &\Eqmark{i}{=} \sum_{\kappa = 1}^{m-1}\sum_{k\in \gC} \frac{1}{\kappa}\E\LRm{ \Indicator{\hkappa^{(j,k)} = \kappa} \frac{\delta_{j,k}(\kappa) \Indicator{T_k \leq \ermT_{(\kappa)}^{\gU_{-j}}, T_j\leq \ermT_{(\kappa)}^{\gU_{-j}}}}{|\hgS_{c_{-k}}(\kappa)|+2} }\nonumber\\
    & + \sum_{\kappa = 1}^{m-1}\sum_{k\in \gC} \frac{1}{\kappa}\E\LRm{ \LRs{\Indicator{\hkappa^{(j)} = \kappa}-\Indicator{\hkappa^{(j,k)} = \kappa}} \frac{\delta_{j,k}(\kappa) \Indicator{T_k \leq \ermT_{(\kappa)}^{\gU_{-j}}, T_j\leq \ermT_{(\kappa)}^{\gU_{-j}}}}{|\hgS_{c_{-k}}(\kappa)|+2} }\nonumber\\
    &\Eqmark{ii}{=} \sum_{k\in \gC}\Bigg\{\E\LRm{\frac{1}{\hkappa^{(j,k)}}\frac{\delta_{j,k}(\hkappa^{(j,k)})}{|\hgS_{c_{-k}}(\hkappa^{(j,k)})| + 2}\Indicator{T_k \leq \ermT_{(\hkappa^{(j,k)})}^{\gU_{-j}}, T_j\leq \ermT_{(\hkappa^{(j,k)})}^{\gU_{-j}}} }\nonumber\\
    &\qquad \qquad \qquad - \E\LRm{\frac{1}{\hkappa^{(j)}}\frac{\delta_{j,k}(\hkappa^{(j)})}{|\hgS_{c_{-k}}(\hkappa^{(j)})| + 2}\Indicator{T_k \leq \ermT_{(\hkappa^{(j)})}^{\gU_{-j}}, T_j\leq \ermT_{(\hkappa^{(j)})}^{\gU_{-j}}} }\Bigg\},
\end{align}
where $(i)$ holds due to facts $\hgS_{c_{-k}}(\kappa) \cup \{k\} = \hgS_c(\kappa)$ and $|\hgS_c(\kappa)| = |\hgS_{c_{-k}}(\kappa)|+1$ under the event $T_k \leq \ermT_{(\kappa)}^{\gU_{-j}}$; and $(ii)$ holds because of the exchangeability between $(R_j,T_j)$ and $(R_k,T_k)$ such that
\begin{align}
    &\E_{-(j,k)}\LRm{\delta_{j,k}(\kappa)\Indicator{T_k \leq \ermT_{(\kappa)}^{\gU_{-j}}, T_j\leq \ermT_{(\kappa)}^{\gU_{-j}}}}\nonumber\\
    &\qquad = \E_{-(j,k)}\LRm{\Indicator{T_k \leq \ermT_{(\kappa)}^{\gU_{-j}}, T_j\leq \ermT_{(\kappa)}^{\gU_{-j}}} \Indicator{R_j > \ermR_{(\lceil(1-\alpha)(|\hgS_{c_{-k}}(\kappa)|+2)\rceil)}^{\hgS_{c_{-k}}(\kappa)\cup \{j,k\}}}}\nonumber\\
    &\qquad - \E_{-(j,k)}\LRm{\Indicator{T_k \leq \ermT_{(\kappa)}^{\gU_{-j}}, T_j\leq \ermT_{(\kappa)}^{\gU_{-j}}}  \Indicator{R_k > \ermR_{(\lceil(1-\alpha)(|\hgS_{c_{-k}}(\kappa)|+2)\rceil)}^{\hgS_{c_{-k}}(\kappa)\cup \{j,k\}}} }\nonumber\\
    &\qquad = 0.\nonumber
\end{align}
In fact, the value of $\ermR_{(\lceil(1-\alpha)(|\hgS_{c_{-k}}(\kappa)|+2)\rceil)}^{\hgS_{c_{-k}}(\kappa)\cup \{j,k\}}$ depends on only the unordered set $[Z_k,Z_j]$ once given $\{Z_i\}_{i\in (\gC \cup \gU) \setminus \{j,k\}}$. Therefore, we can verify the relation above by similar arguments in Appendix \ref{proof:lemma:zero_gap_exchange_tau_1}. Further, we can upper bound \eqref{eq:cal-assis-upper} by
\begin{align}\label{eq:M_1j_expansion}
    \gM_{1,j} &\leq \sum_{k\in \gC}\underbrace{\E\LRm{\frac{1}{\hkappa^{(j,k)}} \frac{\delta_{j,k}(\hkappa^{(j,k)})}{|\hgS_{c_{-k}}(\hkappa^{(j,k)})| + 2}\Indicator{T_j\leq \ermT_{(\hkappa^{(j,k)})}^{\gU_{-j}}, \ermT_{(\hkappa^{(j)})}^{\gU_{-j}} < T_k \leq \ermT_{(\hkappa^{(j,k)})}^{\gU_{-j}}} }}_{\Delta_{jk, 1}}\nonumber\\
    & + \sum_{k\in \gC} \underbrace{\E\LRm{\frac{1}{\hkappa^{(j,k)}}\frac{\delta_{j,k}(\hkappa^{(j,k)}) - \delta_{j,k}(\hkappa^{(j)})}{|\hgS_{c_{-k}}(\hkappa^{(j,k)})| + 2}\Indicator{T_k \leq \ermT_{(\hkappa^{(j)})}^{\gU_{-j}}, T_j\leq \ermT_{(\hkappa^{(j,k)})}^{\gU_{-j}}} }}_{\Delta_{jk, 2}}\nonumber\\
    & + \sum_{k\in \gC}\underbrace{\E\LRm{\frac{1}{\hkappa^{(j,k)}}\frac{\delta_{j,k}(\hkappa^{(j)})}{|\hgS_{c_{-k}}(\hkappa^{(j,k)})| + 2} \Indicator{\ermT_{(\hkappa^{(j)})}^{\gU_{-j}} < T_j \leq \ermT_{(\hkappa^{(j,k)})}^{\gU_{-j}}, T_k\leq \ermT_{(\hkappa^{(j)})}^{\gU_{-j}}} }}_{\Delta_{jk, 3}}\nonumber\\
    & + \sum_{k\in \gC}\underbrace{\E\LRm{\LRs{\frac{1}{\hkappa^{(j,k)}} - \frac{1}{\hkappa^{(j)}}} \frac{\delta_{j,k}(\hkappa^{(j)})}{|\hgS_{c_{-k}}(\hkappa^{(j,k)})| + 2}\Indicator{T_k \leq \ermT_{(\hkappa^{(j)})}^{\gU_{-j}}, T_j\leq \ermT_{(\hkappa^{(j)})}^{\gU_{-j}}} }}_{\Delta_{jk,4}}\nonumber\\
    & + \sum_{k\in \gC}\underbrace{\E\LRm{\frac{\delta_{j,k}(\hkappa^{(j)})}{\hkappa^{(j)}}\LRs{\frac{1}{|\hgS_{c_{-k}}(\hkappa^{(j,k)})| + 2} - \frac{1}{|\hgS_{c_{-k}}(\hkappa^{(j)})| + 2}} \Indicator{T_j \leq \ermT_{(\hkappa^{(j)})}^{\gU_{-j}}, T_k\leq \ermT_{(\hkappa^{(j)})}^{\gU_{-j}}} }}_{\Delta_{jk, 5}}.
\end{align}
\textbf{Bound $|\Delta_{jk, 1}|$.} Using $\hkappa^{(j,k)} \leq \hkappa^{(j)} + I_c$ in Assumption \ref{assum:decouple_k}, we have
\begin{align}\label{eq:Delta_jk_1_bound}
    |\Delta_{jk, 1}| &\leq \E\LRm{\frac{1}{\hkappa^{(j,k)}} \frac{1}{|\hgS_{c_{-k}}(\hkappa^{(j,k)})| + 2}\Indicator{\ermT_{(\hkappa^{(j)})}^{\gU_{-j}} < T_k \leq \ermT_{(\hkappa^{(j,k)})}^{\gU_{-j}}, T_j\leq \ermT_{(\hkappa^{(j,k)})}^{\gU_{-j}}} }\nonumber\\
    &\leq \E\LRm{\frac{1}{\hkappa^{(j,k)}} \frac{1}{|\hgS_{c_{-k}}(\hkappa^{(j,k)})| + 2}\Indicator{\ermT_{(\hkappa^{(j,k)} - I_c)}^{\gU_{-j}} < T_k \leq \ermT_{(\hkappa^{(j,k)})}^{\gU_{-j}}, T_j\leq \ermT_{(\hkappa^{(j,k)})}^{\gU_{-j}}} }\nonumber\\
    &= \E\LRm{\frac{1}{\hkappa^{(j,k)}} \frac{1}{|\hgS_{c_{-k}}(\hkappa^{(j,k)})| + 2} \E_{-(j,k)}\LRm{\Indicator{\ermT_{(\hkappa^{(j,k)} - I_c)}^{\gU_{-j}} < T_k \leq \ermT_{(\hkappa^{(j,k)})}^{\gU_{-j}}, T_j\leq \ermT_{(\hkappa^{(j,k)})}^{\gU_{-j}}}} }\nonumber\\
    &= \E\LRm{\frac{\ermT_{(\hkappa^{(j,k)})}^{\gU_{-j}}}{\hkappa^{(j,k)}} \frac{\ermT_{(\hkappa^{(j,k)})}^{\gU_{-j}} - \ermT_{(\hkappa^{(j,k)} - I_c)}^{\gU_{-j}}}{|\hgS_{c_{-k}}(\hkappa^{(j,k)})| + 2} }.
\end{align}
\textbf{Bound $|\Delta_{jk, 2}|$.} We write the difference as
\begin{align*}
    \delta_{j,k}(\hkappa^{(j,k)}) - \delta_{j,k}(\hkappa^{(j)}) &= \Indicator{R_j > \ermR_{(\lceil(1-\alpha)(|\hgS_c(\hkappa^{(j,k)})|+1)\rceil)}^{\hgS_c(\hkappa^{(j,k)}) \cup \{j\}}} - \Indicator{R_j > \ermR_{(\lceil(1-\alpha)(|\hgS_c(\hkappa^{(j)})|+1)\rceil)}^{\hgS_c(\hkappa^{(j)}) \cup \{j\}}}\nonumber\\
    &+ \Indicator{R_k > \ermR_{(\lceil(1-\alpha)(|\hgS_c(\hkappa^{(j)})|+1)\rceil)}^{\hgS_c(\hkappa^{(j)})\cup \{j\}}} - \Indicator{R_k > \ermR_{(\lceil(1-\alpha)(|\hgS_c(\hkappa^{(j,k)})|+1)\rceil)}^{\hgS_c(\hkappa^{(j,k)})\cup \{j\}}}\nonumber\\
    &= \delta^{(j)}(R_j) - \delta^{(j)}(R_k).
\end{align*}
Then we can bound $\Delta_{jk,2}$ by
\begin{align}\label{eq:Delta_2_jk_expand}
    |\Delta_{jk, 2}| &\leq \LRabs{\E\LRm{\frac{1}{\hkappa^{(j,k)}}\frac{\delta^{(j)}(R_j)}{|\hgS_{c_{-k}}(\hkappa^{(j,k)})| + 2}\Indicator{T_k \leq \ermT_{(\hkappa^{(j)})}^{\gU_{-j}}, T_j\leq \ermT_{(\hkappa^{(j,k)})}^{\gU_{-j}}} }}\nonumber\\
    & + \LRabs{\E\LRm{\frac{1}{\hkappa^{(j,k)}}\frac{\delta^{(j)}(R_k)}{|\hgS_{c_{-k}}(\hkappa^{(j,k)})| + 2}\Indicator{T_k \leq \ermT_{(\hkappa^{(j)})}^{\gU_{-j}}, T_j\leq \ermT_{(\hkappa^{(j,k)})}^{\gU_{-j}}} }}\nonumber\\
    &\leq \E\LRm{\frac{1}{\hkappa^{(j,k)}}\frac{\LRabs{\delta^{(j)}(R_j)} }{|\hgS_{c_{-k}}(\hkappa^{(j,k)})| + 2} } + \LRabs{\E\LRm{\frac{1}{\hkappa^{(j,k)}}\frac{\delta^{(j)}(R_k) \Indicator{T_k \leq \ermT_{(\hkappa^{(j)})}^{\gU_{-j}}}}{|\hgS_{c_{-k}}(\hkappa^{(j,k)})| + 2} }}\nonumber\\
    &= \E\LRm{\frac{1}{\hkappa^{(j,k)}}\frac{1 }{|\hgS_{c_{-k}}(\hkappa^{(j,k)})| + 2} \E_{-j}\LRm{\LRabs{\delta^{(j)}(R_j)} } }\nonumber\\
    & + \LRabs{\E\LRm{\frac{1}{\hkappa^{(j,k)}}\frac{1}{|\hgS_{c_{-k}}(\hkappa^{(j,k)})| + 2} \E_{-k}\LRm{\delta^{(j)}(R_k) \Indicator{T_k \leq \ermT_{(\hkappa^{(j)})}^{\gU_{-j}}}} } },
\end{align}
where equality holds since $|\hgS_{c_{-k}}(\hkappa^{(j,k)})|$ is independent of samples $Z_j$ and $Z_k$. From the definition of $\delta^{(j)}(R_j)$, we have
\begin{align}\label{eq:delta_j_expand}
    &\E_{-j}\LRm{\LRabs{\delta^{(j)}(R_j)} }\nonumber\\
    &= \E_{-j} \LRm{\LRabs{\Indicator{R_j > \ermR_{(\lceil(1-\alpha)(|\hgS_c(\hkappa^{(j,k)})|+1)\rceil)}^{\hgS_c(\hkappa^{(j,k)})\cup\{j\}}} - \Indicator{R_j > \ermR_{(\lceil(1-\alpha)(|\hgS_c(\hkappa^{(j)})|+1)\rceil)}^{\hgS_c(\hkappa^{(j)})\cup\{j\}}}}}\nonumber\\
    &= \E_{-j} \LRm{\LRabs{\Indicator{R_j > \ermR_{(\lceil(1-\alpha)(|\hgS_c(\hkappa^{(j,k)})|+1)\rceil)}^{\hgS_c(\hkappa^{(j,k)})}} - \Indicator{R_j > \ermR_{(\lceil(1-\alpha)(|\hgS_c(\hkappa^{(j)})|+1)\rceil)}^{\hgS_c(\hkappa^{(j)})}}}},
\end{align}
where the last equality holds due to Lemma \ref{lemma:order_drop_one} by dropping $j$. Note that $\hgS_c(\hkappa^{(j)}) \subseteq \hgS_c(\hkappa^{(j,k)})$ holds because $\hkappa^{(j,k)} \geq \hkappa^{(j)}$ due to Assumption \ref{assum:decouple_k}. Define
\begin{align*}
     |\hgS_c(\hkappa^{(j,k)})| - |\hgS_c(\hkappa^{(j)})| &= \sum_{i\in \hgS_c(\hkappa^{(j,k)})} \Indicator{\ermT_{(\hkappa^{(j)})}^{\gU_{-j}} < T_i \leq \ermT_{(\hkappa^{(j,k)})}^{\gU_{-j}}}\\
     &\leq \sum_{i\in \hgS_c(\hkappa^{(j,k)})} \Indicator{\ermT_{(\hkappa^{(j,k)} - I_c)}^{\gU_{-j}} < T_i \leq \ermT_{(\hkappa^{(j,k)})}^{\gU_{-j}}}\\
     &=: d^{(j,k)},
\end{align*}
which is independent of samples $Z_j$ and $Z_k$. Next, we bound the rank of the quantile $\ermR_{(\lceil(1-\alpha)(|\hgS_c(\hkappa^{(j)})|+1)\rceil)}^{\hgS_c(\hkappa^{(j)})}$ in the set $\{R_i: i\in \hgS_c(\hkappa^{(j,k)})\}$. Notice that
\begin{align}\label{eq:rank_Sc_lower}
    &\sum_{i\in \hgS_c(\hkappa^{(j,k)})}\Indicator{R_i \leq \ermR_{(\lceil(1-\alpha)(|\hgS_c(\hkappa^{(j)})|+1)\rceil)}^{\hgS_c(\hkappa^{(j)})}}\nonumber\\
    &\qquad\geq \sum_{i\in \hgS_c(\hkappa^{(j)}) }\Indicator{R_i \leq \ermR_{(\lceil(1-\alpha)(|\hgS_c(\hkappa^{(j)})|+1)\rceil)}^{\hgS_c(\hkappa^{(j)})} }\nonumber\\
    &\qquad= \lceil(1-\alpha)(|\hgS_c(\hkappa^{(j)})|+1)\rceil \nonumber\\
    &\qquad\geq \lceil(1-\alpha)(|\hgS_c(\hkappa^{(j,k)})| - d^{(j,k)}+1)\rceil,
\end{align}
and
\begin{align}\label{eq:rank_Sc_upper}
    &\sum_{i\in \hgS_c(\hkappa^{(j,k)})} \Indicator{R_i \leq \ermR_{(\lceil(1-\alpha)(|\hgS_c(\hkappa^{(j)})|+1)\rceil)}^{\hgS_c(\hkappa^{(j)})}}\nonumber\\
    &\qquad\leq \sum_{i\in \hgS_c(\hkappa^{(j)})}\Indicator{R_i \leq \ermR_{(\lceil(1-\alpha)(|\hgS_c(\hkappa^{(j)})|+1)\rceil)}^{\hgS_c(\hkappa^{(j)})}} + d^{(j,k)} \nonumber\\
    &\qquad= \lceil(1-\alpha)(|\hgS_c(\hkappa^{(j)})|+1)\rceil + d^{(j,k)}\nonumber\\
    &\qquad\leq \lceil(1-\alpha)(|\hgS_c(\hkappa^{(j,k)})|+1)\rceil + d^{(j,k)}.
\end{align}
Combining \eqref{eq:rank_Sc_lower} and \eqref{eq:rank_Sc_upper}, we can guarantee
\begin{align}\label{eq:conformal_qunatile_change}
    &\LRabs{\ermR_{(\lceil(1-\alpha)(|\hgS_c(\hkappa^{(j,k)})|+1)\rceil)}^{\hgS_c(\hkappa^{(j,k)})} - \ermR_{(\lceil(1-\alpha)(|\hgS_c(\hkappa^{(j)})|+1)\rceil)}^{\hgS_c(\hkappa^{(j)})}}\nonumber\\
    &\qquad \leq  \max\Bigg\{\ermR_{(\lceil(1-\alpha)(|\hgS_c(\hkappa^{(j,k)})|+1)\rceil)}^{\hgS_c(\hkappa^{(j,k)})} - \ermR_{(\lceil(1-\alpha)(|\hgS_c(\hkappa^{(j,k)})|-d^{(j,k)} +1)\rceil)}^{\hgS_c(\hkappa^{(j,k)})},\nonumber\\
    &\qquad \qquad \qquad \qquad \ermR_{(\lceil(1-\alpha)(|\hgS_c(\hkappa^{(j,k)})|+1)\rceil + d^{(j,k)})}^{\hgS_c(\hkappa^{(j,k)})} - \ermR_{(\lceil(1-\alpha)(|\hgS_c(\hkappa^{(j,k)})|+1)\rceil)}^{\hgS_c(\hkappa^{(j,k)})} \Bigg\}\nonumber\\
    &\qquad \leq \ermR_{(\lceil(1-\alpha)(|\hgS_c(\hkappa^{(j,k)})|+1)\rceil + d^{(j,k)})}^{\hgS_c(\hkappa^{(j,k)})} - \ermR_{(\lceil(1-\alpha)( |\hgS_c(\hkappa^{(j,k)})| - d^{(j,k)} + 1)\rceil)}^{\hgS_c(\hkappa^{(j,k)})}\nonumber\\
    &\qquad = \ermR_{(U^{(j,k)})}^{\hgS_c(\hkappa^{(j,k)})} - \ermR_{(L^{(j,k)})}^{\hgS_c(\hkappa^{(j,k)})},
\end{align}
where $U^{(j,k)} = \lceil(1-\alpha)(|\hgS_c(\hkappa^{(j,k)})|+1)\rceil + d^{(j,k)}$ and $L^{(j,k)} = \lceil(1-\alpha)( |\hgS_c(\hkappa^{(j,k)})| - d^{(j,k)} + 1)\rceil$. Plugging \eqref{eq:conformal_qunatile_change} into \eqref{eq:delta_j_expand}, we can get
\begin{align}\label{eq:delta_j_bound}
    \E_{-j}\LRm{\LRabs{\delta^{(j)}(R_j)}} &\leq \sP_{-j}\LRs{\ermR_{(L^{(j,k)})}^{\hgS_c(\hkappa^{(j,k)})} < T_j \leq \ermR_{(U^{(j,k)})}^{\hgS_c(\hkappa^{(j,k)})}} = \ermR_{(U^{(j,k)})}^{\hgS_c(\hkappa^{(j,k)})} - \ermR_{(L^{(j,k)})}^{\hgS_c(\hkappa^{(j,k)})}.
\end{align}
For $\delta^{(j)}(R_k)$ in \eqref{eq:Delta_2_jk_expand}, using similar arguments, we can still get
\begin{align}\label{eq:delta_k_bound}
    &\LRabs{\E_{-k}\LRm{\delta^{(j)}(R_k) \Indicator{T_k \leq \ermT_{(\hkappa^{(j)})}^{\gU_{-j}}} }}\nonumber\\
    &\leq \E_{-k} \LRm{\Indicator{T_k \leq \ermT_{(\hkappa^{(j)})}^{\gU_{-j}}} \LRabs{\Indicator{R_k > \ermR_{(\lceil (1 - \alpha)(|\hgS_c(\hkappa^{(j,k)})|+1)\rceil)}^{\hgS_c(\hkappa^{(j,k)})\cup\{j\}}} - \Indicator{R_k > \ermR_{(\lceil (1 - \alpha)(|\hgS_c(\hkappa^{(j)})|+1) \rceil)}^{\hgS_c(\hkappa^{(j)})\cup\{j\}}}} }\nonumber\\
    &\Eqmark{i}{=} \E_{-k} \LRm{\Indicator{T_k \leq \ermT_{(\hkappa^{(j)})}^{\gU_{-j}}}\LRabs{\Indicator{R_k > \ermR_{(\lceil (1 - \alpha)(|\hgS_c(\hkappa^{(j,k)})|+1)\rceil)}^{\hgS_{c_{-k}}(\hkappa^{(j,k)})\cup\{j\}}} - \Indicator{R_k > \ermR_{(\lceil (1 - \alpha)(|\hgS_c(\hkappa^{(j)})|+1) \rceil)}^{\hgS_{c_{-k}}(\hkappa^{(j)})\cup\{j\}}}} }\nonumber\\
    &\Eqmark{ii}{\leq} \E_{-k} \LRm{\Indicator{\ermR_{(L^{(j,k)})}^{\hgS_{c_{-k}}(\hkappa^{(j,k)}) \cup\{j\}} < R_k \leq \ermR_{(U^{(j,k)})}^{\hgS_{c_{-k}}(\hkappa^{(j,k)}) \cup\{j\}}} }\nonumber\\
    &\leq \ermR_{(U^{(j,k)})}^{\hgS_{c_{-k}}(\hkappa^{(j,k)})\cup\{j\}} - \ermR_{(L^{(j,k)})}^{\hgS_{c_{-k}}(\hkappa^{(j,k)})\cup\{j\}}\nonumber\\
    &\Eqmark{iii}{\leq} \ermR_{(U^{(j,k)} + 2)}^{\hgS_c(\hkappa^{(j,k)})} - \ermR_{(L^{(j,k)}-2)}^{\hgS_c(\hkappa^{(j,k)})},
\end{align}
where the equality $(i)$ holds because $k\in \hgS_c(\hkappa^{(j)}) \subseteq \hgS_c(\hkappa^{(j,k)})$ under the event $\{T_k \leq \ermT_{(\hkappa^{(j)})}^{\gU_{-j}}\}$, so we can apply Lemma \ref{lemma:order_drop_one}; $(ii)$ holds due to $|\hgS_{c_{-k}}(\hkappa^{(j,k)})| - |\hgS_{c_{-k}}(\hkappa^{(j)})| = |\hgS_{c}(\hkappa^{(j,k)})| - |\hgS_{c}(\hkappa^{(j)})|$ and dropping the $\Indicator{T_k \leq \ermT_{(\hkappa^{(j)})}^{\gU_{-j}}}$; and $(iii)$ follows from Lemma \ref{lemma:order_drop_one} twice with respective to dropping $j$ and adding $k$. Substituting \eqref{eq:delta_j_bound} and \eqref{eq:delta_k_bound} into \eqref{eq:Delta_2_jk_expand} leads
\begin{align}\label{eq:Delta_jk_2_bound}
    |\Delta_{jk, 2}| &\leq \E\LRm{\frac{1}{\hkappa^{(j,k)}}\frac{2 \LRs{\ermR_{(U^{(j,k)}+2)}^{\hgS_c(\hkappa^{(j,k)})} - \ermR_{(L^{(j,k)}-2)}^{\hgS_c(\hkappa^{(j,k)})}}}{|\hgS_{c_{-k}}(\hkappa^{(j,k)})| + 2}}.
\end{align}
\textbf{Bound $|\Delta_{jk, 3}|$.}
Using the assumption $\hkappa^{(j)}\leq \hkappa^{(j,k)} \leq \hkappa^{(j)} + I_c$, we have
\begin{align}\label{eq:Delta_jk_3_bound}
    |\Delta_{jk, 3}| 
    &\Eqmark{i}{=} \left|\E\LRm{\frac{1}{\hkappa^{(j,k)}}\frac{\delta_{j,k}(\hkappa^{(j,k)})}{|\hgS_{c_{-k}}(\hkappa^{(j)})| + 2} \Indicator{\ermT_{(\hkappa^{(j)})}^{\gU_{-j}} < T_j \leq \ermT_{(\hkappa^{(j,k)})}^{\gU_{-j}}, T_k\leq \ermT_{(\hkappa^{(j)})}^{\gU_{-j}}} }\right|\nonumber\\
    &\leq \E\LRm{\frac{1}{\hkappa^{(j,k)}}\frac{\Indicator{ T_k\leq \ermT_{(\hkappa^{(j)})}^{\gU_{-j}}}}{|\hgS_{c_{-k}}(\hkappa^{(j)})| + 2}\sP_{-j}\LRs{\ermT_{(\hkappa^{(j)})}^{\gU_{-j}} < T_j \leq \ermT_{(\hkappa^{(j,k)})}^{\gU_{-j}}} }\nonumber\\
    &\leq \E\LRm{\frac{\ermT_{(\hkappa^{(j,k)})}^{\gU_{-j}} - \ermT_{(\hkappa^{(j)})}^{\gU_{-j}}}{\hkappa^{(j,k)}}\frac{\Indicator{ T_k\leq \ermT_{(\hkappa^{(j)})}^{\gU_{-j}}}}{|\hgS_{c_{-k}}(\hkappa^{(j)})| + 2} }\nonumber\\
    &= \E\LRm{\frac{\ermT_{(\hkappa^{(j,k)})}^{\gU_{-j}} - \ermT_{(\hkappa^{(j)})}^{\gU_{-j}}}{\hkappa^{(j,k)}}\frac{\Indicator{ T_k\leq \ermT_{(\hkappa^{(j)})}^{\gU_{-j}}}}{|\hgS_c(\hkappa^{(j)})| + 1} }\nonumber\\
    &\Eqmark{ii}{\leq} \E\LRm{\frac{\ermT_{(\hkappa^{(j,k)})}^{\gU_{-j}} - \ermT_{(\hkappa^{(j)})}^{\gU_{-j}}}{\gamma m}\frac{\Indicator{ T_k\leq \ermT_{(\hkappa^{(j)})}^{\gU_{-j}}}}{|\hgS_c(\hkappa^{(j)})| + 1} },
\end{align}
where $(i)$ holds due to $\hgS_{c_{-k}}(\hkappa^{(j)}) \cup \{k\} = \hgS_c(\hkappa^{(j)})$ under the event $T_k\leq \ermT_{(\hkappa^{(j)})}^{\gU_{-j}}$; and $(ii)$ follows from $\hkappa^{(j,k)} \geq \hkappa^{(j)} \geq \hkappa \geq \lceil \gamma m \rceil$ in Assumption \ref{assum:decouple_k}.\newline
\textbf{Bound $|\Delta_{jk, 4}|$.}
For this term, we have
\begin{align}\label{eq:Delta_jk_4_bound}
    |\Delta_{jk, 4}| &\leq \E\LRm{\LRs{\frac{1}{\hkappa^{(j)}} - \frac{1}{\hkappa^{(j,k)}} } \frac{1}{|\hgS_{c_{-k}}(\hkappa^{(j)})| + 2}\Indicator{T_k \leq \ermT_{(\hkappa^{(j)})}^{\gU_{-j}}, T_j\leq \ermT_{(\hkappa^{(j)})}^{\gU_{-j}}} }\nonumber\\
    &\leq \E\LRm{\LRs{\frac{1}{\hkappa^{(j)}} - \frac{1}{\hkappa^{(j,k)}} } \frac{1}{|\hgS_{c_{-k}}(\hkappa^{(j)})| + 2}\Indicator{T_k \leq \ermT_{(\hkappa^{(j)})}^{\gU_{-j}}, T_j\leq \ermT_{(\hkappa^{(j,k)})}^{\gU_{-j}}} }\nonumber\\
    &= \E\LRm{\frac{\ermT_{(\hkappa^{(j,k)})}^{\gU_{-j}}}{\hkappa^{(j,k)}} \frac{\hkappa^{(j,k)} - \hkappa^{(j)}}{\hkappa^{(j)}} \frac{\Indicator{T_k \leq \ermT_{(\hkappa^{(j)})}^{\gU_{-j}}} }{|\hgS_{c_{-k}}(\hkappa^{(j)})| + 2} }\nonumber\\
    &= \E\LRm{\frac{\ermT_{(\hkappa^{(j,k)})}^{\gU_{-j}}}{\hkappa^{(j,k)}} \frac{\hkappa^{(j,k)} - \hkappa^{(j)}}{\hkappa^{(j)}} \frac{\Indicator{T_k \leq \ermT_{(\hkappa^{(j)})}^{\gU_{-j}}} }{|\hgS_c(\hkappa^{(j)})| + 1} }\nonumber\\
    &\leq \E\LRm{\frac{\ermT_{(\hkappa^{(j,k)})}^{\gU_{-j}}}{\hkappa^{(j,k)}} \frac{I_c}{\gamma m} \frac{\Indicator{T_k \leq \ermT_{(\hkappa^{(j)})}^{\gU_{-j}}} }{|\hgS_c(\hkappa^{(j)})| + 1} }.
\end{align}
\textbf{Bound $|\Delta_{jk, 5}|$.} For this term, we have
\begin{align}\label{eq:Delta_jk_5_bound}
    |\Delta_{jk, 5}| &\leq \E\LRm{\frac{1}{\hkappa^{(j)}}\LRs{\frac{1}{|\hgS_{c_{-k}}(\hkappa^{(j)})| + 2} - \frac{1}{|\hgS_{c_{-k}}(\hkappa^{(j,k)})| + 2}} \Indicator{T_j \leq \ermT_{(\hkappa^{(j)})}^{\gU_{-j}}, T_k\leq \ermT_{(\hkappa^{(j)})}^{\gU_{-j}}} }\nonumber\\
    &= \E\LRm{\frac{\ermT_{(\hkappa^{(j)})}^{\gU_{-j}}}{\hkappa^{(j)}}\LRs{\frac{1}{|\hgS_{c_{-k}}(\hkappa^{(j)})| + 2} - \frac{1}{|\hgS_{c_{-k}}(\hkappa^{(j,k)})| + 2}} \Indicator{T_k\leq \ermT_{(\hkappa^{(j)})}^{\gU_{-j}}} }\nonumber\\
    &\Eqmark{i}{=} \E\LRm{\frac{\ermT_{(\hkappa^{(j)})}^{\gU_{-j}}}{\hkappa^{(j)}}\LRs{\frac{1}{|\hgS_c(\hkappa^{(j)})| + 1} - \frac{1}{|\hgS_c(\hkappa^{(j,k)})| + 1}} \Indicator{T_k\leq \ermT_{(\hkappa^{(j)})}^{\gU_{-j}}} }\nonumber\\
    &= \E\LRm{\frac{\ermT_{(\hkappa^{(j)})}^{\gU_{-j}}}{\hkappa^{(j)}}\frac{d^{(j,k)}}{|\hgS_c(\hkappa^{(j,k)})| + 1} \frac{\Indicator{T_k\leq \ermT_{(\hkappa^{(j)})}^{\gU_{-j}}}}{|\hgS_c(\hkappa^{(j)})| + 1}  }\nonumber\\
    &\Eqmark{ii}{\leq} \E\LRm{\frac{1}{\gamma m}\frac{d^{(j,k)}}{|\hgS_c(\hkappa^{(j,k)})| + 1} \frac{\Indicator{T_k\leq \ermT_{(\hkappa^{(j)})}^{\gU_{-j}}}}{|\hgS_c(\hkappa^{(j)})| + 1}  },
\end{align}
where $(i)$ holds due to $\hgS_{c_{-k}}(\hkappa^{(j)}) \cup \{k\} = \hgS_c(\hkappa^{(j)})$ under the event $T_k\leq \ermT_{(\hkappa^{(j)})}^{\gU_{-j}}$; and $(ii)$ follows from $\hkappa^{(j)} \geq \hkappa \geq \lceil \gamma m \rceil$ in Assumption \ref{assum:decouple_k}. Plugging \eqref{eq:Delta_jk_1_bound}, \eqref{eq:Delta_jk_2_bound}-\eqref{eq:Delta_jk_5_bound} into \eqref{eq:M_1j_expansion}, we get
\begin{align}\label{eq:M_1j_upper_mid}
    &\gM_{1,j} \leq (\Delta_{jk,1} + \Delta_{jk,2}) + (\Delta_{jk,3} + \Delta_{jk,4} + \Delta_{jk,5})\nonumber\\
    &\leq \sum_{k\in \gC} \E\LRm{\frac{1}{\hkappa^{(j,k)}}\frac{1}{|\hgS_{c_{-k}}(\hkappa^{(j,k)})| + 2} \LRl{\ermT_{(\hkappa^{(j,k)})}^{\gU_{-j}} - \ermT_{(\hkappa^{(j,k)} - I_c)}^{\gU_{-j}} + 2 \LRs{\ermR_{(U^{(j,k)}+2)}^{\hgS_c(\hkappa^{(j,k)})} - \ermR_{(L^{(j,k)}-2)}^{\hgS_c(\hkappa^{(j,k)})}}}}\nonumber\\
    &+ \sum_{k\in \gC} \E\LRm{\frac{\Indicator{k\in \hgS_c(\hkappa^{(j)})}}{|\hgS_c(\hkappa^{(j)})| + 1}\frac{1}{\gamma m} \LRl{\ermT_{(\hkappa^{(j,k)})}^{\gU_{-j}} - \ermT_{(\hkappa^{(j,k)} - I_c)}^{\gU_{-j}} +\frac{\ermT_{(\hkappa^{(j,k)})}^{\gU_{-j}}}{\hkappa^{(j,k)}} + \frac{d^{(j,k)}}{|\hgS_{c_{-k}}(\hkappa^{(j,k)})| + 1}} }.
\end{align}
In the following proof, we assume $n \leq m$. Otherwise, we replace all $\log n$ with $\log m$.
Applying \eqref{eq:max_scaled_order_stat} in Lemma \ref{lemma:uniform_spacing_and_ratio}, we can get
\begin{align}\label{eq:order_stat_ratio_1}
    \sP\LRs{\frac{\ermT_{(\hkappa^{(j,k)})}^{\gU_{-j}}}{\hkappa^{(j,k)}} \geq \frac{1}{1 - 2\sqrt{\frac{C\log m}{m}}} \frac{2C \log m}{m}} \leq 2m^{-C}.
\end{align}
Applying \eqref{eq:max_spacing} in Lemma \ref{lemma:uniform_spacing_and_ratio}, we can get
\begin{align}\label{eq:spacing_Ic}
    \sP\LRs{\ermT_{(\hkappa^{(j,k)})}^{\gU_{-j}} - \ermT_{(\hkappa^{(j,k)} - I_c)}^{\gU_{-j}} \geq I_c \cdot \frac{1}{1 - 2\sqrt{\frac{C\log m}{m}}} \frac{2C \log m}{m}} \leq 2m^{-C}.
\end{align}
Now we take $C = 4$. Lemma \ref{lemma:denominator_lower_bound} and Assumption \ref{assum:decouple_k} guarantee that
\begin{align}\label{eq:good_event_1}
    \ermT_{(\hkappa^{(j,k)})}^{\gU_{-j}} \geq \ermT_{(\hkappa^{(j)})}^{\gU_{-j}}  \geq \ermT_{(\lceil \gamma m \rceil)}^{\gU_{-j}} \geq \frac{\gamma}{2},
\end{align}
holds with probability at least $1 - 2m^{-C}$. By the definitions of $U^{(j,k)}$ and $L^{(j,k)}$, we have
\begin{align}
    &\ermR_{(U^{(j,k)}+2)}^{\hgS_c(\hkappa^{(j,k)})} - \ermR_{(L^{(j,k)}-2)}^{\hgS_c(\hkappa^{(j,k)})}\nonumber\\
    &\qquad = \sum_{\ell = L^{(j,k)} - 2}^{U^{(j,k)} + 2} \ermR_{(\ell+1)}^{\hgS_c(\hkappa^{(j,k)})} - \ermR_{(\ell)}^{\hgS_c(\hkappa^{(j,k)})}\nonumber\\
    &\qquad \leq \LRs{U^{(j,k)} - L^{(j,k)}+4}\max_{1 \leq \ell \leq |\hgS_c(\hkappa^{(j,k)})|-1}\LRl{\ermR_{(\ell+1)}^{\hgS_c(\hkappa^{(j,k)})} - \ermR_{(\ell)}^{\hgS_c(\hkappa^{(j,k)})}}\nonumber\\
    &\qquad \leq 2 (d^{(j,k)}+2)\max_{1 \leq \ell \leq |\hgS_c(\hkappa^{(j,k)})|-1}\LRl{\ermR_{(\ell+1)}^{\hgS_c(\hkappa^{(j,k)})} - \ermR_{(\ell)}^{\hgS_c(\hkappa^{(j,k)})}},\nonumber
\end{align}
where
\begin{align}
    d^{(j,k)} &= \sum_{i \in \hgS_c(\hkappa^{(j,k)})}\Indicator{\ermT_{(\hkappa^{(j,k)} - I_c)}^{\gU_{-j}} < T_i \leq \ermT_{(\hkappa^{(j,k)})}^{\gU_{-j}}}.\nonumber
\end{align}
Recalling the notations $\hgS_c(t) = \{i\in \gC: T_i \leq t\}$ and $\hgS_c^{(j,k)}\equiv \hgS_c(\ermT_{(\hkappa^{(j,k)})}^{\gU_{-j}})$. Let $a_m = \frac{16I_c \log m}{m}$.
With probability at least $1 - 4m^{-4}$ (from \eqref{eq:spacing_Ic} and \eqref{eq:good_event_1}), we can guarantee
\begin{align}\label{eq:R_diff_chaining}
    &\ermR_{(U^{(j,k)})}^{\hgS_c(\hkappa^{(j,k)})} - \ermR_{(L^{(j,k)})}^{\hgS_c(\hkappa^{(j,k)})}\nonumber\\
    &\qquad\leq 2\sup_{t \in [\gamma/2, 1]} \LRl{\LRs{ \sum_{i \in \hgS_c(t)}\Indicator{t - a_m < T_i \leq t} + 2} \cdot \max_{1 \leq \ell \leq |\hgS_c(t)|-1}\LRl{\ermR_{(\ell+1)}^{\hgS_c(t)} - \ermR_{(\ell)}^{\hgS_c(t)}}}\nonumber\\
    &\qquad =: 2 \sup_{t \in [\gamma/2, 1]} \Delta_1(\hgS_c(t)),
\end{align}
which depends only on the calibration set. Similarly, with the same probability, we also have
\begin{align}\label{eq:djk_ratio_chaining}
    \frac{d^{(j,k)}}{|\hgS_c(\hkappa^{(j,k)})| + 1} &\leq \sup_{t \in [\gamma/2, 1]} \frac{\sum_{i \in \hgS_c(t)}\Indicator{t - a_m < T_i \leq t}}{|\hgS_c(t)|+1}=: \sup_{t \in [\gamma/2, 1]} \Delta_2(\hgS_c(t)).
\end{align}
Notice that $|\hgS_c(\hkappa^{(j)})| \geq |\hgS_c(\gamma/2)|$ due to \eqref{eq:good_event_1}.
Plugging \eqref{eq:order_stat_ratio_1}, \eqref{eq:spacing_Ic}, \eqref{eq:R_diff_chaining} and \eqref{eq:djk_ratio_chaining} into \eqref{eq:M_1j_upper_mid}, we have
\begin{align}
    \gM_{1,j} &\Eqmark{i}{\leq} n\cdot \E\LRm{\frac{1}{|\hgS_c(\gamma/2)|+1}\frac{2}{m\gamma} \LRl{\frac{16 I_c \log m}{m} +  \sup_{t \in [\gamma/2, 1]} \Delta_1(\hgS_c(t)) + 10m^{-4}} }\nonumber\\
    &\qquad + \sum_{k\in \gC}\E\LRm{\frac{\Indicator{ T_k\leq \ermT_{(\hkappa^{(j)})}^{\gU_{-j}}}}{|\hgS_c(\gamma/2)| + 1}\frac{1}{\gamma m}\LRl{\frac{32I_c \log m}{m} + \sup_{t \in [\gamma/2, 1]} \Delta_2(\hgS_c(t)) + 10m^{-4}} }\nonumber\\
    &\Eqmark{ii}{\leq} \E\LRm{\frac{4}{m\gamma^2} \LRl{\frac{16 I_c \log m}{m} +  \sup_{t \in [\gamma/2, 1]} \Delta_1(\hgS_c(t)) + 10m^{-4}} }\nonumber\\
    &\qquad + \E\LRm{\frac{1}{\gamma m}\LRl{\frac{32I_c \log m}{m} + \sup_{t \in [\gamma/2, 1]} \Delta_2(\hgS_c(t)) + 10m^{-4}} }\nonumber\\
    &\Eqmark{iii}{\lesssim} \frac{1}{\rho m \gamma^3}\LRs{\frac{I_c \log m}{ m} + \frac{\log n}{n}}.\label{eq:M_1j_upper}
\end{align}
where $(i)$ follows from $|\hgS_{c_{-k}}(\hkappa^{(j,k)})|+1 \geq |\hgS_c(\gamma/2)|$ in \eqref{eq:good_event_1}; $(ii)$ follows from Lemma \ref{lemma:hgSc_size_lower_bound} such that $|\hgS_c(\gamma/2)| \geq n \gamma/2$; and $(iii)$ follows from Lemma \ref{lemma:sup_chaining_bound} and $a_m = 16I_c\log m/m$. Now plugging \eqref{eq:M_1j_upper} and \eqref{eq:M_2j_upper} into \eqref{eq:FCR_upper_M_12j}, we can prove the desired upper bound.

For the lower bound \eqref{eq:conformal_miscover_gap_lower}, similar to \eqref{eq:M_1j_expansion}, we have
\begin{align*}
    \FCR &\geq \alpha + \sum_{j\in \gU}\sum_{k\in \gC}\sum_{\ell = 1}^5\Delta_{jk, \ell}  - \sum_{\kappa = 1}^{m-1} \sum_{j\in \gU} \frac{1}{\kappa}\E\LRm{\Indicator{\hkappa = \kappa} \Indicator{j\in \hgS_u} \frac{1}{|\hgS_c(\ermT_{(\kappa)}^{\gU_{-j}})| + 1}}\nonumber\\
    &\Eqmark{i}{\geq} \alpha + \sum_{j\in \gU}\sum_{k\in \gC} \sum_{\ell = 1}^5\Delta_{jk, \ell} - \sum_{\kappa = 1}^{m-1} \sum_{j\in \gU} \frac{1}{\kappa}\E\LRm{\Indicator{\hkappa^{(j)} = \kappa}  \frac{\Indicator{T_j \leq \ermT_{(\kappa)}^{\gU_{-j} }}}{|\hgS_c(\ermT_{(\kappa)}^{\gU_{-j}})| + 1}}\nonumber\\
    &= \alpha + \sum_{j\in \gU}\sum_{k\in \gC} \sum_{\ell = 1}^5\Delta_{jk, \ell} - \sum_{\kappa = 1}^{m-1} \sum_{j\in \gU} \frac{1}{\kappa}\E\LRm{\Indicator{\hkappa^{(j)} = \kappa} \frac{\ermT_{(\kappa)}^{\gU_{-j} }}{|\hgS_c(\ermT_{(\kappa)}^{\gU_{-j}})| + 1}}\nonumber\\
    &= \alpha + \sum_{j\in \gU}\sum_{k\in \gC} \sum_{\ell = 1}^5\Delta_{jk, \ell} - \sum_{j\in \gU} \E\LRm{\frac{\ermT_{(\hkappa^{(j)})}^{\gU_{-j} }}{\hkappa^{(j)}}\frac{1}{|\hgS_c(\hkappa^{(j)})| + 1}}\nonumber\\
    &\Eqmark{ii}{\geq} \alpha + \sum_{j\in \gU}\sum_{k\in \gC} \sum_{\ell = 1}^5|\Delta_{jk, \ell}| - \frac{1}{\gamma}\E\LRm{ \frac{1}{|\hgS_c(\gamma/2)| + 1}},
\end{align*}
where $(i)$ holds because $\Indicator{T_j \leq \ermT_{(\hkappa)}^{\gU}} = \Indicator{T_j \leq \ermT_{(\hkappa)}^{\gU_{-j}}} \leq \Indicator{T_j \leq \ermT_{(\hkappa^{(j)})}^{\gU_{-j}}}$; and $(ii)$ holds due to $\hkappa^{(j)} \geq \lceil \gamma m\rceil$.  Invoking the upper bounds for $\Delta_{jk,\ell}$ for $\ell = 1,\ldots,5$, we can get the desired lower bound.
\end{proof}

\begin{lemma}\label{lemma:sup_chaining_bound}
    Under the conditions of Theorem \ref{thm:FCR_bound_cal}. If $64C\log n \leq n \gamma$, with probability at least $1 - 8 n^{-4}$, we have
    \begin{align*}
        \sup_{t \in [\gamma/2, 1]} \Delta_1(\hgS_c(t)) \leq \frac{48 \log n}{\rho \gamma} \LRs{a_m + n^{-2} + \frac{48 \log n}{n}},
    \end{align*}
    and
    \begin{align*}
       \sup_{t \in [\gamma/2, 1]} \Delta_2(\hgS_c(t)) \leq \frac{4}{\gamma} \LRs{a_m + n^{-2} + \frac{48 \log n}{n}}.
    \end{align*}
\end{lemma}

\subsection{Proof of Lemma \ref{lemma:sup_chaining_bound}}
\begin{proof}
Now we divide $[\gamma/2,1]$ into equal-length subintervals $[t_{\ell-1}, t_{\ell}]$ for $\ell = 1,...,n^{K}$, where $t_0 = \gamma/2$ and $t_{n^{K}} = 1$. For any $t \in [\gamma/2,1]$, we can find $\ell \leq n^{K}$ such that $t \in [t_{\ell}, t_{\ell+1}]$. Notice that,
\begin{align}\label{eq:d_jk_chaining}
    \sup_{t\in [\gamma/2,1]} &\sum_{i \in \hgS_c(t)}\Indicator{t - a_m < T_i \leq t}\nonumber\\
    &\leq \sup_{t\in [\gamma/2,1]} \sum_{i \in \gC} \Indicator{t - a_m < T_i \leq t}\nonumber\\
    &\leq \max_{0\leq \ell \leq n^K-1} \sum_{i \in \gC}\Indicator{t_{\ell} - a_m < T_i \leq t_{\ell}}\nonumber\\
    &\qquad + \max_{0\leq \ell \leq n^K-1}\sup_{t\in [t_{\ell}, t_{\ell+1}]} \sum_{i \in \gC}\LRs{\Indicator{t - a_m < T_i \leq t} - \Indicator{t_{\ell} - a_m < T_i \leq t_{\ell}}}\nonumber\\
    &= \max_{0\leq \ell \leq n^K-1} \sum_{i \in \gC} \Indicator{t_{\ell} - a_m < T_i \leq t_{\ell}}\nonumber\\
    &\qquad + \max_{0\leq \ell \leq n^K-1}\sup_{t\in [t_{\ell}, t_{\ell+1}]} \sum_{i \in \gC}\LRs{\Indicator{t_{\ell} < T_i \leq t} - \Indicator{t_{\ell} - a_m < T_i \leq t - a_m}}\nonumber\\
    &\leq \max_{0\leq \ell \leq n^K-1} \LRl{\sum_{i \in \gC}\Indicator{t_{\ell} - a_m < T_i \leq t_{\ell}} + \sum_{i \in \gC} \Indicator{t_{\ell} < T_i \leq t_{\ell+1}}}.
\end{align}
Applying Lemma \ref{lemma:Sc_interval_concentration}, we have
\begin{align}\label{eq:djk_chaining_bound_1}
    \sP\LRs{\frac{1}{n}\sum_{i \in \gC} \Indicator{t_{\ell} - a_m < T_i \leq t_{\ell}} \leq 2a_m + \frac{12C \log n}{n}} \geq 1 - 2n^{-C},
\end{align}
and
\begin{align}\label{eq:djk_chaining_bound_2}
    \sP\LRs{\frac{1}{n}\sum_{i \in \gC} \Indicator{t_{\ell} < T_i \leq t_{\ell+1}} \leq 2n^{-K} + \frac{12C \log n}{n}} \geq 1 - 2n^{-C}.
\end{align}
Plugging \eqref{eq:djk_chaining_bound_1} and \eqref{eq:djk_chaining_bound_2} into \eqref{eq:d_jk_chaining}, with probability at least $1 - 4n^{-C+K}$, we can guarantee that
\begin{align}\label{eq:djk_chaining_bound}
   \sup_{t\in [\gamma/2,1]} \sum_{i \in \gC}\Indicator{t - a_m < T_i \leq t} \leq n\cdot a_m + n^{-K+1} + 12C \log n.
\end{align}
Applying Lemma \ref{lemma:Sc_max_spacing} and \ref{lemma:hgSc_size_lower_bound} with $t = \gamma/2$, with probability at least $1 - 4n^{-C}$, we can guarantee that
\begin{align}\label{eq:R_diff_chaining_bound}
    \sup_{t\in [\gamma/2,1]}\max_{1 \leq \ell \leq |\hgS_c(t)|-1}\LRl{\ermR_{(\ell+1)}^{\hgS_c(t)} - \ermR_{(\ell)}^{\hgS_c(t)}} &\leq \max_{1 \leq \ell \leq |\hgS_c(\gamma/2)|-1}\LRl{\ermR_{(\ell+1)}^{\hgS_c(\gamma/2)} - \ermR_{(\ell)}^{\hgS_c(\gamma/2)}}\nonumber\\
    &\leq \frac{1}{\rho} \frac{1}{1 - 2\sqrt{\frac{C \log n}{|\hgS_c(\gamma/2)| +1}}} \frac{2C \log n}{|\hgS_c(\gamma/2)| + 1}\nonumber\\
    &\leq \frac{1}{\rho} \frac{1}{1 - 2\sqrt{\frac{C \log n}{n \gamma/4 +1}}} \frac{2C \log n}{n \gamma/4 + 1},
\end{align}
where the first inequality holds since $\hgS_c(\gamma/2) \subseteq \hgS_c(t)$ for any $t \geq \gamma/2$. In fact, we also used the fact that the maximum spacing of $\ermR^{\gS_1}$ is always no smaller than that of $\ermR^{\gS_2}$ if $\gS_1 \subseteq \gS_2$. Choosing $C = 8$ and $K = 2$, and combining \eqref{eq:djk_chaining_bound} and \eqref{eq:R_diff_chaining_bound}, we have
\begin{align*}
  \sup_{t\in [\gamma/2,1]}\Delta_1(\hgS_c(t)) &\leq  \frac{96 \log n}{\rho \gamma} \LRs{a_m + n^{-2} + \frac{48 \log n}{n}},
\end{align*}
holds with probability at least $1 - 8n^{-4}$. 

From the definition of $\Delta_2(\hgS_c(t))$, with probability at least $1 - 8n^{-4}$, we have
\begin{align*}
    \sup_{t\in [\gamma/2,1]} \Delta_2(\hgS_c(t)) &=  \sup_{t\in [\gamma/2,1]} \frac{\sum_{i \in \hgS_c(t)}\Indicator{t - a_m < T_i \leq t}}{|\hgS_c(t)|+1}\\
    &\leq \frac{1}{|\hgS_c(\gamma/2)|+1} \sup_{t\in [\gamma/2,1]} \sum_{i \in \gC}\Indicator{t - a_m < T_i \leq t}\\
    &\leq \frac{8}{\gamma} \LRs{a_m + n^{-2} + \frac{48 \log n}{n}},
\end{align*}
where we used Lemma \ref{lemma:hgSc_size_lower_bound} and \eqref{eq:djk_chaining_bound}.

\end{proof}

\setcounter{equation}{0}
\def\theequation{D.\arabic{equation}}

\section{Proof of Auxiliary Lemmas in Section \ref{appen:aux_lemmas}}

\subsection{Proof of Lemma \ref{lemma:uniform_spacing_and_ratio}}\label{proof:lemma:uniform_spacing_and_ratio}
\begin{lemma}[Representation of spacing of uniform order statistics \citep{arnold2008first}]\label{lemma:spacing}
Let $U_1,\cdots ,U_n \stackrel{i.i.d.}{\sim}\unif([0,1])$, and $U_{(1)}\leq U_{(2)}\leq \cdots \leq U_{(n)}$ be their order statistics. Then
\begin{align*}
    \LRs{U_{(1)} - U_{(0)},\cdots,U_{(n+1)} - U_{(n)}} \stackrel{d}{=} \LRs{\frac{V_1}{\sum_{k=1}^{n+1} V_k},\cdots,\frac{V_{n+1}}{\sum_{k=1}^{n+1} V_k}},
\end{align*}
where $U_{0}=0$, $U_{(n+1)} = 1$, and $V_1,\cdots,V_{n+1}\stackrel{i.i.d.}{\sim} \operatorname{Exp}(1)$.
\end{lemma}

\begin{lemma}[Quantile transformation of order statistics, Theorem 1.2.5 in \citet{reiss2012approximate}]\label{lemma:quantile_transform}
Let $X_1,\cdots, X_n$ be i.i.d. random variables with cumulative distribution function $F(\cdot)$, and $U_1,\cdots ,U_n \stackrel{i.i.d.}{\sim}\unif([0,1])$, then
\begin{align*}
    \LRs{F^{-1}(U_{(1)}),\cdots, F^{-1}(U_{(n)})} \stackrel{d}{=} \LRs{X_{(1)},\cdots,X_{(n)}}.
\end{align*}
\end{lemma}

\textbf{Fact.} For the random variable $X \sim \operatorname{Exp}(\lambda)$, it holds that $\sP\LRs{X \geq x} = e^{-\lambda x}$.

\textbf{Fact.} For the random variable $X \sim \chi^2_{\nu}$, it holds that $\sP\LRs{X - \nu \geq x} \leq e^{-\nu x^2/8}$.

\begin{proof}
We first prove the conclusion \eqref{eq:max_spacing}. Using the spacing representation in Lemma \ref{lemma:spacing}, we have
\begin{align}\label{eq:tail_R_spacing_1}
    &\sP_{\gD_{u}}\LRs{\max_{0\leq \ell \leq n}\LRl{U_{(\ell+1)} - U_{(\ell)}} \geq \frac{1}{1 - 2\sqrt{\frac{C \log n}{n+1}}}\frac{2 C \log n}{n+1}}\nonumber\\
    &\qquad = \sP\LRs{\max_{0\leq \ell \leq n} \frac{V_{\ell}}{\sum_{i=1}^{n+1} V_i}  \geq \frac{1}{1 - 2\sqrt{\frac{C \log n}{n+1}}}\frac{2C \log n}{n+1}}\nonumber\\
    &\qquad \leq \sP\LRs{\max_{0\leq \ell \leq n} V_{\ell} \geq 2C \log n} + \sP\LRs{\frac{1}{n+1} \sum_{i=1}^{n+1} V_{i} \leq 1 - 2\sqrt{\frac{C \log n}{n+1}}},
\end{align}
where $U_{(0)} = 0$ and $U_{(n+1)} = 1$.
By the tail probability of $\operatorname{Exp}(1)$ and union bound, we know
\begin{align}\label{eq:tail_max_ell}
    \sP\LRs{\max_{0\leq \ell \leq n} V_{\ell} \geq 2C \log n} &\leq \LRs{n+1}\sP\LRs{V_1 \geq 2C \log n}
    \leq n^{-C},
\end{align}
holds for any $C \geq 1$. In addition, we know that $\frac{1}{2}\sum_{i=1}^{n+1} V_{i} \sim \Gamma(n+1, 2)\stackrel{d}{=} \chi^2_{2(n+1)}$. 
Using the tail bound of $\chi^2$ distribution with $\nu = 2(n+1)$, we have
\begin{align}\label{eq:tail_average}
    \sP\LRs{\frac{1}{n+1} \sum_{i=1}^{n+1} V_{i} \leq 1 - 2\sqrt{\frac{C \log n}{n+1}}} \leq n^{-C}.
\end{align}
Substituting \eqref{eq:tail_max_ell} and \eqref{eq:tail_average} into \eqref{eq:tail_R_spacing_1} gives the desired result. For the second conclusion \eqref{eq:max_scaled_order_stat}, we have
\begin{align*}
    &\sP\LRs{\max_{1\leq k\leq n}\frac{U_{(k)}}{k} \geq \frac{1}{1 - 2\sqrt{\frac{C \log n}{n+1}}}\frac{C \log n}{n+1}}\\
    &\qquad = \sP\LRs{\max_{1\leq k\leq n} \frac{\frac{1}{k} \sum_{\ell=1}^k V_{\ell}}{\sum_{i = 1}^{n+1} V_{i}} \geq \frac{1}{1 - 2\sqrt{\frac{C \log n}{n+1}}}\frac{2C \log n}{n+1}} \\
    &\qquad \leq \sP\LRs{\max_{0\leq \ell \leq n} \frac{V_{\ell}}{\sum_{i=1}^{n+1} V_i}  \geq \frac{1}{1 - 2\sqrt{\frac{C \log n}{n+1}}}\frac{2C \log n}{n+1}}\\
    &\qquad \leq 2n^{-C}.
\end{align*}
\end{proof}

\subsection{Proof of Lemma \ref{lemma:Sc_max_spacing}}\label{proof:lemma:Sc_max_spacing}
\begin{proof}
Notice that, for any $\gS_c \subseteq \gC$, we can write
\begin{align}\label{eq:hSc_event}
    \LRl{\hgS_c(t) = \gS_c} = \LRs{\bigcap_{i\in \gS_c}\LRl{T_i \leq t}} \bigcap \LRs{\bigcap_{i\in \gC\setminus \gS_c}\LRl{T_i > t}}.
\end{align}
Then for any $i\in \gS_c$, we have
\begin{align*}
     \sP\LRs{R_i \leq r \mid \hgS_c(t) = \gS_c} &= \sP\LRs{R_i \leq r \mid T_i \leq t} = \frac{F_{(R, T)}(r, t)}{t} = G(r).
\end{align*}
Hence, given $\hgS_c(t) = \gS_c$, $\{R_i\}_{i \in \gS_{c}}$ are independent and identically distributed random variables with the common cumulative distribution function $G(\cdot)$. Applying Lemma \ref{lemma:quantile_transform}, we know there exist $U_i \stackrel{i.i.d.}{\sim} \unif([0,1])$ for $i\in \gS_c$ such that
\begin{align}\label{eq:G_transform}
    \LRs{\ermR_{(1)}^{\gS_c},\cdots,\ermR_{(|\gS_c|)}^{\gS_c}\mid \hgS_c(t) = \gS_c} \stackrel{d}{=} \LRs{G^{-1}(U_{(1)},\cdots,G^{-1}(U_{(|\gS_c|)}}.
\end{align} 
Let $G^{-1}(\cdot)$ be the inverse function of $G(\cdot)$, and use our assumption on $\frac{\partial}{\partial r}F_{(R,T)}(r, t) \geq \rho t $, we can get
\begin{align}\label{eq:G_inverse_grad_upper_bound}
    \frac{d}{d r}G^{-1}(r) = \LRs{\frac{d}{d r}G(r)}^{-1} = \frac{t}{\frac{\partial}{\partial r} F_{(R,T)}(r, t)} \leq \frac{1}{\rho}.
\end{align}
Then for any $x \geq 0$, we have
\begin{align*}
    &\sP\LRs{\max_{0\leq \ell \leq |\hgS_c(t)|-1}\LRl{\ermR_{(\ell+1)}^{\hgS_c(t)} - \ermR_{(\ell+1)}^{\hgS_c(t)}} \geq \frac{x}{\rho}\mid \hgS_c(t) = \gS_c}\nonumber\\
    &\qquad \Eqmark{i}{=} \sP\LRs{\max_{0\leq \ell \leq |\gS_c|-1}\LRl{\ermR_{(\ell+1)}^{\gS_c} - \ermR_{(\ell)}^{\gS_c}} \geq \frac{x}{\rho} \mid \bigcap_{i\in \gS_c}\LRl{T_i \leq t}}\nonumber\\
    &\qquad \Eqmark{ii}{=} \sP\LRs{\max_{0\leq \ell \leq |\gS_c|-1}\LRl{G^{-1}(U_{(\ell+1)}) - G^{-1}(U_{(\ell)})} \geq \frac{x}{\rho}}\nonumber\\
    &\qquad \Eqmark{iii}{\leq} \sP\LRs{\max_{0\leq \ell \leq |\gS_c|-1}\LRl{U_{(\ell+1)} - U_{(\ell)}} \geq x},
\end{align*}
where $(i)$ holds due to \eqref{eq:hSc_event} and the fact that $\{T_i\}_{i\in \gS_c}$ are independent of $\{T_i\}_{i\in \gC\setminus\gS_c}$, $(ii)$ follows from \eqref{eq:G_transform}, and $(iii)$ comes from \eqref{eq:G_inverse_grad_upper_bound}. Invoking Lemma \ref{lemma:uniform_spacing_and_ratio}, we can finish the proof.
\end{proof}

\subsection{Proof of Lemma \ref{lemma:Sc_interval_concentration}}\label{proof:lemma:Sc_interval_concentration}
\begin{proof}
    Recall the definition $Z_i = \Indicator{t_1 < T_i \leq t_2} - (t_2 - t_1)$, then we know $\E[Z_i] = 0$.
    Next, we will bound the moment generating function of $\sum_{i\in \gS_c}Z_i$. For any $\lambda > 0$, we have
    \begin{align}\label{eq:MGF_sum_Z}
        \E\LRm{e^{\lambda \sum_{i\in \gC} Z_i}} = \E\LRm{\prod_{i\in \gC} e^{\lambda Z_i} } &= \prod_{i\in \gC}\E\LRm{ e^{\lambda Z_i}}\nonumber\\
        &\Eqmark{i}{\leq} \prod_{i\in \gC}\LRs{1 + \lambda^2 \E\LRm{Z_i^2 e^{\lambda |Z_i|}}}\nonumber\\
        &\Eqmark{ii}{\leq} \exp\LRs{\lambda^2\sum_{i\in \gC} \E\LRm{Z_i^2 e^{\lambda |Z_i|}}},
    \end{align}
    where $(i)$ holds due to the basic inequality $e^y - 1 - y \leq y^2 e^{|y|}$, $(ii)$ comes from the basic inequality $1+y \leq e^y$. Notice that
    \begin{align}\label{eq:MGF_with_variance}
        \sum_{i\in \gC}\E\LRm{Z_i^2 e^{\lambda |Z_i|}} &= n \E\LRm{Z_1^2 e^{\lambda |Z_1|}}\nonumber\\
        &\leq e^{\lambda}\cdot n \E\LRm{\LRs{\Indicator{t_1 < T_1 \leq t_2} - (t_2 - t_1)}^2 }\nonumber\\
        &\leq e^{\lambda}\cdot n (t_2 - t_1)\nonumber\\
        &= K^2(\lambda).
    \end{align}
    Using Markov's inequality and \eqref{eq:MGF_sum_Z}, for any $z \geq 0$, we can get
    \begin{align}
        \sP\LRs{\sum_{i\in \gC} Z_i \geq 2K(1) z}
        = \sP\LRs{e^{\lambda \sum_{i\in \gC} Z_i} \geq e^{2 \lambda K(1) z} } &\leq e^{-2 \lambda K(1) z}\E\LRm{e^{\lambda \sum_{i\in \gC} Z_i}}\nonumber\\
        &\leq \exp\LRl{-2 \lambda  K(1) z + \lambda^2 K^2(\lambda)},\nonumber
    \end{align}
    where the last inequality comes from \eqref{eq:MGF_sum_Z} and \eqref{eq:MGF_with_variance}. If $n \geq \frac{C\log n}{t_2 - t_1}$, we know that $K^2(1) \geq e C \log n$. Taking $z = (C \log n)^{1/2}$, $\lambda = z/K(1) \leq 1$, then we have
    \begin{align}\label{eq:exponential_ineq_Xi_1}
        \sP\LRs{\LRabs{\sum_{i\in \gC} Z_i} \geq 2K(1) z } &\leq 2\exp\LRl{-2z^2 + z^2 \frac{K^2(\lambda)}{K^2(1)}}\nonumber\\
        &\leq 2e^{-z^2} = 2n^{-C}.
    \end{align}
    If $n < \frac{C\log n}{t_2 - t_1}$ such that $K^2(1) < e C \log n$, applying \eqref{eq:MGF_with_variance} with $\lambda = 1$ gives
    \begin{align}\label{eq:exponential_ineq_Xi_2}
        \sP\LRs{\LRabs{\sum_{i\in \gC} Z_i} \geq 2e C\log n}
        &\leq 2e^{- 2e C\log n}\E\LRm{e^{\sum_{i\in \gC} Z_i}} \leq 2e^{- 2e C\log n + K^2(1)}\nonumber\\
        &\leq 2e^{- 2eC\log n + eC \log n} \leq 2n^{-C}.
    \end{align}
    Combining \eqref{eq:exponential_ineq_Xi_1} and \eqref{eq:exponential_ineq_Xi_2}, we can get
    \begin{align*}
        \sP\LRs{\frac{1}{n}\LRabs{\sum_{i\in \gC} Z_i} \geq  2\sqrt{\frac{e C \log n}{n} \cdot (t_2-t_1)} + \frac{2e C\log n}{n}}\leq 2n^{-C}.
    \end{align*}
    Thus we have completed the proof.
\end{proof}

\subsection{Proof of Lemma \ref{lemma:hgSc_size_lower_bound}}\label{proof:lemma:hgSc_size_lower_bound}
\begin{proof}
By the definition of $\hgS_c(t)$, we have
\begin{align*}
    |\hgS_c(t)| - nt = \sum_{i\in \gC}\Indicator{T_i\leq t} - nt = \sum_{i\in \gC}\Indicator{T_i\leq t} - \sP\LRs{T_i\leq t}.
\end{align*}
Applying Hoeffding's inequality, we have
\begin{align*}
    \sP\LRs{\LRabs{|\hgS_c(t)| - nt} \geq 2C\sqrt{n \log n}} \leq n^{-C}.
\end{align*}
Using the assumption $8C \log n / (nt) \leq 1$, we can finish the proof.
\end{proof}

\subsection{Proof of Lemma \ref{lemma:denominator_lower_bound}}\label{proof:lemma:denominator_lower_bound}
\begin{proof}
Invoking the spacing representation in Lemma \ref{lemma:spacing}, we have
\begin{align*}
    \sP\LRs{\ermT_{(\lceil \gamma m \rceil)}^{\gU_{-j}} \leq \frac{\gamma}{2}} &= \sP\LRs{\frac{\sum_{i = 1}^{\lceil \gamma m \rceil} V_i}{\sum_{k=1}^{m} V_i} \leq \frac{\gamma}{2}}\\
    &= \sP\LRs{\frac{\frac{1}{\lceil \gamma m \rceil}\sum_{i = 1}^{\lceil \gamma m \rceil} V_i}{\frac{1}{m}\sum_{k=1}^{m} V_k} \leq \frac{\gamma}{2}\frac{m}{\lceil \gamma m \rceil}}\\
    &\leq \sP\LRs{\frac{\frac{1}{\lceil \gamma m \rceil}\sum_{i = 1}^{\lceil \gamma m \rceil} V_i}{\frac{1}{m}\sum_{k=1}^{m} V_k} \leq \frac{1}{2}}\\
    &\leq \sP\LRs{\frac{1}{\lceil \gamma m \rceil} \sum_{i = 1}^{\lceil \gamma m \rceil} (V_i - 1) \leq -\frac{1}{4}} + \sP\LRs{\frac{1}{m} \sum_{k=1}^{m} (V_k-1) \geq \frac{1}{2}}\\
    &\Eqmark{i}{\leq} \sP\LRs{\frac{1}{\lceil \gamma m \rceil} \sum_{i = 1}^{\lceil \gamma m \rceil} (V_i - 1) \leq - 2\sqrt{\frac{C \log m}{\lceil \gamma m \rceil}}} + \sP\LRs{\frac{1}{m} \sum_{k=1}^{m} (V_k-1) \geq 2\sqrt{\frac{C \log m}{m}}}\\
    &\Eqmark{ii}{\leq} 2 m^{-C},
\end{align*}
where $(i)$ holds due to the assumption $64C \log m \leq \lceil \gamma m \rceil$; and $(ii)$ follows from the tail probability of Chi-square distribution.
\end{proof}

\setcounter{equation}{0}
\def\theequation{E.\arabic{equation}}
\section{Additional details on the algorithms}

\begin{figure}[htbp]
\centering
\includegraphics[width=0.9\textwidth]{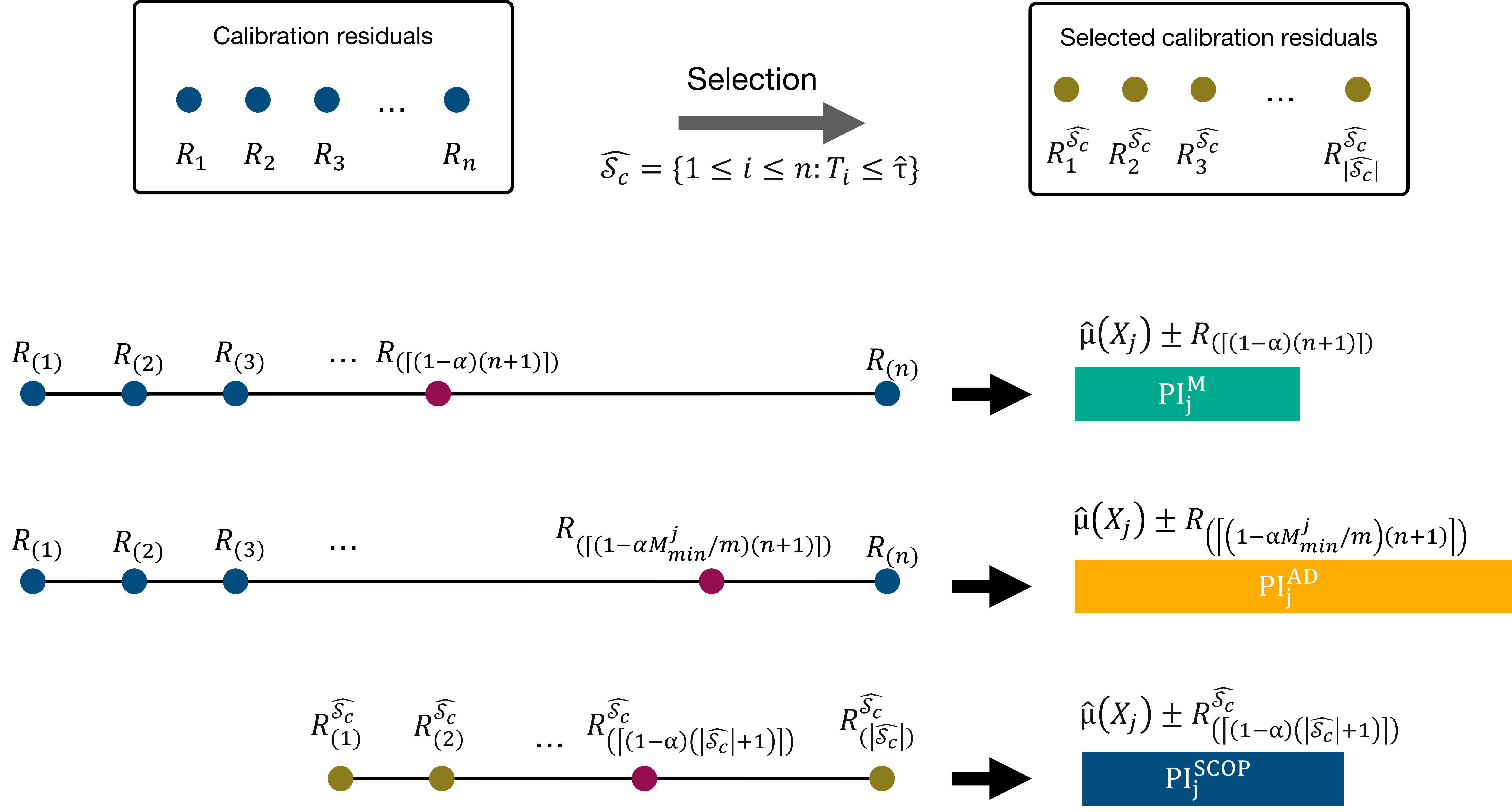}
\caption{The illustration of three different prediction intervals for the selected test point $j\in \hgS_u$ (that is $T_j \leq \hat{\tau}$): $\PI_j^{\rm{M}}$, $\PI_j^{\rm{AD}}$ and $\PI_j^{\rm{SCOP}}$.}
\label{fig:cartoon}
\end{figure}
\subsection{Illustration for the mentioned algorithms}
We illustrate the implementation of three different methods (our proposed method, selective conditional conformal prediction; ordinary conformal prediction; and adjusted conformal prediction) and enhance the visual representation of our work in Figure \ref{fig:cartoon}.

\subsection{Details on the $\FCR$-adjusted conformal prediction}
Recall The $\FCR$-adjusted conformal prediction intervals can be written as
\begin{align*}
    \PI_j^{\text{AD}} = \hat{\mu}(X_j) \pm Q_{\alpha_j}\LRs{\LRl{R_i}_{i\in \gC}},
\end{align*}
where $\alpha_j = \alpha\times  M_{\min}^j/m$ and \[
 M^j_{\min} = \min_{t}\LRl{|\hgS_u^{j\gets t}|:j\in \hgS_u^{j\gets t}}.
 \]
For many plausible selection rules such as fixed-threshold selection, the $ M_{\min}^j$ can be replaced by the cardinality of the selected subset $|\widehat{\gS}_{u}|$. 
In practice, for ease of computation, one may prefer to use this simplification, even though it does not have a theoretical guarantee for many data-dependent selection rules.

The $\FCR$-adjusted method is known to be quite conservative because it does not incorporate the selection event into the calculation \citep{weinstein2013selection}. Take the Top-K selection as an intuitive example. The selected set $\hat{\gS}_u$ is fixed with $|\hat{\gS}_u|=K$ and the $\FCR$ can be written as  
\begin{align*}
    \FCR = \frac{1}{K}\sum_{j\in\gU}\sP(j\in\hat{\gS}_u,Y_j \not \in \PI_j).
\end{align*}
Since the marginal $\PI^{\text{AD}}_j$ reaches the $1-\alpha K/m$ confidence level for any fixed $K$, the $\FCR$-adjusted method achieves the $\FCR$ control via 
\[
\FCR=\frac{1}{K}\sum_{j\in\gU}\sP\LRs{j\in\hat{\gS}_u, Y_j \not \in \PI^{\text{AD}}_j}\leq \frac{1}{K}\sum_{j\in\gU} \sP\LRs{Y_j \not \in \PI^{\text{AD}}_j}\leq\alpha,
\]
where the first inequality might be rather loose.

\subsection{Details on the selection with conformal $p$-values}
Let us first recall the definition of conformal $p$-values. We split the calibration set according to the null hypothesis and alternative hypothesis, that is $\gC = \gC_0 \cup \gC_1$. The conformal $p$-value is defined as
\begin{align*}
    p_j = \frac{1 + \sum_{i\in \gC_0} \Indicator{T_i \leq T_j}}{|\gC_0| + 1}.
\end{align*}
Therefore, the ranking threshold $\hkappa$ depends only on the test set $\gD_u$ and the null calibration set $\gD_{c_0}$, where $\gD_{c_0} = \{(X_i,Y_i):i \in \gC, Y_i \leq b_0\}$.
Next, we illustrate the application of Algorithm \ref{alg:SCOP} under a more general selection rule: step-up procedures \citet{lei2020}, which includes the Benjamini--Hochberg procedure. Let $0\leq \delta(1) \leq \cdots \leq \delta(m) \leq 1$ denote an increasing sequence of thresholds, we choose the ranking threshold for step-up procedures as
\begin{align}\label{eq:setp-up-procedure}
    \hkappa = \max\LRl{r:\ p_{(r)} \leq \delta(r)},
\end{align}
where $p_{(r)}$ is the $r$-th smallest conformal $p$-value. Specially, the Benjamini--Hochberg procedure takes $\delta(r)=r\beta/m$. We provide the detailed procedure of our proposed method under step-up procedures in Algorithm \ref{alg:step-up}.

\begin{algorithm}[tb]\fontsize{9.8pt}{11.76pt}\selectfont
	\renewcommand{\algorithmicrequire}{\textbf{Input:}}
	\renewcommand{\algorithmicensure}{\textbf{Output:}}
	\caption{Selective conditional conformal prediction under selection with conformal $p$-values}
	\label{alg:step-up}
	\begin{algorithmic}
		\REQUIRE Training data $\gD_t$, calibration data $\gD_c$, test data $\gD_u$, threshold sequence $\{\delta(r):r\in [m]\}$.
		
		\STATE \textbf{Step 1} Fit prediction model $\hmu(\cdot)$ and score function $g(\cdot)$ on $\gD_t$. Compute the score values $\ermT^{\gC} = \{T_i = g(X_i):i \in \gC\}$ and $\ermT^{\gU} = \{T_i = g(X_i):i \in \gU\}$.
		
		\STATE \textbf{Step 2} Compute the conformal $p$-values $\{p_i: i \in \gU\}$ according to \eqref{eq:conformal_p_values} based on $\gD_{\gC_0}$. Apply the Benjamini--Hochberg procedure with target level $\beta$ to $\ermT^{\gU}$ and obtain (\ref{eq:setp-up-procedure}).
		Obtain the post-selection subsets: $\hgS_u = \{i\in \gU: T_i \leq \ermT_{(\hkappa)}^{\gU}\}$ and $\hgS_c = \{i\in \gC: T_i \leq \ermT_{(\hkappa+1)}^{\gU}\}$.
		
     \STATE \textbf{Step 3} Compute residuals: $\ermR^{\hgS_c} = \{R_i = |Y_i - \hmu(X_i)|: i\in \hgS_c\}$. Find the $\lceil(1-\alpha)(|\hgS_c|+1)\rceil$-st smallest value of $\ermR^{\hgS_c}$, denoted by $\ermR^{\hgS_c}_{(\lceil (1 -\alpha)(|\hgS_c|+1)\rceil)}$.
	
   \STATE \textbf{Step 4} Construct $\PI_j$ for each $j \in \hgS_u$ as $\PI_j = \LRm{\hmu(X_j) - \ermR^{\hgS_c}_{(\lceil (1 -\alpha)(|\hgS_c|+1)\rceil)}, \hmu(X_j) + \ermR^{\hgS_c}_{(\lceil (1 -\alpha)(|\hgS_c|+1)\rceil)}}$.
		
        \ENSURE $\{\PI_j: j \in \hgS_u\}$.
	\end{algorithmic}  
\end{algorithm}

\textbf{Proposition \ref{pro:$p$-value_T-value} restated.}
Let $p_{(1)} \leq ... \leq p_{(m)}$ be order statistics of conformal $p$-values in the test set $\gU$. For any $i\in \gU$, it holds that $\{p_i \leq p_{(\hkappa)}\} = \{T_i \leq \ermT_{(\hkappa)}\}$.

\begin{proof}
We prove the conclusion under step-up procedures.
Notice that, there may exist ties in $\{p_j:j\in \gU\}$, but there is no tie in $\{T_j: j \in \gU\}$. Let $\rank(p_j) = \sum_{i\in \gU}\Indicator{p_i \leq p_j}$ and $\rank(T_j) = \sum_{i\in \gU}\Indicator{T_i \leq T_j}$ be the ranks of $p_j$ and $T_j$ in the test set, respectively. Then we have $\rank(p_j) \geq \rank(T_j)$, which means
\begin{align*}
    \LRl{p_j \leq p_{(\hkappa)}} = \LRl{\rank(p_j) \leq \hkappa} \subseteq \LRl{\rank(T_j) \leq \hkappa} = \LRl{T_j \leq \ermT_{(\hkappa)}^{\gU}}.
\end{align*}
By the definition of $\hkappa = \max\LRl{r : p_{(r)} \leq \delta(r)}$, we know
\begin{align*}
    \LRl{T_j \leq \ermT_{(\hkappa)}^{\gU}} = \LRl{T_j \leq \max\LRl{T_i: p_i = p_{(\hkappa)}}}  \Longrightarrow \LRl{p_j \leq p_{(\hkappa)}}.
\end{align*}
Therefore, we have proved that $\{p_j \leq p_{(\hkappa)}\} = \{T_j \leq \ermT_{(\hkappa)}\}$ for any $j \in \gU$. 
\end{proof}

\setcounter{equation}{0}
\def\theequation{F.\arabic{equation}}

\section{Additional numerical studies}

\subsection{Visualization of the selective conditional densities of $R_i$}
To visualize the effects of the selection procedure and the approximation via calibration set, we consider a linear model with heterogeneous noise, where we use ordinary least-squares for predictions and select $\hgS_u=\{j\in\gU:\hmu(X_j)\leq -1\}$. In Figure \ref{fig:compare}, we display the densities of $R_i$ for $i\in\gD_c$, $R_i$ for $i\in\hgS_c$ and the density of $R_j$ for $j\in\hgS_u$, respectively. The selection procedure significantly distorts the distribution of residuals, but the conditional uncertainty on $\hgS_u$ can be well approximated by that on $\hgS_c$.
\begin{figure}
\centering
\includegraphics[width = 0.65\linewidth]{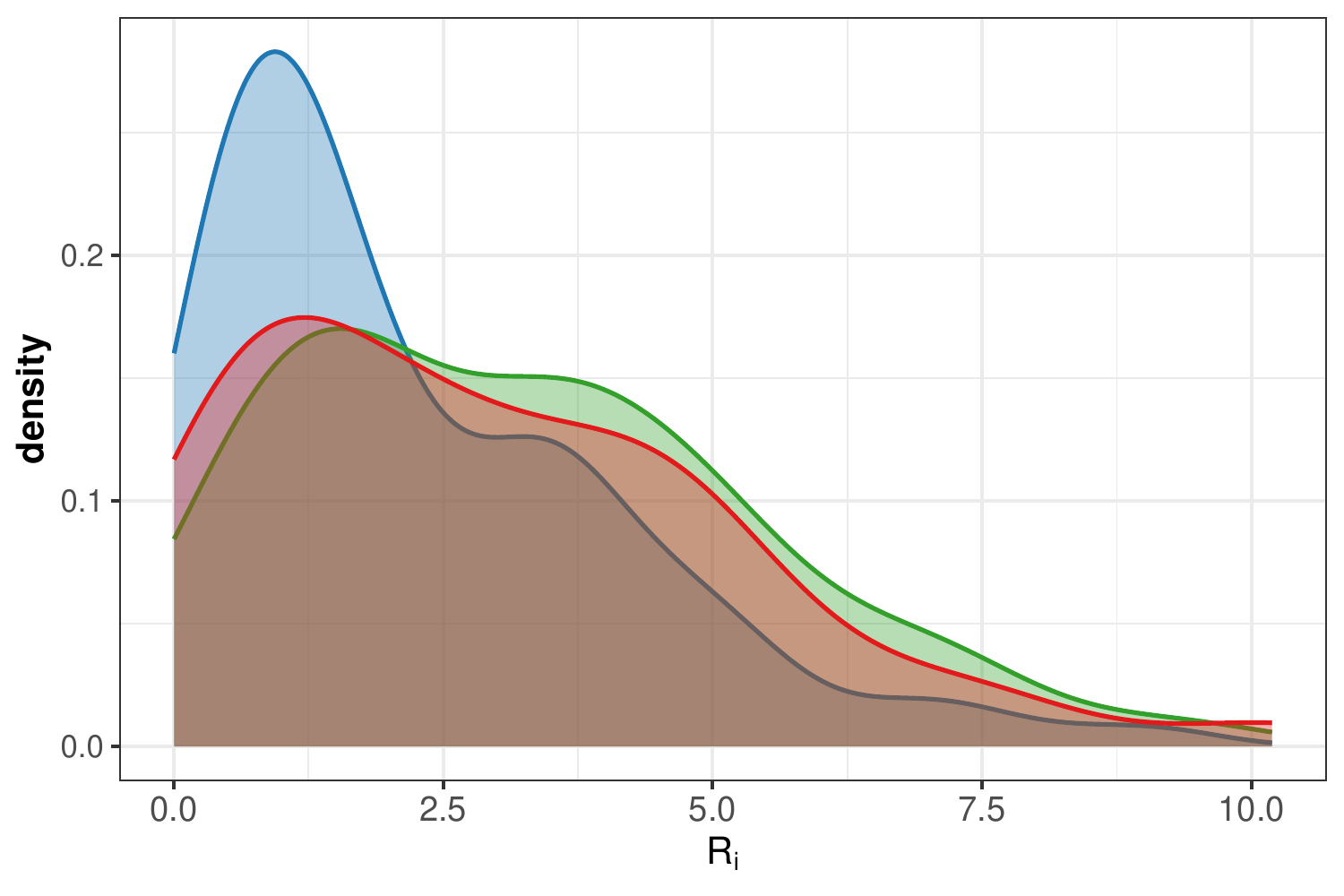}
\caption{The densities of $R_i$ for $i\in\gD_c$ (in blue), $R_i$ for $i\in\hgS_c$ (in green) and the density of $R_j$ for $j\in\hgS_u$ (in red), respectively. There are $2n=400$ labeled data and $m=200$ test data generated from a linear model with heterogeneous noise, that is the Scenario A. The selection rule is $\hgS=\{k:\hmu(X_k)\leq -1\}$.}
\label{fig:compare}
\end{figure}

\subsection{Visualization of the prediction interval}
We provide a plot of interval bands for better illustration. The simulation setting is similar to Scenario A except that we fix the data generating model as $Y=2X+\epsilon$ where $X\sim \unif([-1,1])$ is a one-dimensional uniform random variable and $\epsilon\mid X\sim N(0,(1+|2X|)^2)$. We would like to select a subset from test data by threshold {T-test($30\%$)}, i.e. 
$30\%$-quantile of predictions $\hmu(X)$ in test set. The sizes of training, calibration, and test data are fixed as $100$, $500$, and $500$, respectively. The target $\FCR$ level is $\alpha=10\%$. The true prediction interval for each selected $X_j$ is also considered. It can be constructed by $2X_j\pm z_{\alpha/2}(1+2|X_j|)$ where $z_{\alpha/2}$ is the $\alpha/2$ upper quantile of standard normal distribution.

Figure \ref{fig:PIplot} displays the selected samples and their prediction intervals. The brown area is the 90\%-level true prediction interval, where the false coverage proportion is $12\%$ and the average length is $7.97$. As for our method, it produces the red interval with length $8.55$ and maintains the false coverage proportion at $10\%$. We can see that our method approaches the oracle ones with the true prediction intervals in terms of FCR control, and meanwhile our method yields satisfactory prediction intervals whose lengths are slightly greater than the oracle ones. In addition, we depicted the prediction intervals of two benchmarks. The green interval is from the ordinary conformal prediction. Note that it is almost always inside the region of the true prediction interval but yields a much inflated false coverage proportion, i.e., $21\%$. The blue interval, which is the FCR-adjusted method, is the widest and always outside of the true prediction interval due to constructing larger $(1-\alpha |\hgS_u|/m)$ prediction intervals.

\begin{figure}
\centering
\includegraphics[width=0.85\textwidth]{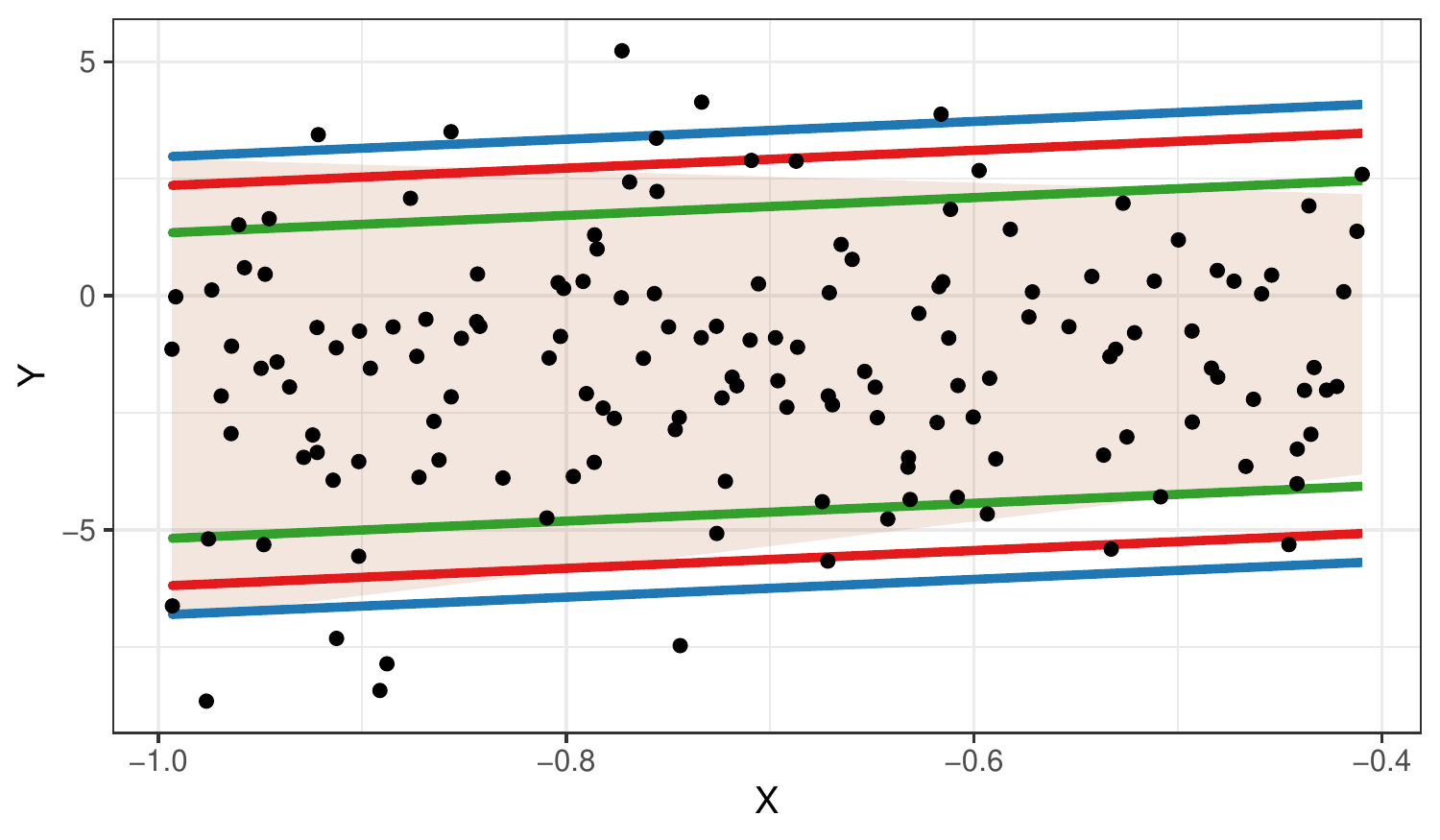}
\caption{The illustration plot of prediction interval for selected individuals. Red: selection conditional conformal prediction. Green: ordinary conformal prediction. Blue: $\FCR$-adjusted conformal prediction. And the brown area is the true prediction interval.}
\label{fig:PIplot}
\end{figure}

\subsection{Comparison with \citet{weinstein2020online}'s method}

The method in \citet{weinstein2020online} is designed for online control of $\FCR$, which would exhibit conservative performances in the offline setting. At each time point $t$, they construct $1-\alpha_t$ marginal confidence interval for some $\alpha_t<\alpha$. And the choice of $\alpha_t$ should keep the estimated false coverage proportion controlled at $\alpha$ at each time $t$, i.e.
    $$\widehat{\operatorname{FCP} }(t):=\frac{\sum_{i=1}^t\alpha_i}{1\vee \sum_{i=1}^tS_i}\leq\alpha.$$
    
   In the implementation of the online $\FCR$ control method, $\alpha_t$ is considerably small, leading to a confidence interval of inflated length. In Figure \ref{fig:quantileLORD}, we compared our method to the one in \citet{weinstein2020online} when three quantile-based thresholds i.e. {T-cal($q$)}, {T-test($q$)} and {T-exch($q$)} are considered under both scenarios with varying $q$. The online $\FCR$ control method in \citet{weinstein2020online} is implemented with its defaulted parameters by adding the offline unlabeled data one by one and then computing the $\FCR$ value when the ``online procedure'' stops.
As expected, this online method is not as efficient as our proposed method in the sense that it yields extremely conservative $\FCR$ values as well as inflated intervals. Thus, we focused on the comparisons between our method and other offline methods in the main text and put this figure in the Supplementary Material. 

\begin{figure}
\centering
\includegraphics[width=\textwidth]{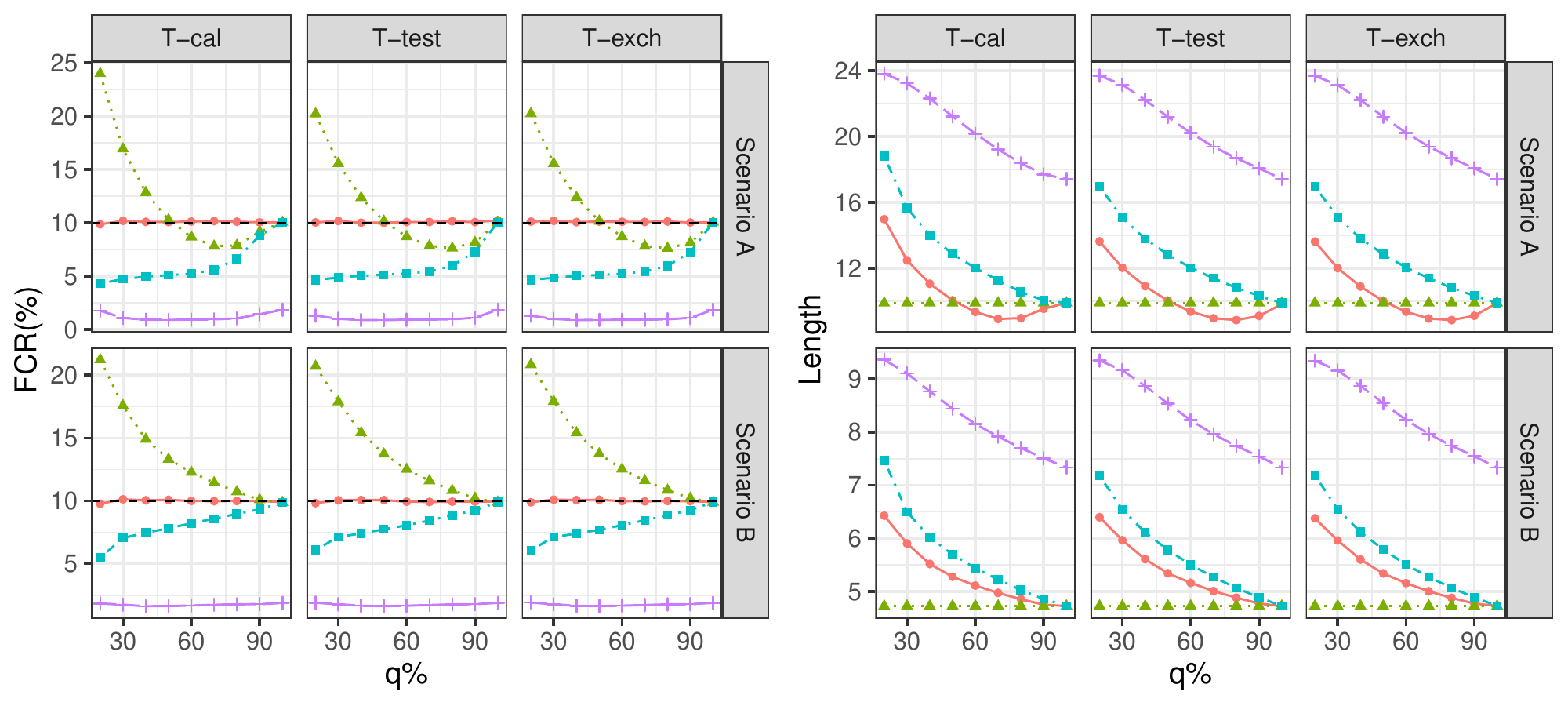}
\caption{Empirical $\FCR$ (\%) and average length of $\PI_j$ for quantile based thresholds with varying quantile level $q\%$  of three methods as follows: selection conditional conformal prediction (circle dot solid red line); ordinary conformal prediction (triangle dashed green line); $\FCR$-adjusted conformal prediction (square dot-dash blue line); \citet{weinstein2020online}'s online $\FCR$ method (plus sign purple line). The black dashed line represents the target $\FCR$ level $10\%$.}
\label{fig:quantileLORD}
\end{figure}

\subsection{Extension to conformalized quantile regression}
In fact, our framework can readily be extended to more general nonconformity scores such as the one based on quantile regression \citep{romano2019conformalized} or distributional regression \citep{victor2021distributional} to adapt to heteroscedastic cases. Similar theoretical results are still valid for those more general methods. We make an additional simulation based on Scenario A using conformalized quantile regression \citep{romano2019conformalized} to illustrate the flexibility of our method. 

Conformalized quantile regression firstly fits two conditional quantile regressions at level $\alpha_{\operatorname{lo}}$ and $\alpha_{\operatorname{hi}}$ on the training set, where $\alpha_{\operatorname{lo}}=\alpha/2$ and $\alpha_{\operatorname{hi}}=1-\alpha/2$. Denote these two estimated regression functions as $\hat{q}_{\alpha_{\operatorname{lo}}}(\cdot)$ and $\hat{q}_{\alpha_{\operatorname{hi}}}(\cdot)$. Then the conformalized quantile regression computes a new nonconformity score $R^{\operatorname{CQR}}_i$ based on the calibration set:
$$R^{\operatorname{CQR}}_i=\max\{\hat{q}_{\alpha_{\operatorname{lo}}}(X_i)-Y_i,Y_i-\hat{q}_{\alpha_{\operatorname{hi}}}(X_i)\},i\in\gC.$$
We can replace the residual nonconformity score with this quantile regression based score, and the corresponding procedure is just the same as those for both $\FCR$-adjusted method and our proposed method given in Algorithm \ref{alg:SCOP}. In the additional simulation, we use linear quantile regression to fit the model, and the results are listed in Table \ref{table:cqr}. The adopted thresholds are 
{T-con}, {T-pos} and {T-top} as previously. 
These results demonstrate that our method can be validly combined with quantile regression nonconformity score to adapt to the heteroscedastic cases.
 \begin{table}[htbp]
\centering
\caption{Comparisons of empirical $\FCR$ (\%) and average length under Scenario A for quantile regression score. The sample sizes of the calibration set and the test set are fixed as $n=m=200$. The $b_0$ is fixed as $-1$ and the target $\FCR$ $\alpha=10\%$. SCOP: the proposed method; OCP: ordinary conformal interval; ACP: $\FCR$-adjusted conformal interval.}\label{table:cqr}
\resizebox{\textwidth}{!}{
\begin{tabular}{ccccccccccccc}
    && \multicolumn{3}{c}{{T-con($b_0$)}} && \multicolumn{3}{c}{{T-pos($b_0,30\%$)}} && \multicolumn{3}{c}{{T-top($60$)}}   \\ 
    && SCOP    & OCP    & ACP    && SCOP    & OCP    & ACP    && SCOP    & OCP    & ACP    \\                                 
{Scenario A}& $\FCR$ & 9.58    & 13.96  & 4.65   && 8.10    & 15.91  & 3.71 && 9.51 & 14.49 & 3.71    \\
{(quantile regression score)}&Length  & 13.05    & 11.33   & 16.25   && 15.14& 11.33   & 20.37  && 13.25 & 11.33 & 16.49    \\ 
\end{tabular}}
\end{table}

\subsection{Additional thresholds and scenarios }\label{subsec:add_thresh}
In geospatial analysis, researchers require to classify different regions into two groups by their values of average annual rainfall \citep{de2007geospatial}. The method is to perform clustering on the values of rainfall and find the boundary value of two groups, which is referred as to a ``natural break'' in scientific literature \citep{baah2015risk}. When a region has larger rainfall than the boundary value, it would be regarded as a high rainfall area. Therefore, If we want to select those high rainfall areas and are only available to the rainfall predictions obtained by machine learning algorithms, a natural choice is to cluster the rainfall predictions and find the boundary value as a selection threshold. Hence we consider an additional threshold which is based on this clustering method, denoted as {T-clu}.

The clustering method used in the natural break analysis is the Fisher optimal division \citep{fisher1958grouping}. In our implementation, we put the calibration set and test set together for clustering, which ensures the exchangeability of the threshold. Denote the threshold as $\hat{\tau}$. There are two groups determined by $\hat{\tau}$, i.e. $\gS_1(\hat{\tau})=\{i\in\gC\cup\gU:T_i\leq \hat{\tau}\}$ and $\gS_1^c(\hat{\tau})=\gC\cup\gU/\gS_1(\hat{\tau})$. We want to choose $\hat{\tau}$ such that the within-group discrepancy is minimized, i.e.
$$\htau=\mathop{\arg\min}\limits_{t\in \ermT^{\gC\cup\gU}} \sum_{i\in\gS_1(t)}(T_i-\Bar{\ermT}^{\gS_1(t)})^2+\sum_{k\in\gS_1^c(t)}(T_k-\Bar{\ermT}^{\gS_1^c(t)})^2,$$
 where $\Bar{\ermT}^{\gS_1(t)}$ and $\Bar{\ermT}^{\gS_2(t)}$ are the sample mean in $\gS_1(t)$ and $\gS_2(t)$, respectively.

Further, we consider an additional scenario to match the clustering threshold. {Scenario C} is an aggregation model formed by $\mu(X)=4\{X^{(1)}+1\}|X^{(3)}|\Indicator{{X^{(2)}>-0.4}}+4\{X^{(1)}-1\}\Indicator{X^{(2)}\leq-0.4}. $ The noise is $\epsilon\sim N(0,1)$ and independent of $X$. We use random forest to give predictions for Scenario C, where the algorithm is implemented by \texttt{R} package \texttt{randomForest}  with default parameters. The simulation results under Scenario C are shown in Table \ref{table:n-m-combination}. And the additional threshold {T-clu} is examined via our experiments in Table \ref{table:n-m-combination} and Table \ref{table:houseprice}.

\subsection{Additional Simulation results}
Table \ref{table:diff_algo} displays the results for both regression algorithms, ordinary least square and support vector machine, on the mentioned scenarios. It verifies that our method is model-agnostic and can wrap around various learning algorithms.

    \begin{table}[htbp]
\centering
\caption{Comparisons of empirical $\FCR$ (\%) and average length under different scenarios and thresholds with target $\FCR$ $\alpha=10\%$. The sample sizes of the calibration set and the test set are fixed as $n=m=200$. OLS: ordinary least square; SVM: support vector machine.}\label{table:diff_algo}
\resizebox{\textwidth}{!}{
\begin{tabular}{lcccccccccccc}
    && \multicolumn{3}{c}{{T-con($b_0$)}} && \multicolumn{3}{c}{{T-pos($b_0,20\%$)}} && \multicolumn{3}{c}{{T-top($60$)}}   \\ 
    && SCOP    & OCP    & ACP    && SCOP    & OCP    & ACP    && SCOP    & OCP    & ACP    \\ 
\specialrule{0em}{2pt}{2pt}
\multirow{2}{*}{{Scenario A (OLS)}}& $\FCR$ & 9.76   & 14.67  & 4.91   && 5.45    & 13.19& 2.17    && 9.73 & 15.26 & 4.90   \\
&Length  & 11.83   & 9.91   & 14.87  && 16.06   & 9.91   & 22.57 && 12.09 & 9.91 & 15.10   \\    
\specialrule{0em}{2pt}{2pt}
\multirow{2}{*}{{Scenario A (SVM)}}& $\FCR$ & 9.80   & 14.87  & 4.15   && 5.08    & 9.86  & 1.20    && 9.87 & 13.69 & 4.37   \\
&Length  & 12.40   & 10.20   & 16.32  && 15.58   & 10.20   & 24.04 && 11.90 & 10.20 & 15.62   \\    
\specialrule{0em}{2pt}{2pt}
\multirow{2}{*}{{Scenario B (OLS)}}& $\FCR$ & 9.99    & 12.25  & 5.05   && 9.91    & 12.13  & 5.12 && 9.90 & 13.01 & 4.67    \\
&Length  & 5.03    & 4.67   & 5.94   && 5.01    & 4.67   & 5.91  && 5.13 & 4.67 & 6.15    \\   
\specialrule{0em}{2pt}{2pt}
\multirow{2}{*}{{Scenario B (SVM)}}& $\FCR$ & 9.84    & 16.99  & 7.04   && 9.93    & 15.83  & 7.09 && 9.80 & 17.71& 6.98    \\
&Length  & 5.87    & 4.72   & 6.41   && 5.68    & 4.72   & 6.23  && 5.95 & 4.72 & 6.54    \\ 
\end{tabular}}
\end{table}

Table \ref{table:non-normal} shows the results under Scenario B when the noise term follows 
$t$-distribution with two degrees of freedom, i.e. $\epsilon\sim t(2)$. We can see the proposed method still outperforms others in the fully non-normal situation, which further verifies the distribution-free property of our method.

\begin{table}[htbp]
\centering
\caption{Comparisons of empirical $\FCR$ (\%) and average length under Scenario B with non-normal error. The sample sizes of the calibration set and the test set are fixed as $n=m=200$. The $b_0$ is fixed as $-8$ and the target $\FCR$ $\alpha=10\%$.}\label{table:non-normal}
\resizebox{\textwidth}{!}{
\begin{tabular}{ccccccccccccc}
    && \multicolumn{3}{c}{{T-con($b_0$)}} && \multicolumn{3}{c}{{T-pos($b_0,20\%$)}} && \multicolumn{3}{c}{{T-top($60$)}}   \\ 
    && SCOP    & OCP    & ACP    && SCOP    & OCP    & ACP    && SCOP    & OCP    & ACP    \\                                 
{Scenario B}& $\FCR$ & 9.78    & 12.73  & 3.83   && 9.81    & 12.57  & 3.83 && 9.75 & 12.91 & 3.50    \\
{($\epsilon\sim t(2)$)}&Length  & 8.14    & 7.17   & 12.18   && 8.09    & 7.17   & 12.06  && 8.23 & 7.17 & 12.72    \\ 
\end{tabular}}
\end{table}

Table \ref{table:compar_algorithm+} compares the performances of Algorithm \ref{alg:SCOP} and Algorithm \ref{alg:SCOP}+ under the same setting as Table \ref{table:other_thresholds} in the main text.  It can be shown that Algorithm \ref{alg:SCOP} and Algorithm \ref{alg:SCOP}+ yield very close empirical results.

\begin{table}[htbp]
\centering
\caption{Comparisons of empirical $\FCR$ (\%) and average length with target $\FCR$ $\alpha=10\%$ for Algorithm \ref{alg:SCOP} and Algorithm \ref{alg:SCOP}+. The sample sizes of the calibration set and the test set are fixed as $n=m=200$}\label{table:compar_algorithm+}
\resizebox{\textwidth}{!}{
\begin{tabular}{cccccccccc}
 & &\multicolumn{2}{c}{{T-con($b_0$)}} && \multicolumn{2}{c}{{T-pos($b_0,20\%$)}} && \multicolumn{2}{c}{{T-top($60$)}}   \\ 
 & &  Algorithm \ref{alg:SCOP}  &Algorithm \ref{alg:SCOP}+      && Algorithm \ref{alg:SCOP} &Algorithm \ref{alg:SCOP}+      && Algorithm \ref{alg:SCOP}  &Algorithm \ref{alg:SCOP}+      \\ 
  
 \multirow{2}{*}{{Scenario A}}&$\FCR$ & 9.759   & 9.874    && 6.966 &5.449      && 9.725 & 9.917    \\
&Length  & 11.829   & 11.772   && 15.479 &16.062  && 12.115 &12.091    \\    
\specialrule{0em}{2pt}{2pt}
\multirow{2}{*}{{Scenario B}}& $\FCR$ & 9.843   & 9.973    && 9.819 &9.931      && 9.667 & 9.890    \\
&Length  & 5.866   & 5.845   && 5.737 &5.683  && 5.959 &5.954    \\    

\end{tabular}}
\end{table}

Table \ref{table:n-m-combination} reports the results with different sizes of calibration and test sets under Scenario C in Section \ref{subsec:add_thresh} by varying $n$ and $m$ from 100 to 200. Here four selection thresholds are included: {T-test($30\%$)}, {T-pos($-1.2$,$20\%$)}, {T-top($60$)} and {T-clu}.  We see that all three methods tend to yield narrowed prediction intervals as the calibration size $n$ increases. However, our method in Algorithm \ref{alg:SCOP} performs much better than ordinary conformal prediction and $\FCR$-adjusted conformal prediction in terms of $\FCR$ control across all the settings. This clearly demonstrates the efficiency of our proposed method, which is a data-driven method that enables $\FCR$ control with a wide range of selection procedures and meanwhile tends to build a relatively narrowed $\PI_j$ with $1-\alpha$ level.

\begin{table}[tb]
\caption{Empirical $\FCR$ (\%) and average length values under Scenario C with different combinations of $(n,m)$. The target $\FCR$ level is $\alpha=10\%$}\label{table:n-m-combination}
\resizebox{\textwidth}{!}{
\begin{tabular}{ccccccccccccccccc}
                       ($n$,$m$)    &     & \multicolumn{3}{c}{{T-test($30\%$)}} && \multicolumn{3}{c}{{T-pos($-1.2$,$20\%$)}} && \multicolumn{3}{c}{{T-top($60$)}} && \multicolumn{3}{c}{{T-clu}} \\ 
                           &     & SCOP    & OCP    & ACP    && SCOP    & OCP    & ACP    && SCOP    & OCP    & ACP    && SCOP   & OCP   & ACP  \\ 
                                                                 
\multirow{2}{*}{(100,100)} & $\FCR$ & 9.94   & 14.60 & 4.53   &&  10.16  & 14.83  & 4.49   &&9.88     &  8.94     &  5.40   &&  9.74   & 15.20  &  4.33   \\
                           & Length  &  7.17&  6.26&  8.64&&  7.13&  6.26&  8.72&&   6.10&   6.26&   7.32  &&  7.28&  6.26&  8.80   \\ 
\specialrule{0em}{2pt}{2pt}
 \multirow{2}{*}{(200,100)} & $\FCR$ & 10.08& 12.70& 3.89&& 10.34 &12.48 &3.75  &&  10.10&  7.92  &4.63  && 10.16& 13.16 &3.74   \\
                           & Length  &   5.76& 5.29& 7.29&& 5.69& 5.29 &7.28 && 4.92&  5.29 & 6.17&& 5.86& 5.29 &7.42      \\     
\specialrule{0em}{2pt}{2pt}
\multirow{2}{*}{(100,200)} & $\FCR$ & 9.84 &14.53& 4.59 &&9.91 &14.57 &4.41  && 9.84 &14.53  &4.59 &&9.81& 14.92 &4.40     \\
                           & Length  &   7.22 & 6.27 &8.67&& 7.23 & 6.27 &8.73   &&7.22  & 6.27  &8.67   &&7.28  &6.27 &8.79  \\ 
\specialrule{0em}{2pt}{2pt}
\multirow{2}{*}{(200,200)} & $\FCR$ & 9.81 &12.48 &3.78 &&9.96 &12.30& 3.86   &&9.81 &12.48  &3.78&&9.82 &12.93 &3.70      \\
                           & Length  & 5.79 &5.30 &7.32 &&5.73 &5.30 &7.27  &&5.79  &5.60  &7.32  && 5.85 &5.30  &7.41    \\                     
\end{tabular}}
\end{table}

\section{Additional real data applications on house price analysis}\label{appen:real_data}
We apply our method to implement the prediction of house prices of interest. As an economic indicator, a better understanding of house prices can provide meaningful suggestions to researchers and decision-makers in the real estate market \citep{anundsen2013self}. In recent decades, business analysts use machine learning tools to forecast house prices and determine the investment strategy \citep{park2015using}. In this example, we use our method to build prediction intervals for those house prices which exceed certain thresholds.

We consider one house price prediction dataset from Kaggle, which contains $4,251$ observations after removing the missing data (\texttt{https://www.kaggle.com/datasets/shree1992/housedata}). The data records the house price 
and other covariates about house area, location, building years and so on. We randomly sample $1,500$ observations and equally split them into three parts as training, calibration and test sets respectively in each repetition. We firstly train a random forest model to predict the house prices and consider three different thresholds to select those test observations with high predicted house prices: {T-test($70\%$)}, {T-pos($0.6$,$20\%$)} and {T-clu}. For example, the threshold {T-pos($0.6$,$20\%$)} means that one would like to select those observations with house prices larger than $0.6$ million under the $\FDR$ control at 20\% level.  After selection, we construct prediction intervals with $\alpha=10\%$. Table \ref{table:houseprice} reports the empirical $\FCR$ level and lengths of $\PI_j$ among $500$ replications. We observe that both our method and $\FCR$-adjusted conformal prediction achieve valid $\FCR$ control, but our method has more narrowed prediction intervals compared to $\FCR$-adjusted conformal prediction. The $\FCR$ values of the ordinary conformal prediction are much inflated which implies that many test samples with truly high house prices cannot be covered. It demonstrates that our proposed method works well for building prediction intervals of selected samples in practical applications.

\begin{table}[tb]
\centering
\caption{Empirical $\FCR$ (\%) and average length of the prediction intervals for house price dataset with $\alpha=10\%$}\label{table:houseprice}
\begin{tabular*}\hsize{@{}@{\extracolsep{\fill}}cccccccccccc@{}}
                     & \multicolumn{3}{c}{{T-test($70\%$)}}         & &\multicolumn{3}{c}{{T-pos($0.6$,$20\%$)}}          && \multicolumn{3}{c}{{T-clu}}          \\ 
    & SCOP    & OCP    & ACP    && SCOP    & OCP    & ACP    && SCOP    & OCP    & ACP    \\ 
FCR  &  8.91 &  24.80 &  7.56 && 8.02 & 33.58 & 4.87 && 8.57 &  38.49 & 6.41   \\
Length  & 1.23 & 0.67&  1.31 &&  1.58& 0.67 &3.13 &&1.69& 0.67 &1.87 \\ 
\end{tabular*}
\end{table}

\end{document}